\documentclass[a4paper,12pt]{article}
\usepackage{verbatim}
\usepackage[utf8]{inputenc}
\usepackage{amssymb, amsmath}
\usepackage{amsthm}
\usepackage{amsbsy}
\usepackage{amsfonts}
\usepackage[colorlinks]{hyperref}
\usepackage{mathtools}
\usepackage{bookmark}

\usepackage{dsfont}

\usepackage[sort,nocompress]{cite}

\usepackage{array,esint}

\newcommand{\R}{{\mathord{\mathbb R}}}
\newcommand{\Z}{{\mathord{\mathbb Z}}}
\newcommand{\N}{{\mathord{\mathbb N}}}
\newcommand{\C}{{\mathord{\mathbb C}}}

\newcommand{\HH}{\mathcal{H}}
\newcommand{\hh}{\mathfrak{h}}
\newcommand{\FF}{\mathcal{F}}
\newcommand{\LL}{\mathcal{L}}
\newcommand{\WW}{\mathcal{W}}

\newcommand{\UU}{\mathcal{U}}

\newcommand{\RR}{\mathcal{R}}

\newcommand{\BB}{\mathcal{B}}
\newcommand{\ran}{{\rm Ran}}
\renewcommand{\ker}{{\rm Ker}}

\def\Eat{E_{\rm at}}
\def\ph  {\varphi}

\def\one{\mathds{1}}

\def\inf{{\rm inf}\,}

\theoremstyle{plain}
\newtheorem{lemma}{Lemma}[section]
\newtheorem{theorem}[lemma]{Theorem}
\newtheorem{proposition}[lemma]{Proposition}

\newtheorem{corollary}[lemma]{Corollary}

\newtheorem{definition}[lemma]{Definition}

\newtheorem{remark}[lemma]{Remark}

\newtheorem{hypr}{\bf Hypothesis}

\newtheorem*{remark*}{Remark}

\newcommand{\Hel}[1]{H_{\rm at}(#1)}

\numberwithin{equation}{section}

\def\chib {\overline{\chi}}

\newcommand{\inn}[1]{\langle {#1} \rangle }

\newcommand{\degendim}{\textrm{d}}
\newcommand{\Patsdimension}{\textrm{$d$}}

\begin{document}

\title{Degenerate  perturbation theory for   models of quantum field theory with symmetries}
\author{\vspace{5pt} David Hasler$^1$\footnote{
E-mail: david.hasler@uni-jena.de} and Markus Lange$^2$\footnote{E-mail: markus.lange@dlr.de}
\\
\vspace{-4pt} \small{$1.$ Department of Mathematics,
Friedrich Schiller  University  Jena} \\ \small{Jena, Germany }\\
\vspace{-4pt}
\small{$2.$ German Aerospace Center (DLR), Institute for AI-Safety and Security} \\
\small{Sankt Augustin \& Ulm, Germany}\\
}

\date{\today}
\maketitle

\begin{abstract}
We consider Hamiltonians of   models  describing  non-relativistic quantum mechanical matter coupled to 
a relativistic field of bosons.  If the free Hamiltonian  has an eigenvalue, we show that this 
eigenvalue persists also for nonzero coupling.  The eigenvalue of the free Hamiltonian may 
be degenerate  provided there exists a symmetry group acting irreducibly on the eigenspace.
Furthermore, if the Hamiltonian depends analytically on external parameters then so does 
 the eigenvalue and eigenvector. Our result applies to the ground state as well as resonance 
states. For our results we assume a mild infrared condition. The proof 
is based on operator theoretic renormalization. It generalizes the method used in \cite{GriHas09}
to  non-degenerate situations,  where the degeneracy is protected by a symmetry group, 
and utilizes  Schur's lemma from  representation theory. 
\end{abstract}

\section{Introduction} 

%******WHATS THE GENERAL TYPE OF MODEL

We consider    mathematical models describing non-relativistic   quantum mechanical
matter interacting with a quantized field consisting of infinitely many  bosons. 
%with a massless relativistic dispersion relation.    
Such models are 
used to describe atoms or molecules  interacting with the surrounding  electromagnetic field or particles  
in solids interacting with lattice excitation, so called phonons.

%*************** MASSLESS CASE 

In this paper we will focus on models describing interaction with the  electromagnetic field. In that case 
the bosons are    photons and   have a massless relativistic dispersion relation but the electrons and nuclei
are treated as non-relativistic quantum mechanical particles. Such type of models are often referred to as non-relativistic qed.

%*************** MATHEMATICAL DESCRIPTION

The dynamics as well as the energy of  these models is  determined by a self-adjoint operator called 
the Hamiltonian. For these models the Hamiltonian is   typically  bounded from below and 
the infimum of its  spectrum is called ground state energy.
If the ground state energy is an eigenvalue the corresponding eigenvector is called 
ground state.    As a consequence of the  massless nature of  photons    the ground state 
energy is   not isolated from the rest of the spectrum of the Hamiltonian. The question of existence 
of   a ground  state is nontrivial.  It  has been shown that for models of  non-relativistic 
qed
  a ground state exists \cite{Spo98,BacFroSig99,Ger00,GriLieLos01,LieLos03} under natural assumptions.

%**** WAHT DO WE STUDY WHY IS IT  USEFUL

 In this paper we consider models for which 
the  existence of  a ground state has been established. We address
the question, how the    ground state  as well as  the ground state  energy,  $E$, 
depend on
 parameters of the system. For example one  is  interested on its dependence on the coupling constant,  on  the positions of static  nuclei for  
molecules,  or on    analytic extensions of  dilations  and   translations. 
The regularity of $E$ as a
function of such  parameters is of fundamental importance for Born-Oppenheimer approximation, scattering theory, adiabatic theory,  cf.  \cite{GriHas09}.

%****** WHY IS PROBLEM HARD

 If $E$ were an isolated eigenvalue, like it
is in quantum mechanical description of molecules without
radiation, then analyticity of $E$ with respect to any of the
aforementioned parameters would follow from regular perturbation
theory. But in models of qed  describing photons  the energy $E$ is not isolated and the analysis
of its regularity is a difficult mathematical problem.

%***** WHAT HAS BEEN DONE

The aforementioned  question has been adressed in \cite{GriHas09}.
In that paper, it was shown that  if the Hamiltonian of the model depends analytically on  some parameter, $s$, then also
the ground state as well as  $E$  depend analytically on $s$. For the proof of the result in \cite{GriHas09}  a mild infrared regularization was needed. In the special case of the classical spin-boson model analyticity of the ground state and the ground state energy as a function 
 of the coupling 
constant could be established  without the necessity of an infrared regularization \cite{HasHer11-1}. 
Analyticity of  ground states and ground state  energies   
as  a function of the coupling parameter has  been shown in \cite{HasHer11-2} for atoms  in the framework   of non-relativistic qed.
For models of non-relativistic qed and the spin boson model analytic extensions of   dilations  have  been  studied   in connection with 
resonances   \cite{BacFroSig98-1,BacFroSig98-2,BalDecHan.2019}.

Furthermore,  we  want to mention   related results about translation invariant models of quantum field theory, where  
the Hamiltonian commutes with the generators of translations. In such a situation one can restrict the 
Hamiltonian to the generalized eigenspaces corresponding to the eigenvalues $p \in \R^3$  of the generators of translations. 
 This restriction, $H(p)$, is called  fiber Hamiltonian. 
Motivated by the construction of 
scattering states, regularity of the infimum of the spectrum  for  these fiber Hamiltonians $H(p)$  
as a function of $p$  has  been intensively investigated for various models  \cite{Fro73,Che08,BacCheFroSig07,CheFroPiz09,AbdHas12} with results ranging from H\"older continuity  up to real analyticity.

%********************** DEGENERATE 

A common assumption of the aforementioned analyticity results  in  \cite{BacFroSig98-1,BacFroSig98-2,GriHas09,HasHer11-1,HasHer11-2,BalDecHan.2019} is that 
the ground state energy of the  Hamiltonian  describing the massive non-relativistic matter
is   non-degenerate. 
However, in many  situations this assumption is not met. 
For example for  almost all  atoms, except the noble atoms, the valence shell is  not fully
occupied  and therefore  by common physical folklore the ground state energy is degenerate by   
rotation symmetry (we have not found a rigorous proof of this fact but there is almost certain physical evidence corroborating it). Even for molecules, where rotation invariance is broken,   degeneracy may occur  by 
 the spinorial degrees of 
freedom.   

%************************PHYSICAL PICTURE 

If  an eigenvalue  of the   Hamiltonian describing the non-relativistic quantum mechanical matter 
is degenerate, the coupling to the quantized field can  lift the degeneracy.  It may be lifted  completely or  there might remain 
some degeneracy of  possibly smaller multiplicity.

%************************DEGENERATE WHAT HAS BEEN DONE

The lifting of the  degeneracy of an eigenvalue of an  atomic Hamiltonian due to the coupling 
of the electromagnetic field is usually referred to as the Lamb shift. The most prominent example is 
the spliting of the first excited  energy level in the hydrogen atom \cite{LamRet47}.   For a mathematical  discussion of 
  such a  phenomenon in the framework of non-relativistic qed, see for example \cite{AmoFau13} and references therein.
%coupling to the field lifts the degeneracy of the ground state energy 
%of an atom have been studied  note that splitting of eigenvalues for models of qft   has been studied
%intensively    in the physics literature in particular  
%in connection with the so called Lambshift [ ]. For a mathematical treatment of resonances in this context see  [Faupin,...]. 
The Lamb shift was studied in  \cite{HasLan18-1} in a situation where the degeneracy of the ground state energy is lifted 
at second order formal perturbation theory. It was shown under a mild infrared condition
 that the ground state as well as the ground state energy
are analytic functions of the coupling constant  in a sectorial region around the origin.  
This is in contrast to perturbation theory of isolated eigenvalues, where by general principles
analyticity holds  on a  whole   ball  around  the origin, cf. \cite{ReeSim4} and references therein.

%T. CHEN 

In \cite{Che08}   the ground state energy  of the fiber Hamiltonian $H(p)$ for an electron with spin interacting with  the quantized electromagnetic field was  studied and 
its regularity properties  as a function 
of $p$ in a neighborhood of zero were investigated. In this case,    the coupling to the quantized electromagnetic field does not lift the spin degeneracy, which can be seen using time reversal symmetry and Kramer's degeneracy theorem \cite{Wigner1932}.

%*********WHAT DO WE PROVE**************************************

 In this paper we consider the situation where the so called atomic    Hamiltonian, 
describing the non-relativistic matter,  has a discrete eigenvalue. This eigenvalue  may be  degenerate,  but we assume that there 
exists  an underlying symmetry of the  full   Hamiltonian, which  acts irreducibly on the corresponding eigenspace. 
In that case the interaction does not lift nor decrease    the degeneracy, which turns out to be   protected by the symmetry.  
In particular, we show the existence of an eigenvalue for  small but nonzero coupling.  
Moreover, the main result states that  if  the Hamiltonian depends analytically on a parameter $s$, then also the  eigenvalue  as well as the eigenstate 
depend  analytically on $s$.

The result is formulated  analogously to  the main result in \cite{GriHas09}. We 
  generalize the main  result in that paper  to degenerate situations, i.e.,  
we relax  the non-degeneracy condition % in \cite{GriHas09} 
 to   an irreducibility
condition with respect to a symmetry group.  Furthermore, we generalize the result in \cite{GriHas09}
to include general eigenvalues, which  may be  different from  the ground state energy. This allows the 
treatment of resonance  states, by which we understand 
eigenvectors  of an analytically dilated Hamiltonian. 

As in  \cite{GriHas09} we assume that the interaction is linear in the field operator 
of the quantized field and that there is a mild infrared regularization.  
In fact, the  main part of the proof   also applies 
to situations arising for the standard model of non-relativstic qed, which is quadratic in the field operators. We  isolate the  part of the  proof which applies to general situations
 as a corollary of the proof      in separate theorem
within the last section.

%**************************************
%
%
%
%
%  
%
%
%*********HOW  DO WE PROVE IT**************************************

The proof of the main  result is based on operator theoretic renormalization \cite{BacFroSig98-2}.
This method is based on an iterated application of the Schur complement also
called Feshbach map. 
One can show that this procedure leads to a fix point, provided infrared behaviour of the original operator 
is not to singular.  Using this fixed point one can construct  the ground 
as the limit of a convergent sequence.  If the original Hamiltonian is analytic
one can show, as in  \cite{GriHas09}, that this approximating sequence is analytic. Analyticity of 
the eigenvalue as well as the eigenvector  will then follow from uniform convergence.

The main difficulty posed  by  the degeneracy is  the  iteration procedure of the renormalization analysis. To prove that 
an iteration step is contracting, one has to control the relevant direction.
For this one adjusts the spectral parameter to make vacuum expectations
of the $n$-th renormalized Hamilton operator small. However, in
 a degenerate situation the vacuum expectation is a matrix.
% on the ground 
%state space of the atomic Hamiltonian. 
 The key  idea  is to use the 
symmetry to conclude   
that this matrix is in fact a multiple of the identity, using irreducibility and Schur's Lemma.
This will then turn the analysis of the relevant direction  essentially into a 
one dimensional problem, which can then be handeled with the 
methods in  \cite{GriHas09}. Thus our result is based on results 
from   \cite{GriHas09} as well as  from  \cite{BCFS}.
To this end we  need to  show that 
the symmetry  property as well as the irreducibility property  are  preserved at each iteration step. 

%
%********************OUTLOOK********************

%We believe that the reducibility method which we use can be applied 
%to other degenerate problems in RG. 

%
%******OUTLINE OF THE PAPER*****

Let us give an outline of the paper. 
In Section \ref{sec:ModelSymDegenBoson} we introduce the model and state the main result. 
In Section \ref{sec:symconsiderations} we discuss the analysis related to the symmetry which 
we will need in the proof of the main theorem. 
In Section~\ref{sec:initialFeshbach}
we perform a first  Feshbach map. Note that details 
about the Feshbach map can be found in  Appendix \ref{app:Feshbach}.
We show that the assumptions needed for the Feshbach map
to be applicable are satisfied. 
In Section  \ref{sec:renorm} we introduce Banach spaces of 
matrix valued integral kernels, which describe operators on Fock space.
Polydiscs in 
these spaces will later be  needed to show that the iteration procedure
of the renormalization analysis  converges to a fixed point. 
In Section \ref{sec:firstrafo}  we show that the first Feshbach
map maps the original Hamiltonian into initial polydisc.
In Section \ref{sec:defrgtrafo} we give an explicit definition 
of the renormalization transformation, as a composition of 
the Feshbach map and a rescaling of the energy. 
In Section 
\ref{sec:RenormPresAnalytSym} we show that the 
renormalization transformation preserves analyticity and symmetry.
In Section \ref{sec:iterationSymDegen} we derive conditions 
under which    an 
iterated application of the renormalization transformation is possible
and converges to a fixed point. 
Moreover, we 
show how one can construct the eigenvector, provided the renormalization 
analysis  converges.  
In Section 
\ref{sec:analyt-ev} we provide the proof of the main theorem  by combining
the results which are discussed in previous sections. 
In this section we isolate in Theorem  \ref{genrenthm}   the  part of the renormalization analysis 
which is not model dependent and can be applied to larger 
class of Hamiltonians including for example the standard model of non-relativistic qed. 

In Appendix  \ref{app:symm} we review basic properties 
of antilinear maps. In Appendix 
\ref{sec:eigenproj} we collect properties of eigenprojections of 
isolated eigenvalues. 
In Section \ref{sec:tecAux} we review formal definitions 
of creation and annihilation operators,  and 
 collect identities and  estimates of these operators. 
We plan do consider applications of the main result  in a forthcoming paper elaborating on examples discussed  in \cite{Lan18}.

\section{Model and Statement of Results}
\label{sec:ModelSymDegenBoson}

We consider the following model.
Let the atomic Hilbert space, $\mathcal{H}_{\rm at}$, be a separable complex Hilbert space.
Let $\hh = L^2(\R^3 \times \Z_2 )$  and let
$$
\mathcal{F} = \bigoplus_{n=0}^\infty \FF_n  , \quad \FF_n := S_n(  \hh^{\otimes n} ) 
$$
denote the  Fock space, which is used to describe quantum states of the field. 
Here $S_0(\otimes^0\mathfrak{h}):=\C$ and for $n\geq 1$,
$S_n\in\LL(\otimes^n\mathfrak{h})$ denotes the
orthogonal projection onto the subspace left invariant by all
permutation of the $n$ factors of $\mathfrak{h}$. 
We call  $\FF_n$ the space of $n$-particle subspace.  
 A vector $\psi \in \FF$ can be identified with a 
sequences $(\psi_n)_{n  \in \N_0 }$ such that $\psi_n \in \FF_n$. The vector $\Omega := (1,0,0,...) \in \mathcal{F}$ is called the  Fock vacuum.
Furthermore, 
we shall use the following identification   
%can be identified with 
\begin{align*} 
\FF_n \cong L^2_s(  [\R^3 \times \Z_2]^n  ) 
\end{align*} 
where the  subsript $s$ indicates that the elements are symmetric with respect  to interchange of coordinates.
For details we refer the reader to  \cite{ReeSim1}  or Appendix~\ref{sec:tecAux}.

A unitary operator $U \in \mathcal{L}(\hh)$ can be naturally extended  to  the linear operator  $\Gamma(U)$  in $\mathcal{F}$ by
\begin{align*}
\Gamma(U) |_{\FF_0} = 1 , \quad \Gamma(U) |_{\FF_n} = U^{\otimes n}  , \quad n \in \N 
\end{align*} 
An easy  calculation shows   that $\Gamma(U)$  is unitary again. 
 For $\rho > 0$ and  $f \in \hh$  define  
$$
    (U_{\rho}f)(\boldsymbol{k},\lambda) := \rho^{3/2} f(\rho \boldsymbol{k},\lambda)  , \quad (\boldsymbol{k},\lambda) \in \R^3  \times \Z_2  . 
$$
It is straight forward to see that $U_\rho$ is a unitary operator on $\hh$. The so called dilation operator on $\FF$ is then given  by  
\begin{align} \label{defscaling} 
\Gamma_\rho :=  \Gamma(U_\rho ) . 
\end{align} 
For a vector $z \in \C^N$ we write    $| z | = \left(  \sum_{j=1}^N |z_j|^2 \right)^{1/2}$.
To simplify our notation we define for $(\boldsymbol{k},\lambda) \in \R^3 \times \Z_2$
\begin{align*}%\label{eq:easyNotation}
  k := (\textbf{k}, \lambda), \quad
  \int dk :=  \sum_{\lambda=1,2} \int d^3\textbf{k}, \quad  %\omega(k) := |k| := |\textbf{k}| .
\end{align*}

We will identify the tensor product of the Fock space $\FF$ with a separable Hilbert space  $\HH'$ using the canonical identification 
$$
 \HH' \otimes \mathcal{F} \cong \bigoplus_{n=0}^\infty L_s^2(  [\R^3 \times Z_2]^n  ; \HH' )  , 
$$
cf. \cite{ReeSim1}.
 For $G \in L^2(\R^3\times \Z_2 ; \mathcal{L}(\HH'))$ one associates an  {\bf  annihilation operator } 
$a(G)  $ as follows. For $\psi = (\psi_n)_{n =0}^\infty \in \HH' \otimes \FF$ with the property that $\psi_n = 0$ for all
but finitely many $n$, we define  $a(G) \psi$ as a  sequence of  
 $\HH'$-valued measurable functions such that the $n$-th term  satisfies a.e. 
\begin{align} \label{eqdefofaG}
[a(G) \psi]_n(k_1,....,k_n)  =  (n+1)^{1/2} \int  G(k)^* \psi_{n+1}(k,k_1,....,k_n)  dk  , 
\end{align} 
where the integral on the right hand side is defined as a Bochner integral. 
Eq. \eqref{eqdefofaG} defines a closable operator $a(G)$ whose closure is also denoted by $a(G)$.
The creation operator $a^*(G)$ is defined to be the adjoint of $a(G)$ with respect  to the  natural scalar product in $\FF$. 
In  Appendix~\ref{sec:tecAux}   further properties 
about  creation and annihilation operators can be found. 

In this paper, we are interested in the dynamics of bosonic 
particles of mass zero. The energy, $\omega(k)$, of such a 
particle with wave vector $k$ is 
\begin{align*}%\label{defofomega}
\omega(k)  := |k| := |\textbf{k}| .
\end{align*} 
We define the free-field Hamiltonian, $H_{\rm f}$, on a vector $\psi \in \HH' \otimes \FF$ 
%we define  $H_{\rm f} \psi$  
as the sequence of $\HH'$ -valued functions whose $n$-th term is defined by 
\begin{align}  \label{eq:fieldenergy}
 (H_{\rm f} \psi)_n(k_1,...,k_n) = \sum_{j=1}^n \omega(k_j) \psi_n(k_1,...,k_n)  . 
\end{align} 
The domain  of $H_{\rm f}$, denoted by $D(H_{\rm f})$
 is the set of all $\psi \in \HH' \otimes \FF$ such that  \eqref{eq:fieldenergy} is an element of $\HH' \otimes \FF$.
One verifies that  $H_{\rm f}$  with this domain defines 
a positive, self-adjoint linear operator on $\HH' \otimes \FF$ with purely 
absolutely continuous spectrum, except for an eigenvalue at 0,
with eigenspace consisting of all vectors of the form $(v,0,0,...)$
with $v \in \HH'$.

Let us now fix an atomic Hilbert space $\HH_{\rm at}$.
The Hilbert space, describing the  atomic degrees of freedom  and the quantized field,
is given by the tensor product 
$$
\HH := \HH_{\rm at} \otimes \mathcal{F} .
$$

Let $X$ be an open subset of $\C^\nu$, where $\nu \in \N$.
For each $s \in X$ let  $H_{\rm at}(s)$ be a densely defined closed operator in $\HH_{\rm at}$.
For $g \geq 0$ and $s \in X$ we study the operator
\begin{equation*}%\label{eq:defofmainham}
 H_g(s) := \Hel{s} \otimes \one_\FF
	+ \one_{\HH_{\rm at}} \otimes H_{\rm f} + g W(s)\,.
\end{equation*}
where the   interaction operator is given by
\begin{align} \label{eq:interaction}
	W(s) :=   a(G_{1,\overline{s}}) +a^*(G_{2,s})  .
\end{align}
where $k\mapsto G_{i,s}(k)$ is an element of $L^2(\R^3\times \Z_2;\mathcal{L}(\HH_{\rm at}))$ for each $s \in X$. 
For $\mu > 0$ and $G \in L^2(\R^3\times \Z_2;\mathcal{L}(\HH_{\rm at}))$ we define
	\begin{align}
		\|G\|_\mu := \left(  \int \frac{1}{|k|^{2+2\mu}}
\|G(k)\|^2 dk \right)^{1/2}  ,   \label{def:weaklymearsurBanach}
\end{align}
which possibly may be infinite.

In the following we formulate Hypotheses, which will be used in the statements of the main results Theorem~\ref{thm:symdegenSpinBoson}. 

\begin{hypr}\label{Hypo1}
For $s \in X$ and $j=1,2$ the mapping $s \mapsto G_{j,s}$ is a
bounded analytic function that has values in
$L^2(\R^3\times \Z_2;\mathcal{L}(\HH_{\rm at}))$.
Moreover there exists a $\mu >0$ such that
\begin{equation*}
	\max_{j=1,2} \sup_{s \in X} \|G_{j,s}\|_\mu < \infty \, .
\end{equation*}
\end{hypr}

A consequence of this Hypothesis is that
the interaction operators $W(s)$ and its
adjoint $W(s)^*$ are well-defined operators on
$\HH_{\rm at} \otimes D(H_{\rm f})$ which are 
infinitesimally bounded with respect to $H_{\rm f}$ for all
$s \in X$, cf.  Lemma~\ref{lem:bddforHf}.
Hence the operator $H_g(s)$ is defined on
$D\big(H_{\rm at}(s)\big) \otimes D(H_{\rm f})$.
Since $H_{\rm at}(s)$ is closed, this space is dense
in $\HH$ and $H_g(s)$ is densely defined.
Thus the adjoint $H_g(s)^*$ exists and is closed.
Moreover, $D\big(H_{\rm at}(s)\big) \otimes D(H_{\rm f})$
is contained in the domain of $H_g(s)^*$.
Hence the map
$H_g(s) : D\big(H_{\rm at}(s)\big) \otimes D(H_{\rm f})
\subset \HH \to \HH$
has a densely defined adjoint and is therefore closable
\cite[Theorem~5.28]{Kat95}.
Let us now introduce the notation of a symmetry of an operator. Details can be found in Appendix \ref{app:symm}. 

\begin{definition} %\label{def:symmetry} 
Let $\HH$ be a Hilbert space and $T$ an operator in $\HH$ (possibly unbounded). 
A unitary or antiunitary operator $S$ in $\HH$  is called 
{\bf symmetry  of the operator $T$}, if
\begin{align*}
 S T S^* &= T  \,, \quad \textrm{for }  S
	\textrm{ unitary,} \\
 S T S^* &= T^* \,, \;\; \textrm{for } S
	\textrm{ antiunitary.}
\end{align*}
In that case we say that $T$ is {\bf  symmetric or invariant with respect} to $S$. If $T$ is symmetric with respect to all elements of a set $\mathcal{S}$  of symmetries,
we say $T$ is symmetric or invariant with respect to $\mathcal{S}$. 
\end{definition}

\begin{remark} We note that the set of symmetries of an operator form a group. More precisely,
if  $S_1$ and $S_2$  are  symmetries, then so are  $S_1 S_2$ and $S_1^{-1}$.  Thus without loss of generality 
we can assume that we are given a group of symmetries.
\end{remark}

To formulate the second Hypothesis we need the notion of a discrete point in the spectrum  of a closed operator. 
We use the definition as given in \cite{ReeSim4}. To state it let us first recall the following theorem. We shall make use of the following notation for open balls in the complex plane 
\begin{equation*}
B_r(a) =  \{ z \in \C : |z-a|< r \} \,, 
\end{equation*}
where $a \in \C$ and $r > 0$. 
%We also set $D_r := D_r(0)$.
%\marginpar{Do we need this notation $D_r$?}

\begin{theorem}[ \cite{ReeSim4} Theorems XII.5 (a) \& (b)]\label{thm:ReeSim4XII5FirstPart}
Suppose that $A$ is a closed operator and let $\lambda$ be an isolated point of $\sigma(A)$.
Then   $B_\epsilon(\lambda) \cap \sigma(A) = \{ \lambda \}$  for some $\epsilon > 0$,  and   for any $r \in (0,\epsilon)$ the integral
\begin{equation} \label{eq:defofplambda}
P_\lambda =  \frac{1}{2 \pi i} \ointctrclockwise_{|\mu-\lambda|=r} (\mu-A)^{-1} d \mu
\end{equation}
exists and is independent of $r$. Moreover, $P_\lambda$ is  a projection, i.e., $P_\lambda^2 = P_\lambda$.
\end{theorem}

\begin{definition}
Let $A$ be a closed operator. A point $\lambda \in \sigma(A)$ is called {\bf discrete} if $\lambda$ is isolated
and $P_\lambda$, given by Theorem~\ref{thm:ReeSim4XII5FirstPart}, is finite dimensional. If $P_\lambda$ is one dimensional we say
$\lambda$ is a {\bf nondegenerate} eigenvalue. The dimension of  $P_\lambda$ is called
the {\bf algebraic multiplicity}. The dimension of $\ker(A - \lambda)$ is called the {\bf geometric multiplicity}.
If algebraic and geometric multiplicity agree and are finite, we say $\lambda$ is {\bf non-defective}.
\end{definition}

We can now state the second Hypothesis.

\begin{hypr}\label{Hypo2}
\textrm{ } 
\begin{itemize} 
\item[(i)] The mapping $s \mapsto H_{\rm at}(s)$ is an
analytic family in the sense of Kato.
\item[(ii)]  There exists %a point
		$s_0 \in  X$  such that
		$E_{\rm at}(s_0)$ is 
		a non-defective, discrete element of 
		%an isolated  point  
		the spectrum of  $H_{\rm at}(s_0)$. 
\item[(iii)] If $E_{\rm at}(s_0)$ is degenerate,  there
	exists a group of symmetries,
	$\mathcal{S}$, such that  $H_{\rm at}(s) \otimes \one_{\FF}$, $H_{\rm f}$,  and $W(s)$ are symmetric with respect to $\mathcal{S}$ for all   $s \in X$.
	Each element of  $\mathcal{S}$ can be written in the form
 $S_1 \otimes S_2$, where  $S_1$ is a symmetry in $\HH_{\rm at}$ and $S_2$ is a symmetry in $\FF$.
Furthermore, the set of symmetries in $\HH_{\rm at}$
$$
\mathcal{S}_1 :=  \{ S_1 :  S_1 \otimes S_2 \in \mathcal{S} \}
$$
acts irreducibly on the eigenspace of $H_{\rm at}(s_0)$ with eigenvalue $E_{\rm at}(s_0)$.
 Each element of $\mathcal{S}_2 := \{ S_2 :  S_1 \otimes S_2 \in \mathcal{S} \}$ leaves the 
Fock vacuum as well as  the one particle subspace  invariant and 
 commutes with the operator 
of   dilations, cf.  \eqref{defscaling}.
\end{itemize}
\end{hypr}

By Hypothesis~\ref{Hypo2} and the Kato-Rellich theorem of analytic perturbation theory, \cite{ReeSim4}, together with a symmetry
argument one can show the following lemma, which will be needed to formulate the third hypothesis.  We note that parts (a)
and (b) are well known results and can be found in \cite{ReeSim4}. The proof of (c) will require a symmetry
argument. We will  provide a   proof in  Section~\ref{sec:symconsiderations}.

\begin{lemma}\label{lem:propofpatprojection} 
Suppose the situation is as in Hypothesis~\ref{Hypo2}. Then there exists  an $\epsilon > 0$ sufficiently small
and  a  neighborhood $N \subset X$ of $s_0$,  such that the following holds.
\begin{itemize}
\item[(a)]   $\{  z \in \C : |z - E_{\rm at}(s_0) | = \varepsilon \}  \subset  \rho(H_{\rm at}(s))$ for all $ s \in N$.
\item[(b)]   For all $s \in N$
 \begin{align}\label{eq:spectralprojection}
	p_{\rm at}(s) = - \frac{1}{2\pi i}
	\ointctrclockwise_{|z-E_{\rm at}(s_0)| = \varepsilon}
	\frac{1}{H_{\rm at}(s) - z} \,dz \,
 \end{align}
defines a projection valued  analytic function and the dimension of the range is finite and constant. In particular,  $p_{\rm at}(s_0)$ projects onto the eigenspace of $E_{\rm at}(s_0)$. 
\item[(c)]
There exists an analytic
function $e_{\rm at}: N \to \C$ such that for all $s \in N$
\begin{equation*}
	H_{\rm at}(s) \upharpoonright \ran \, p_{\rm at}(s) = e_{\rm at}(s) \upharpoonright \ran \, p_{\rm at}(s)\, .
\end{equation*}
For  $s \in N$ the point $e_{\rm at}(s) \in \C$ is the only
point in the spectrum of   $H_{\rm at}(s)$ in a neighborhood of $E_{\rm at}(s_0)$.  
The number $e_{\rm at}(s)$ is a non-defective, discrete element of the spectrum of $H_{\rm at}(s)$. 
Furthermore, $e_{\rm at}(s_0) = E_{\rm at}(s_0)$.
\end{itemize}
\end{lemma}

If Hypothesis~\ref{Hypo2} holds, it follows from a repeated application of Lemma~\ref{lem:propofpatprojection}, that there exists a connected open neighborhood $X_1 \subset X$ of $s_0$, 
 an analytic projection valued function $P_{\rm at}$ on $X_1$,  and an analytic function $E_{\rm at}$ on $X_1$ extending $E_{\rm at}(s_0)$ 
such that the following holds. 
For all $s \in X_1$ the number $E_{\rm at}(s)$ is in the discrete spectrum  of $H_{\rm at}(s)$ and it  is non-defective,  moreover    
 \begin{equation*}%\label{eq:spectralprojection11}
	H_{\rm at}(s) \upharpoonright \ran P_{\rm at}(s) = E_{\rm at}(s) \upharpoonright \ran P_{\rm at}(s)\, .
\end{equation*}
For any $s_1 \in X_1$ the  there   exists an $\epsilon_1 > 0$ and a neighborhood $N_{1} \subset X_1$ of $s_1$ such that for all $s \in N_1$
$$\{  z \in \C : |z - E_{\rm at}(s_1) | = \varepsilon_1 \}  \subset  \rho(H_{\rm at}(s_1))$$
and 
 \begin{align}\label{eq:spectralprojection1}
	P_{\rm at}(s) = - \frac{1}{2\pi i}
	\ointctrclockwise_{|z-E_{\rm at}(s_1)| = \varepsilon_1}
	\frac{1}{H_{\rm at}(s) - z} \,dz \, . 
 \end{align}
Henceforth, we denote by $P_{\rm at}$ and $E_{\rm at}$ any mappings  having  the  properties stated above on an open connected neighborhood $X_1 \subset X$  of $s_0$.

\begin{remark} In principle one could use Lemma~\ref{lem:propofpatprojection} to obtain a maximal analytic extension of $P_{\rm at}$ and $E_{\rm at}$.  This will not be needed as it does not necessarily improve the main result. 
\end{remark}

To formulate the third  Hypothesis, we use   the notion of a reduced resolvent, which is introduced in Remark \ref{re:reducedres}, below.

\begin{remark}\label{re:reducedres}   Let $A :D(A) \subset X \to X$ be a densely defined  closed linear operator and let $P$ 
be a bounded projection in $X$  such for $\overline{P} = 1 - P$
\begin{align}\label{assumpforreduced} & \ran \overline{P} \text{ is closed } , \quad   \ran \overline{P} \cap D(A)  \text{ is dense in  } \ran \overline{P} \\
& A \left[ \ran \overline{P} \cap D(A) \right] \subset \ran \overline{P} . \nonumber 
\end{align}
Then  it is reasonable to  study the densely defined operator $A |_{\ran \overline{P} \cap D(A)}$ in $\ran \overline{P}$. 
If  $z \in \rho(A|_{\ran \overline{P} \cap D(A)})$ we shall use the notation
 $(A-z)^{-1} \overline{P} :=     ((A-z)|_{\ran \overline{P} \cap D(A)})^{-1}  \overline{P}$, and
refer to this expression as the \textbf{reduced resolvent}.
\end{remark}

The third  Hypothesis  will be used     to invert for $z$ close to $E_{\rm at}(s_0)$ the operator $H_{\rm at}(s)-z$     when
restricted to
the range of
$$\overline{P}_{\rm at}(s)
:= \one_{\HH_{\rm at}}-P_{\rm at}(s) . $$
Aforementioned  we formulate this in terms of the reduced resolvent. For this,  
we  note that 
it follows from well known properties about projections \eqref{eq:defofplambda}, c.f. \cite{ReeSim4} or part (a) of  Lemma \ref{thm:ReeSim4XII5} 
in the appendix, 
that the assumptions \eqref{assumpforreduced}, i.e.,  
\begin{align*} & \ran \overline{P}_{\rm at}(s) \text{ is closed } , \quad   \ran \overline{P}_{\rm at}(s)  \cap D(H_{\rm at}(s))  \text{ is dense in  } \ran \overline{P}_{\rm at}(s)  \\
& H_{\rm at}(s)  \left[ \ran \overline{P}_{\rm at}(s)  \cap D( H_{\rm at}(s)) \right] \subset \ran \overline{P}_{\rm at}(s)  . \nonumber 
\end{align*}
 are satisfied for  $s \in X_1$. 
Thus  the reduced operator $H_{\rm at}(s)|_{\overline{P}_{\rm at}(s)  \cap D(H_{\rm at}(s)))}$  is a densely defined operator 
in  $\ran \overline{P}_{\rm at}(s)$.

\begin{hypr}\label{Hypo3}
Hypothesis~\ref{Hypo2} holds and there exists a neighborhood
$\mathcal{U} \subset X_1 \times \C$ of
$(s_0,E_{\rm at}(s_0))$
such that for all $(s,z) \in \mathcal{U}$ we have
$|E_{\rm at}(s) - z | < 1/2$,  $\sup_{(s,z) \in \mathcal{U} } \| P_{\rm at}(s) \| < \infty$,   and 
\begin{equation*}
	\sup_{(s,z) \in \mathcal{U}} \sup_{q \geq 0}
	\left\| \frac{q+1}{H_{\rm at}(s)- z + q }
	\overline{P}_{\rm at}(s)\right\| < \infty \,.
\end{equation*}
\end{hypr}

\begin{remark}
We note that one can show that Hypothesis~\ref{Hypo3} follows from Hypothesis~\ref{Hypo1} and \ref{Hypo2} and the additional
assumption that $H_g(s)$ is an analytic family of type (A) and that  a semiboundedness condition holds,
see \cite{GriHas09}.
\end{remark}

When dealing with the ground state, we can  assume the following  additional Hypothesis.  
It  will ensure that in the limit, as the interaction strength tends to zero, the ground state of the interacting system converges to the ground state of the non-interacting system. 
For a subset $\Omega \subset \C^n$ we write  $\Omega^* := \{ \overline{z} :   z \in \Omega\}$. 

\begin{hypr} \label{Hypo4}
The following holds.
\begin{itemize}
\item[(i)]
We have  $X = X^*$ and  for all $s \in X$ the identities $G_{1,s} = G_{2,s}$ and $H_{\rm at}(s)^* = H_{\rm at}(\overline{s})$ hold.
 \item[(ii)]
We have $s_0 \in X \cap \R^\nu $ and $E_{\rm at}(s_0) = \inf \sigma (H_{\rm at }(s_0))$. 
\end{itemize}
\end{hypr}

\begin{definition} Let $\HH_0$ be a Hilbert space  and let $X \subset \C^d$ with  $X^* = X$. For  each $x \in X$ let a densely defined  operator $T(x)$  in the  Hilbert  space $\HH_0$
be given. 
We say that $T $ is  {\bf reflection symmetric} if $T(x)^* = T(\overline{x})$. 
\end{definition} 

With theses Hypotheses at hand we can now state the
main result.

\begin{theorem} \label{thm:symdegenSpinBoson}
Suppose Hypotheses~\ref{Hypo1}, \ref{Hypo2}, \ref{Hypo3} %, and IV  
hold and  let $$d = \dim {\rm ker} ( H_{\rm at}(s_0) - E_{\rm at}(s_0) ) .$$
Then there exists a neighborhood $X_b \subset X$ of $s_0$
and a positive constant $g_b$ such that for all $s \in X_b$
and all $g \in [0, g_b]$
 the operator $H_g(s)$ has an eigenvalue $E_g(s)$ with  $\degendim$ linearly independent eigenvectors $\psi_{g,j}(s)$, $j=1,...,\degendim$,
with the following properties.
\begin{itemize}
\item[(i)] The functions $s \mapsto E_g(s)$  and $s \mapsto \psi_{g,j}(s)$  for $j=1,...,\degendim$  are analytic functions on $X_b$. 
%if $g \in [0,  g_0)$.
\item[(ii)]  Uniformly in $s \in X_b$
we have $\lim_{g \to 0} E_g(s) =  E_{\rm at}(s)$ and  $\lim_{g \to 0} \psi_{g,j}(s)  =  \varphi_{{\rm at},j}(s) \otimes  \Omega $ for some $\varphi_{{\rm at},j}(s) \in \ran P_{\rm at}(s)$.
\end{itemize}
If in addition Hypothesis~\ref{Hypo4} holds, then   $X_{\rm b} = X_{\rm b}^*$  and 
\begin{itemize}
\item[(iii)]  for all $s \in X_b \cap \R^\nu$ it holds that $E_g(s) = \inf \sigma(H_g(s))\,,$
\item[(iv)]  for all $s \in X_b $ it holds that $\overline{E}_g(s) =  E_g(\overline{s})$. 
\end{itemize}

\end{theorem}

\begin{remark}
In case that the
irreducibility assumptions of Hypothesis~\ref{Hypo2} (iii) is not met the
eigenspace of the ground-state eigenvalue is
expected to split at higher order in perturbation theory. This phenomenon is known as the Lamb shift
and has been considered in the literature
\cite{KroLam49,HaiSei02}.
It is natural to assume that  degeneracies of   eigenvalues are lifted at some order in perturbation theory 
until they are protected by a set of symmetries.
Analyticity questions for degenerate ground-state eigenvalues which are lifted in second order perturbation 
theory  where investigated in \cite{HasLan18-1} in the framework
of generalized Spin-Boson models.
\end{remark}

We note that the above result can be used to obtain analyticity  in the coupling constant. We note that this will 
  immediately  improve the continuity statement,  Part (ii),   in   Theorem~\ref{thm:symdegenSpinBoson}. 
This will be the content of the following corollary. 
To state the  result first recall  that   $W(s)$ is  infinitesimally $H_{\rm f}$ bounded, cf. Lemma \ref{lem:bddforHf}.
Thus  for each $s \in X$ the map on $\C$ 
\begin{align*}
g  \mapsto H_g(s) 
\end{align*} 
is an analytic family of type (A). It follows that $(g,s) \mapsto H_g(s)$ is an analytic 
family, since the weak analyticity of the resolvent implies strong analyticity of the resolvent and to show jointly  weak 
analyticity we can use Hartog's theorem,  cf.  \cite{Hoe90}.

\begin{corollary} 
Suppose Hypotheses~\ref{Hypo1}, \ref{Hypo2}, \ref{Hypo3} %, and IV  
hold and  let $d = \dim {\rm ker} ( H_{\rm at}(s_0) - E_{\rm at}(s_0) ) $.
Then there exists a neighborhood $X_b \subset X$ of $s_0$
and a positive constant $g_b$ such that for all $s \in X_b$
and all $g \in B_{g_b}(0)$
 the operator $H_g(s)$ has an eigenvalue $E_g(s)$ with  $\degendim$ linearly independent eigenvectors $\psi_{g,j}(s)$, $j=1,...,\degendim$,
with the following property.
\begin{itemize}
\item[ ] The functions $(s,g) \mapsto E_g(s)$  and $(s,g) \mapsto \psi_{g,j}(s)$  for $j=1,...,\degendim$  are analytic functions on $X_b \times B_{g_{\rm b}}(0)$. 
\end{itemize}
\end{corollary}

\begin{proof}
First we  extend the parameter space $\hat{X} = X \times  B_{1}(0)$ and define  for $(s,s') \in \hat{X}$ and $g \geq 0$
\begin{align} \label{anacopcons} 
\hat{H}_{g}(s,s') &= H_{(s'g)}(s) . 
\end{align}
Now  one easily  verfies that $(s,s') \mapsto \hat{H}_{g}(s,s') $ satisfies the assumptions ~\ref{Hypo1}, \ref{Hypo2}, \ref{Hypo3}. Thus  it 
follows from Theorem  \ref{thm:symdegenSpinBoson} that there exists a $g_{\rm b} > 0$ such that $\hat{H}_{g_{\rm b}}(s,s')$
has an eigenvalue $E_{g_{\rm b}}(s,s')$ and an eigenvector $\psi_{g_{\rm b}}(s,s')$ both depending analytically on $(s,s')$. 
Now in view of \eqref{anacopcons} we see that they are also eigenvalue and eigenvector of  $H_{(s'g_{\rm b})}(s) $. This shows the corollary. 
\end{proof}

We note that one can formulate the result in Theorem \ref{thm:symdegenSpinBoson} in  terms  of so called  eigenprojections. 

A densely defined operator $H$  in a Hilbert space with the property that \begin{align} \label{defofcomplexsa} H^* = \mathcal{J} H \mathcal{J}^{-1} \end{align} 
for some antiunitary operator $\mathcal{J}$ is called  {\bf complex-selfadjoint with respect to } $\mathcal{J}$. 
To formulate the next corollary we make another hypothesis.

\begin{hypr} \label{Hypo6} Hypothesis  \ref{Hypo2} holds.
For all $g \geq 0$ and $s \in X$ the operator 
$H_g(s)$ is complex-selfadjoint  with respect to a antiunitary operator $\mathcal{J}$.
The bilinear form $J  : V \times V \to \C$  on  $V  := \ran ( P_{\rm at}(s_0) ) \otimes \Omega$
defined by  $ J(v_1,v_2 ) =  \inn{ v_1 , \mathcal{J} v_2 }$ is non-degenerate.  
\end{hypr} 

\begin{corollary} \label{cor2:symdegenSpinBoson}
Suppose Hypotheses~\ref{Hypo1}, \ref{Hypo2}, \ref{Hypo3} %, and IV  
hold and  let $d = \dim {\rm ker} ( H_{\rm at}(s_0) - E_{\rm at}(s_0) ) $.
Assume  that   Hypothesis \ref{Hypo4}  or  Hypothesis  \ref{Hypo6} holds. 
Then there exists a neighborhood $X_b \subset X$ of $s_0$
and a positive constant $g_b$ such that for all $s \in X_b$
and all $g \in [0, g_b]$
 there exists a  complex number $E_g(s)$ and a  projection $P_g(s)$ with rank $d$ such that 
\begin{align} \label{eq:projprop}  
P_g(s)  H_g(s) \subset H_g(s) P_g(s)  = E_g(s) P_g(s) 
\end{align} 
with the following properties.
\begin{itemize}
\item[(i)]  $s \mapsto P_g(s)$ and $s \mapsto E_g(s) $ are  analytic on $X_b$  
\item[(ii)]   $ \lim_{g \downarrow  0} P_g(s) = P_{\rm at}(s) \otimes P_\Omega$ uniformly on $X_b$. 
\end{itemize}
\end{corollary} 
\begin{proof} Let the situation be as in  Theorem \ref{thm:symdegenSpinBoson}.
First we  assume that Hypothesis~\ref{Hypo4}  holds. By possibly restricting to the intersection of $X_b$ and $X_b^*$ we can assume 
without loss that these sets are equal and nonzero, since both contain $s_0 \in \R^\nu$. 
Define the matrix $M_{a,b}(s) = \inn{ \psi_{g,a}(\overline{s}) , \psi_{g,b}(s) }$, $a,b=1,...,d$,  for $s \in X_{\rm b} \cap X_{\rm b}^*$.
 By linear independence of the  $\psi_{g,j}(s) $ and continuity we can assume without loss that $M$ is invertible for all $s \in X_b$
(by possible  making $X_b$ smaller, by intersecting it with a neighborhood of the real line). 
We define   
\begin{equation*}%\label{eq:ProjectionOpPgs}
P_g(s) = \sum_{a,b=1}^d  |\psi_{g,a}(s)\rangle ( M(s)^{-1})_{a,b}    \langle\psi_{g,b}(\overline{s})| .
\end{equation*}
It is straightforward to verify that this is  a projection 
\begin{align*}%\label{eq:ProjectionOpPgs}
& P_g(s) P_g(s) \\
 &  = \sum_{a,b=1}^d  |\psi_{g,a}(s)\rangle ( M(s)^{-1})_{a,b}    \langle\psi_{g,b}(\overline{s})|  \sum_{c,e=1}^d  |\psi_{g,c}(s)\rangle ( M(s)^{-1})_{c,e}    \langle\psi_{g,e}(\overline{s})|  \\
&  = \sum_{a,b,c,e=1}^d  |\psi_{g,a}(s)\rangle ( M(s)^{-1})_{a,b}    M(s)_{b,c} ( M(s)^{-1})_{c,e}    \langle\psi_{g,e}(\overline{s})|  \\
&  = \sum_{a,b,e=1}^d  |\psi_{g,a}(s)\rangle ( M(s)^{-1})_{a,b}   \delta_{b,e}    \langle\psi_{g,e}(\overline{s})|  \\
&  = \sum_{a,b=1}^d  |\psi_{g,a}(s)\rangle ( M(s)^{-1})_{a,b}    \langle\psi_{g,b}(\overline{s})| = P(s)  .
\end{align*}
Furthermore, since $\psi_{g,a}$ are eigenvectors we find  $H_g(s) P_g(s)  = E_g(s) P_g(s)$ and  with Theorem \ref{thm:symdegenSpinBoson} (iv) 
\begin{align*}  P_g(s) H_g(s)  & \subset  \sum_{a,b=1}^d  |\psi_{g,a}(s)\rangle ( M(s)^{-1})_{a,b}    \langle H_g(s)^* \psi_{g,b}(\overline{s})| \\
& = \sum_{a,b=1}^d  |\psi_{g,a}(s)\rangle ( M(s)^{-1})_{a,b}    \langle  H_g(\overline{s})  \psi_{g,b}(\overline{s})|  \\
& = \sum_{a,b=1}^d  |\psi_{g,a}(s)\rangle ( M(s)^{-1})_{a,b}    \langle  E_g(\overline{s})  \psi_{g,b}(\overline{s})|  \\
& = \sum_{a,b=1}^d  |\psi_{g,a}(s)\rangle ( M(s)^{-1})_{a,b}    \langle  \psi_{g,b}(\overline{s})|    E_g({s}) =  P_g(s) E_g(s)  . 
\end{align*} 
It is now straight forward  using Parts (i) and (ii) of  Theorem \ref{thm:symdegenSpinBoson}   that Parts  (i) and (ii) of  Corollary 
\ref{cor2:symdegenSpinBoson} hold.

Now assume that  Hypothesis~\ref{Hypo6}  holds. In that case we argue analogously.
Define the matrix $N_{a,b}(s) = \inn{ \mathcal{J} \psi_a({s}) , \psi_b(s) }$, $a,b=1,...,d$,  for $s \in X_{\rm b}$.
 Again by linear independence of the  $\psi_{g,j}(s) $ and  Hypothesis~\ref{Hypo6} 
we find  that $N_{a,b}(s)$ is invertible for $s=s_0$ and $g=0$. Now by    continuity  in $s$
and  (ii) of  Theorem  \ref{thm:symdegenSpinBoson} we can assume without loss 
 that $N$ is invertible for all $s \in X_b$
(by possible  making $X_b$ as well as $g_{\rm b} > 0$ smaller). 
 It is now again straightforward to verify using (i) and (ii) of  Theorem \ref{thm:symdegenSpinBoson} that 
\begin{equation*}%\label{eq:ProjectionOpPgs1}
P_g(s) = \sum_{a,b=1}^d  |\psi_{g,a}(s)\rangle ( N(s)^{-1})_{a,b}    \langle \mathcal{J} \psi_{g,b}(s)| .
\end{equation*}
has the claimed properties. To show the first relation in  \eqref{eq:projprop} we observe that using  \eqref{defofcomplexsa}
we find 
\begin{align*}  
P_g(s) H_g(s)  & \subset  \sum_{a,b=1}^d  |\psi_{g,a}(s)\rangle ( N(s)^{-1})_{a,b}    \langle H_g(s)^* \mathcal{J} \psi_{g,b}(s)| \\
& = \sum_{a,b=1}^d  |\psi_{g,a}(s)\rangle ( N(s)^{-1})_{a,b}    \langle \mathcal{J}  H_g({s})  \psi_{g,b}({s})|  \\
& = \sum_{a,b=1}^d  |\psi_{g,a}(s)\rangle ( N(s)^{-1})_{a,b}  \langle \mathcal{J}  E_g({s})  \psi_{g,b}({s})|    \\
& = \sum_{a,b=1}^d  |\psi_{g,a}(s)\rangle ( N(s)^{-1})_{a,b}    \langle  \mathcal{J} \psi_{g,b}({s})|    E_g({s}) = P_g(s) E_g({s})  . 
\end{align*} 

\end{proof} 

\section{Symmetry Considerations}

 \label{sec:symconsiderations}

In this section we consider consequences   of the symmetries which will 
be used for the renomormalization analysis.  Elementary definitions and properties are collected in  Appendix  \ref{app:symm}.
First  we  discuss  Schur's Lemma for symmetries of an operator.
This will be needed to show that certain  matrix valued vacuum expectations,
occurring in the renormalization analysis, are multiples of the identity.
Then we consider general properties of symmetries of  analytic family of operators. 
We will apply these properties  to the Hamiltonian defined in Section~\ref{sec:ModelSymDegenBoson}.
As a main result, see Lemma~\ref{lem:wlog}, we will be able to assume without loss of generality that $P_{\rm at}(s)$ is a constant 
function of  $s$.
Moreover, in Lemma \ref{lem:FeshbachpreservesSymmetry} at the end of this section we prove a crucial property of the Feshbach operator which will be important later during the renormalization procedure.

\begin{definition}%\label{def:irreducible}
Let $V$ be a subspace of a
Hilbert space
$\HH$ and let $\mathcal{S}$ be a
 set whose elements are unitary or antiunitary operators on $\HH$.
 We say that $S \in \mathcal{S}$ acts {irreducibly}
 on  $V$ if for any subspace $W$ of
 $V$ with $S W \subset W$ we have $W = \{ 0 \}$ or
 $W = V$.
\end{definition}

The next two lemmas are versions of
the well-known Lemma of Schur \cite{Sch05}.
The first lemma  is for self-adjoint operators.
Since analytic continuations of the  Hamiltonian are  in general
non-self-adjoint we need  a second lemma  for
ordinary linear operators, as well.
\begin{lemma}%(First Lemma of Schur)
\label{lem:firstlemschur}
Let $\mathcal{S}$ be a set containing  unitary
and antiunitary operators which act
irreducibly on a complex finite-dimensional
Hilbert space $V$.
Let $T$ be a self-adjoint linear operator on $V$
such that
\begin{equation*}
 S T S^* = T\,, \quad \textrm{ for all } S \in \mathcal{S}.
\end{equation*}
Then there exists a number $\lambda \in \R$ such that
$T = \lambda \, \one_V$.
\end{lemma}
\begin{proof}
First observe that $T$ has a real  eigenvalue,
say $\lambda$.  Thus
$T - \lambda$ has a nonvanishing kernel. Now $S$
leaves the space  $\ker (T-\lambda)$ invariant since
$\lambda$ is real.
Thus by irreducibility we see that
$\ker (T - \lambda ) = V$. This yields the claim.
\end{proof}

Now we want to extend the above lemma to
non-self-adjoint operators.
\begin{lemma}%(Second Lemma of Schur)
\label{lem:seclemschur}
Let $\mathcal{S}$ be a set containing unitary
and antiunitary operators which act
irreducibly on a complex finite-dimensional Hilbert
space $V$.
Let $T$ be a  linear operator on $V$ such that
\begin{align} \label{eq:S*TS=T}
&& S T S^* &= T \, , \quad  \textrm{ for all } S \in \mathcal{S} ,
	\ S \ \text{unitary}, \\
&& S T S^* &= T^* \,, \quad\!\!\! \textrm{ for all } S \in
	\mathcal{S}  , \ S \ \text{antiunitary} .
	\nonumber
\end{align}
Then there exists a number $\lambda \in \C$ such that
$T = \lambda \, \one_V$.
\end{lemma}
\begin{proof}
Note that there exits a unique decomposition 
\begin{equation} \label{eq:decompopuniq} 
	T = Z + i Y\,,
\end{equation}
with $Y$ and $Z$ self-adjoint operators on V.
Then it follows from Eq.~\eqref{eq:S*TS=T}     that for $S$ unitary/antiunitary
\begin{align*}
Z  \pm  i Y = S ( Z + i Y ) S^* = S  Z S^*  \pm i  S Y S^*  .
\end{align*}
 The uniqueness of the decomposition \eqref{eq:decompopuniq} and Lemma \ref{lem:prodanti}  (c)    implies 
\begin{equation*}
	S Z  S^* = Z \,, \quad S Y S^* = Y\, ,
\end{equation*}
for all $S \in \mathcal{S}$. Thus $Z$ and $Y$
are multiples of the identity by
Lemma~\ref{lem:firstlemschur}.
\end{proof}

The next proposition will allow us to work with  the constant  projection  $P_{\rm at}(s_0)$  instead of the $s$ dependent projection   $P_{\rm at}(s)$, by means of an invertible analytic family.
This is a standard method  used in analytic perturbation theory.
The theorem below is a version of Theorem XII.12  in \cite{ReeSim4}
incorporating in addition a symmetry property.

\begin{theorem}\label{lem:trafofunc} Let $\HH$ be a Hilbert space.
Let $P(s) \in \mathcal{L}(\HH)$ be a projection-valued  analytic function
on a connected, simple connected region of the
complex plane $X$.
For $s_0 \in X$ there exists an
analytic family $U(s)$ of bounded and invertible operators on $X$ with
the following properties:
\begin{itemize}
	\item[(a)] $U(s)\, P(s_0)\, U(s)^{-1} = P(s)$.
	\item[(b)] If $s_0$ is real and $P(s)$ is self-adjoint for real $s$, 
		then we can choose $U(s)$ unitary for real $s$. Furthermore, $U(\overline{s})^* = U(s)^{-1}$ for all $s \in X \cap X^*$.
\item[(c)] If $S$ is a symmetry of $P(s)$, then one can choose $U(s)$ to satisfy
\begin{align*}
&S U(s) S^* = U(s) , \qquad \text{ if } S \text{ is unitary} , \\
& S U(s) S^* = (U(s)^{-1})^* , \qquad \text{ if } S \text{ is antiunitary} .
\end{align*}
\end{itemize}
\end{theorem}

For the proof we use as in \cite{ReeSim4}   the following lemma.

\begin{lemma} \label{lem:ExistenceAndUniqunessOfIVP}
Let $R$ be a connected, simply connected subset of $\C$ with $\beta_0 \in R$
and let $A(\beta)$ be an analytic function on $R$ with values in the bounded operators on some Banach space $\mathcal{X}$.
Then for any $x_0 \in \mathcal{X}$, there is a unique function $f(\beta)$, analytic in $R$, with values in $X$ obeying
$$
\frac{d}{d\beta} f(\beta) = A(\beta) f(\beta) , \quad f(\beta_0) = x_0 .
$$
\end{lemma}

For a proof of the lemma we refer the reader to  \cite{ReeSim4}.

\begin{proof}[Proof of Theorem \ref{lem:trafofunc}]  The detailed proofs of (a) and (b) can be found  in Theorem~XII.12 of \cite{ReeSim4}. Here we
merely give a sketch.
 Let  $Q(s) =  P'(s)P(s) - P(s) P'(s)$, where $P'(s) = \frac{d}{ds} P(s)$.  Then a calculation shows that
\begin{equation} \label{eq:commpwithq}
P'(s) = [Q(s),P(s)] .
\end{equation}
We now use   Lemma \ref{lem:ExistenceAndUniqunessOfIVP} with $\mathcal{X} = \mathcal{L}(\HH)$. Let
 $U(s)$ is the unique solution of the initial value problem
\begin{align} \label{dq:difffeqorq}
& \frac{d}{ds} U(s) = Q(s) U(s)  ,  \qquad U(s_0) = 1 ,
\end{align}
and let   $V(s)$ be  the unique solution of the initial value problem
\begin{align}\label{dq:difffeqorq2}
& \frac{d}{ds} V(s)  = -V(s) Q(s)  , \qquad V(s_0) = 1 .
\end{align}
Since
$$
\frac{d}{ds} (V(s) U(s)) = \frac{dV}{ds} U(s) + V(s) \frac{d U}{ds} = 0  ,
$$
it follows that 
\begin{align} \label{eq:VU=1}
	V U = 1.
\end{align} 
On the other hand if $F = U V$, then $F$ solves the
differential equation $F' = [Q,F]$ with initial condition $F(s_0)=1$. Since $F=1$ solves
the same initial value problem it follows  by uniqueness that 
\begin{align}\label{eq:UV=1}
	 UV = 1.
\end{align}
It follows that $U$ is invertible. Furthermore,
a calculation shows that  $\tilde{P} = U P(s_0) V$ satisfies that initial value problem $\tilde{P}(s_0)=P(s_0)$
and $\tilde{P}' = [Q,\tilde{P}]$. Thus  from \eqref{eq:commpwithq}  we see that  $\tilde{P}$ and $P$ satisfy the same
initial value problem and hence agree.  This shows (a). To show   (b) let us suppose that $P(s) = P(s)^*$ for $s= \overline{s}$.
By the Schwarz reflection principle, it follows that $P(s)^* = P(\overline{s})$ for all $s \in X \cap X^*$. By the definition of $Q$, $Q(s)^* = - Q(\overline{s})$.
Let $\tilde{V}(s) = U(\overline{s})^*$. Then $\tilde{V}$ obeys $d \tilde{V} / ds = - \tilde{V}(s) Q(s)$; $\tilde{V}(s_0) = I$. By the uniqueness
of solutions of differential equations, $\tilde{V}(s) = V(s)$. Thus,  $U(\overline{s})^*= \tilde{V}(s) = V(s)   = U(s)^{-1}$,  and if  $s$ is real, $U(s)^*  = U(s)^{-1}$ and  so $U(s)$ is unitary.

It remains to show (c).  Suppose first that  $S$ is a  unitary symmetry of $P(s)$. Then we have by assumption  $S P(s) S^* = P(s)$ and hence
$\frac{d}{ds} P(s)=  S \frac{d}{ds} P (s) S^* $.
It follows that $S Q(s) S^* = Q(s)$. Using  \eqref{dq:difffeqorq} we thus obtain
$$
\frac{d}{ds} S U(s) S^* = S \frac{d}{ds} U(s) S^* = S Q(s) U(s) S^* =  Q(s) S  U(s) S^* , \quad S U(s_0) S^* = 1 .
$$
By uniqueness of the initial value problem, Lemma~\ref{lem:ExistenceAndUniqunessOfIVP}, we conclude 
\begin{align*}
	S U(s) S^* = U(s).
\end{align*}
Now let us suppose that $S$ is an antiunitary symmetry of $P(s)$. Then we have by assumption  $S P(s) S^* = P(s)^*$, and hence taking the adjoint we find $S P(s)^* S^* = P(s)$.
Differentiating   we find
$ \frac{d}{ds} P(s) = S \left(  \frac{d}{ds} P (s) \right)^* S^*$.
A calculation now shows  that \begin{equation} \label{eq:unitarysym} S Q(s)^* S^* = -Q(s) . \end{equation}
By   \eqref{dq:difffeqorq} we  have $ ( S U(s_0) S^*)^*  = 1 $ 
and 
\begin{align*}
 \frac{d}{ds} \left( S U(s) S^* \right)^* & = \left(  S \frac{d}{d {s}} U(s) S^* \right)^*  = \left( S Q(s) U(s) S^* \right)^*=   ( S  U(s) S^*)^* S Q(s)^* S^* \\
&   =  - ( S  U(s) S^*)^* Q(s) ,
\end{align*}
where we used \eqref{eq:unitarysym} in the last identity.  Now from  \eqref{dq:difffeqorq2}
 we conclude 
 \begin{align*}
 	(S U(s) S^*)^* = V(s)
 \end{align*}
  by uniqueness of the initial value problem, Lemma~\ref{lem:ExistenceAndUniqunessOfIVP}.
 Since $V(s) = U(s)^{-1}$, by \eqref{eq:VU=1} and \eqref{eq:UV=1},  the identity in (c) for antiunitary symmetries is now also  shown.
\end{proof}

Next we shall give a proof of  Lemma  \ref{lem:propofpatprojection} about the eigenprojection of $P_{\rm at}$ stated in the introduction.

\begin{proof}[Proof of Lemma  \ref{lem:propofpatprojection}]
By Hypothesis~\ref{Hypo2}(ii) we can pick $\epsilon > 0$ such that the only point of $\sigma(H_{\rm at}(s_0))$ within
$ \{ z \in \C : |z-E_{\rm at}(s_0)| \leq \epsilon \}$ is  $E_{\rm at}(s_0)$.  Since the
circle $\{ z : |z-E_{\rm at}(s_0)| \}$ is compact and the set
 \begin{equation*}
	\Gamma = \big\{(s,z) : s \in X, \,z
		\in \rho(H_{\rm at}(s))\big\}
 \end{equation*}
 is open (Theorem XII.7 in \cite{ReeSim4}),
 we can find a $\delta > 0$ so that
 $z \in \rho(H_{\rm at}(s))$ if $
  |z-E_{\rm at}(s_0)| = \epsilon$   and $|s - s_0| \leq \delta$.
Thus (a) holds for the set  
\begin{equation*}
	N:= \{s \in X: |s - s_0| \leq \delta\}\,.
\end{equation*}
(b) It follows from (a) that $p_{\rm at}(s)$, defined in \eqref{eq:spectralprojection}, exists for all $s \in N$. By Theorem~\ref{thm:ReeSim4XII5FirstPart} it is a projection. The analyticity of $p_{\rm at}$  on  $N$ now follows from  expression~\eqref{eq:spectralprojection} and Hypothesis~\ref{Hypo2} (i). 
That   $p_{\rm at}(s_0)$ projects onto the eigenspace of $E_{\rm at}(s_0)$, follows from the non-defectivity assumption of Hypthesis II (ii). 
The range of $p_{\rm at}(s_0)$ is finite by assumption. 
The statement  about the dimension of the range of $p_{\rm at}$ follows, since the rank of continuous projection-valued functions of a connected topological space are constant, cf., Lemma  on page 14 in \cite{ReeSim4}. 
\\
(c) Observe that $H_{\rm at}(s)$ leaves the range of $p_{\rm at}(s)$ invariant by Theorem~\ref{thm:ReeSim4XII5}~(a).  First 
we show that there exist a number  $e_{\rm at}(s) $  such that for all $s \in N$ 
\begin{equation} \label{claimttrivialHat} 
	H_{\rm at}(s) \upharpoonright \ran \, p_{\rm at}(s) = e_{\rm at}(s) \upharpoonright \ran \, p_{\rm at}(s)\, .
\end{equation}
In case  $p_{\rm at}(s_0)=1$, we can use  that the dimension of the projection is constant, i.e. ,
$\dim \ran p_{\rm at} (s) = \dim \ran p_{\rm at}(s_0) = 1$. 
 In that case  \eqref{claimttrivialHat} now follows  since $H_{\rm at}(s)$ leaves the range of $p_{\rm at}(s)$ invariant.
  In case  $p_{\rm at}(s_0) > 1$ we will use the symmetry property of Hypothesis~\ref{Hypo2} (iii). 
   Since $\mathcal{S}_1$ is a symmetry of $H_{\rm at}(s)$ it follows from
the integral representation \eqref{eq:spectralprojection}  that it is also
  a symmetry of $p_{\rm at}(s)$.
   By  Theorem \ref{lem:trafofunc} there exists an analytic family 
 $U(s)$ for $s \in N$  of bounded invertible operators    satisfying the assertions   of  Theorem \ref{lem:trafofunc}  for the projection $p_{\rm at}(s)$.
 In particular,
\begin{align} \label{PatUss0} 
p_{\rm at}(s_0)  = U(s)^{-1} p_{\rm at}(s) U(s)  \text { for all } s \in N . 
\end{align} 
%The operator  $H_P(s) := p_{\rm at}(s)H_{\rm at}(s)$
% leaves the range of $p_{\rm at}(s)$ invariant.
 Recall that by Theorem~\ref{thm:ReeSim4XII5}~(a) the operator  $H_{\rm at}(s)$
 leaves the range of $p_{\rm at}(s)$ invariant.
  Thus by \eqref{PatUss0}  the operator
$$
\tilde{H}_{\rm at}(s) := U(s)^{-1} H_{\rm at}(s) U(s)
$$
leaves the range of $p_{\rm at}(s_0)$ invariant. 
By Theorem~\ref{lem:trafofunc}  (c) we have for unitary $S \in \mathcal{S}_1$ that
$$
S \tilde{H}_{\rm at}(s)  S^*  =  S U(s)^{-1} H_{\rm at}(s)  U(s)  S^* = U(s)^{-1}  H_{\rm at}(s)   U(s) = \tilde{H}_{\rm at}(s) ,
$$
and for antiunitary $S  \in \mathcal{S}_1$ that
\begin{align*}
S \tilde{H}_{\rm at}(s)  S^*  &  =  S U(s)^{-1} H_{\rm}(s)  U(s)  S^* = U(s)^*  H_{\rm at}(s)^*  (U(s)^{-1})^* \\
&  =  ( U(s)^{-1} H_{\rm at}(s)   U(s))^* = \tilde{H}_{\rm at}(s)^* .
\end{align*}
Thus by the Lemma of Schur and the irreducibility condition of Hypothesis~\ref{Hypo2}~(iii), there exists a function $e_{\rm at}:N \to \C$ such that
$$
\tilde{H}_{\rm at}(s)  p_{\rm at}(s_0) =  e_{\rm at}(s) p_{\rm at}(s_0) .
$$
By   \eqref{PatUss0} this implies  
$$
H_{\rm at}(s) p_{\rm at}(s) = e_{\rm at}(s) p_{\rm at}(s)  ,
$$
for all $s \in N$, i.e.,  \eqref{claimttrivialHat}. 
Now  the analyticity of $e_{\rm at}(s)$ follows from the analyticity of  $p_{\rm at}(s)$ and $H_{\rm at}(s)$ 
and by  calculating an inner product with a nonzero vector in the range of $p_{\rm at}(s)$.
Furthermore, it follows from (a)  and   Theorem \ref{thm:ReeSim4XII5} (c) that for all $s \in N$ we have 
$$\sigma(H_{\rm at}(s)) \cap 
B_\epsilon(E_{\rm at}(s_0)) =\sigma(H_{\rm at}(s) |_{\ran p_{\rm at}(s)}) . $$ 
This  and \eqref{claimttrivialHat} imply that for  $s \in N$ 
the point $e_{\rm at}(s) \in \C$ is the only
point in the spectrum of   $H_{\rm at}(s)$ in 
$B_\epsilon(E_{\rm at}(s_0))$.  Thus $e_{\rm at}(s)$ is isolated from the rest of the spectrum.
 Furthermore it follows, by  deforming the contour and Cauchy's theorem
 that for $s \in N$ with $r(s) = \varepsilon - E_{\rm at}(s_0) - e_{\rm at}(s)$
 \begin{align*} %\label{eq:spectralprojection1}
	p_{\rm at}(s) = - \frac{1}{2\pi i}
	\ointctrclockwise_{|z-e_{\rm at}(s)| = r(s)}
	\frac{1}{H_{\rm at}(s) - z} \,dz . 
 \end{align*}
Thus  \eqref{claimttrivialHat} implies that 
the number $e_{\rm at}(s)$ is a non-defective, discrete element of the spectrum of $H_{\rm at}(s)$. 
Finally, it follows for $s=s_0$ from the definition of $p_{\rm at}(s)$ and \eqref{claimttrivialHat}
that  $e_{\rm at}(s_0) = E_{\rm at}(s_0)$.
\end{proof}

In Lemma \ref{lem:wlog},  below,   
we show that in the proof of the main theorem, Theorem~\ref{thm:symdegenSpinBoson},
we can assume without loss of generality that the following Hypothesis holds.

\begin{hypr}\label{Hypo5}  
Hypothesis~\ref{Hypo2} holds and $P_{\rm at}(s) = P_{\rm at}(s_0) $ for all $s \in X$.
\end{hypr}
 
%\noindent

\begin{lemma}\label{lem:wlog}   Theorem~\ref{thm:symdegenSpinBoson}  holds, if its  assertion holds under the additional Assumption of  Hypothesis~\ref{Hypo5}.
\end{lemma} 

\begin{proof} 
Suppose  that  Hypotheses~\ref{Hypo1}, \ref{Hypo2}, and \ref{Hypo3} hold for some $s_0 \in X$ and
some symmetry group $\mathcal{S}$.  By restricting  to a smaller neighborhood
of $s_0$ we can assume without loss of generality that $X$  is open, connected, simply connected.
Then by Theorem \ref{lem:trafofunc}  there exists an analytic
family    $U(s)$  of bounded invertible operators on $X$ such that
$$
U(s)  P_{\rm at}(s_0)  U(s)^{-1} = P_{\rm at}(s) .
$$
We now define
\begin{align*}
\hat{H}_g(s) :=    ( U(s)^{-1}  \otimes \one  )  H_g(s) ( U(s) \otimes \one) .
\end{align*}
Then
$$
 \hat{H}_g(s)  =   \hat{H}_{\rm at}(s)  \otimes \one  + \one \otimes H_{\rm f} + g \hat{W}(s) ,
$$
where
\begin{align*}
 \hat{H}_{\rm at} (s) & = U(s)^{-1} H_{\rm at}(s) U(s) , \\
 \hat{W} (s) & = ( U(s)^{-1}  \otimes \one  ) W(s) ( U(s) \otimes \one)  = a( \hat{G}_{1,\overline{s}}) + a^*(\hat{G}_{2,s}) , \\
\hat{G}_{1,\overline{s}} & = U(s)^* G_{1,\overline{s}}\, (U(s)^{-1})^*, \\
 \hat{G}_{2,s} & = U(s)^{-1} G_{2,s}\, U(s).
\end{align*}
Thus if $G_{j,s}$ satisfy Hypothesis~\ref{Hypo1}, then also $\hat{G}_{j,s}$ satisfies Hypothesis~\ref{Hypo1} on any subset $X_0 \subset X$
on which  $U(s)$ and its inverse are  uniformly bounded operator valued functions (by continuity any bounded open $X_0$ with closure contained in $X$ will work).
By analyticity of $U(s)$ it follows that $\hat{H}_{\rm at}(s)$ is an analytic family in the sense of Kato, and hence
part (i) of Hypothesis~\ref{Hypo2} holds. Now $\hat{H}_{\rm at}(s)$  satisfies part (ii)   of Hypothesis~\ref{Hypo2} by  the invertibility of $U(s)$.
Next we consider part (iii) of Hypothesis~\ref{Hypo2}.
Since by assumption $\mathcal{S}_1$ is a symmetry group for $H_{\rm at}(s)$       it follows from the integral representation of $P_{\rm at}(s)$, cf.  \eqref{eq:spectralprojection1},
that it is also a symmetry of the latter. Thus we can assume by  Part (c) of  Theorem \ref{lem:trafofunc}
that  for all symmetries $S \in \mathcal{S}_1$
\begin{align*}
&S U(s) S^* = U(s) , \qquad\qquad \,\text{ if } S \text{ is unitary} , \\
& S U(s) S^* = (U(s)^{-1})^* , \qquad \text{ if } S \text{ is antiunitary} .
\end{align*}
If follows
 for unitary $S \in \mathcal{S}_1$ that
$$
S \hat{H}_{\rm at}(s) S^*  =  S U(s)^{-1} H_{\rm at}(s)   U(s)  S^* = U(s)^{-1} H_{\rm at}(s)    U(s)  = \hat{H}_{\rm at}(s) ,
$$
and for antiunitary $S  \in \mathcal{S}_1$ that
\begin{align*}
S \hat{H}_{\rm at}(s)  S^* &  =  S U(s)^{-1} H_{\rm at}(s)   U(s)  S^* = U(s)^* H_{\rm at}(s)^* (U(s)^{-1})^* \\
&  =  ( U(s)^{-1} H_{\rm at}(s)  U(s))^* = \hat{H}_{\rm at}(s)^* .
\end{align*}
Thus  $ \hat{H}_{\rm at}(s)$ satisfies also Part (iii) of Hypothesis~\ref{Hypo2}. Similarly  one 
shows that $\hat{W}(s)$ satisfies Part (iii) of Hypothesis~\ref{Hypo2}.
Finally,  if $H_{\rm at}(s)$ satisfies Hypothesis~\ref{Hypo3}, then by invertibility of $U(s)$  also $\hat{H}_{\rm at}$ satisfies Hypothesis~\ref{Hypo3} on any subset $X_0 \subset X$
on which  $U(s)$ and its inverse are uniformly bounded  operator valued functions.
Thus we have shown that $\hat{H}_g(s)$ satisfies Hypothesis~\ref{Hypo1}, \ref{Hypo2}, and \ref{Hypo3} on an 
 open set  $X_0$ containing $s_0$.\\
  Furthermore, Hypothesis~\ref{Hypo5} holds for $\hat{H}_g(s)$ by construction. Thus 
by assumption  the assertion of the main result, Theorem  \ref{thm:symdegenSpinBoson},  holds for the operator $\hat{H}_g(s)$.
We conclude that there exists a neighborhood $X_b \subset X_0$ of $s_0$
and a positive constant $g_b$ such that  for all $g \in [0,g_{\rm b})$ and $s \in X_{\rm b}$ 
the operator   $\hat{H}_g(s)$ has an eigenvalue $\hat{E}_g(s)$ with  $\degendim :=  \dim {\rm ker} ( \hat{H}_{\rm at}(s_0) - E_{\rm at}(s_0) )  = \dim {\rm ker} ( {H}_{\rm at}(s_0) - E_{\rm at}(s_0) ) $  linearly independent eigenvectors $\hat{\psi}_{g,j}(s)$, $j=1,...,\degendim$,   all depending analytically on $s \in X_{\rm b}$. 
By the invertibility of $U(s)$ we see that
the operator $H_g(s)$
has the eigenvalue
 ${E}_g(s) := \hat{E}_g(s)$ with  $\degendim$ linearly independent eigenvectors 
   $\psi_{g,j}(s) := ( U(s)  \otimes \one  )  \hat{\psi}_{g,j}(s)  $, $j=1,...,\degendim$.
They  also depend analyticaly on $s$, since $U(s)$ and its inverse depend  by  Theorem \ref{lem:trafofunc}
analytically on $s$. This shows (i) of  Theorem  \ref{thm:symdegenSpinBoson}. Similarly one verifies (ii) of Theorem  \ref{thm:symdegenSpinBoson}
by using the uniform boundedness  of $U(s)$ and $U(s)^{-1}$. 
Finally, suppose that the operator $H_g(s)$ satisfies  Hypothesis \ref{Hypo4}.
Then  by Theorem \ref{lem:trafofunc} (b)  we can choose
the family of invertible operators $U(s)$ to be unitary for real $s$ such that $U(\overline{s})^* = U(s)^{-1}$ for all $s \in X$.  Thus also $\hat{H}_g(s)$ satisfies  Hypothesis \ref{Hypo4}
and moreover it is  isospectral to $H_g(s)$ for real $s$. In that  case we have for real $s \in \R^\nu \cap X_{\rm b}$
that
$$
E_g(s) = \hat{E}_g(s) = \inf \sigma(\hat{H}_g(s)) = \inf \sigma({H}_g(s)) .
$$
This implies (iii)  of  Theorem  \ref{thm:symdegenSpinBoson}.\\
Thus we have shown  that the assertion of  Theorem  \ref{thm:symdegenSpinBoson} also holds for the original operator $H_g(s)$.
\end{proof}

The next  lemma will be used to show that the so called  relevant direction in the
renormalization analysis is one dimensional.
For this,  let us introduce the following definition.
 For  $V$  a finite dimensional complex vector space and a bounded operator  $T \in \mathcal{B}(V \otimes \FF)$
define
$
\langle  T  \rangle_\Omega
$
as the unique operator on $V$ such  that
\begin{align} \label{defofAOmega} 
\inn{ v_1 , \langle T \rangle_\Omega  v_2 }  =  \inn{ v_1 \otimes \Omega , T  v_2 \otimes \Omega } .
\end{align} 
for all $v_1, v_2 \in V$. Note that it is straight forward to see that 
\begin{equation}  \label{eq:identadjvacexp}
\langle T^*  \rangle_\Omega  =  \langle T  \rangle_\Omega^* ,
\end{equation}
which follows since for all $v_1, v_2 \in V$ we have
\begin{align*}
\inn{ v_1 , \langle T^*  \rangle_\Omega v_2 } & =
  \inn{ v_1 \otimes \Omega , T^*   v_2 \otimes \Omega }   =   \overline{ \inn{  v_2 \otimes \Omega,   T v_1 \otimes \Omega    } } \\
& =
  \overline{\inn{ v_2 , \langle T   \rangle_\Omega v_1 }}  =
 {\inn{ v_1 , \langle T  \rangle_\Omega^* v_2 }}  .
\end{align*}

\begin{lemma} \label{lem:1-dimOp}
Let $V$ be a finite dimensional complex vector space and let $T \in \mathcal{B}(V \otimes \FF)$.
Assume that $T$ is symmetric with respect to a  set of  symmetries  $\mathcal{S}$ such that every element  can be written in the form $S_1 \otimes S_2$, where  $S_1$ is a symmetry in $V$ and $S_2$ is a symmetry in $\FF$ leaving the Fock vacuum invariant. Assume that $\mathcal{S}_1 := \{ S_1 : S_1 \otimes S_2 \in \mathcal{S} \}$ acts irreducibly on $V$. 
Then there exists a number $c \in \C$ such that 
$$
\langle  T  \rangle_\Omega   = c \one
$$
\end{lemma}
\begin{proof}
For all $S_1 \otimes S_2 \in \mathcal{S}$ we have
the following symmetry property. For  $A$ an operator or  a number let  $A^\#$ stand for $A$ or  $A^{*}$  whether the symmetry $S_1 \otimes S_2 $   is unitary or antiunitary, respectively.   Moreover, we  write $c^* = \overline{c}$  if $c \in \C$.   
For all $v_1, v_2 \in V$ we have
\begin{align*}
\inn{ v_1 ,  S_1 \langle  T  \rangle_\Omega  S_1^* v_2 }
 &  =   \inn{S_1^*  v_1  ,  \langle  T  \rangle_\Omega  S_1^* v_2 }^\#   \\
 &  =   \inn{S_1^*  v_1 \otimes \Omega  ,     T     S_1^* v_2 \otimes \Omega   }^\#   \\
 &  =   \inn{(S_1 \otimes S_2)^*  v_1 \otimes \Omega  ,     T    (  S_1\otimes S_2)^* v_2 \otimes \Omega   }^\#   \\
 &  =   \inn{ v_1 \otimes \Omega  ,  (S_1 \otimes S_2 )   T    (  S_1\otimes S_2)^* v_2 \otimes \Omega   } \\
 &  =   \inn{ v_1 \otimes \Omega  ,      T^\#     v_2 \otimes \Omega   } \\
& = \inn{ v_1 , \langle T^\#  \rangle_\Omega  v_2 } \\
& = \inn{ v_1 , \langle T  \rangle_\Omega^\#  v_2 } ,
\end{align*}
where in the last line we used \eqref{eq:identadjvacexp}.
Thus
$$
  S_1 \langle  T  \rangle_\Omega  S_1^*  =  \langle T  \rangle_\Omega^\#  .
$$
The claim now follows from  Schur's Lemma  \ref{lem:seclemschur} and the irreduciblity assumption.
\end{proof}

To conclude this section we show that the Feshbach transformation preserves symmetry properties. A detailed review of the properties of the Feshbach-Schur map, which was introduced in \cite{BCFS} is given in Appendix~\ref{app:Feshbach}.

\begin{lemma} \label{lem:FeshbachpreservesSymmetry}
Let $(H,T)$ be a Feshbach pair for $\chi$. Assume that there exists
a group of symmetries $\mathcal{S}$ of the  operator $H, T$ and $\chi$.
Then $\mathcal{S}$ is also a group of symmetries for the  Feshbach operator $F_\chi(H,T)$.
\end{lemma}
\begin{proof}
This follows from the definition of the Feshbach operator given in Eq.~\eqref{eq:feshbachmap}. Let
  $S \in \mathcal{S}$ be a symmetry and let  $A^\#$ stands for $A$ or  $A^{*}$ if
$S$ is unitary or antiunitary, respectively.  Then inserting $S^* S = \one$, we find
\begin{align*}
	&  S F_{\chi}(H,T) S^* \\ & = S H_\chi S^* - S \chi W \overline{\chi}
		(( T + \overline{\chi} W \overline{\chi})|_{{\ran \overline{\chi}}})^{-1} \overline{\chi} W \chi S^*  \\
		& = S ( T + \chi S^* S W S^* S  \chi ) S^* \\
		&  - S \chi  S^* S W S^* S \overline{\chi} S^*
		(( S T S^* +  S \overline{\chi} S^* S  W  S^* S \overline{\chi}) S^*|_{{\ran S \overline{\chi}}S^*})^{-1}   S \overline{\chi} S^* S  W S^* S \chi S^*  \\
 & = %T^\# + {\chi}^\# W^\# {\chi}^\# 
 	H^\#_{\chi^\#} -  \chi^\# W^\# \overline{\chi}^\#
		(( T^\# + \overline{\chi}^\# W^\# \overline{\chi}^\#)|_{{\ran \overline{\chi}^\#}})^{-1} \overline{\chi}^\# W^\# \chi^\#   \\
& = F_\chi(H,T)^\# \,. \tag*{\qedhere}
\end{align*}
\end{proof}

%%%%%%%%%%%%%%%%%%%%%%%%%%%%%%%%%%%%%%%%%%%%%%%%%%%%%
%%%%%%%%%%%%% initial Feshbach step  %%%%%%%%%%%%%%%%
%%%%%%%%%%%%%%%%%%%%%%%%%%%%%%%%%%%%%%%%%%%%%%%%%%%%%
\section{The initial Hamiltonian}
\label{sec:initialFeshbach}
The first step of  the operator-theoretic
renormalization analysis is
to prove that $H_g(s)$ and $H_0(s)$ are a Feshbach
pair for a suitable choice for the projection operator, see \eqref{defofboldchi} below.
This is the content of Theorem \ref{thm:Feshbachpairchp41}. 
For a definition as well as the properties of Feshbach pairs we refer to Appendix~\ref{app:Feshbach}.
Moreover we will show in this section,  that the associated Feshbach operator,
cf.   \eqref{eq:feshbachmap}, is an analytic function of  $s$ and the spectral parameter $z$ and that it inherits the symmetry property 
of the original operator.  This will be shown in Theorem \ref{thm:feshbachpair}.

We choose smooth functions
$\chi, \overline{\chi} \in C^\infty(\R ; [0,1] )$ such that
$\chi^2 + \overline{\chi}^2 = 1$ and
\begin{equation*}
\chi(r) =
\begin{cases} 1\,, \quad \textrm{ if } r \leq \frac{3}{4}\,, \\
		  0\,, \quad \textrm{ if }r \geq 1\,.
\end{cases}
\end{equation*}
For $\rho > 0$ we then define
\begin{equation*}
\chi_\rho(r) := \chi(r / \rho)
%\equiv \chi_{r \leq \rho}
\, , \qquad
\overline{\chi}_\rho(r) := \overline{\chi}(r / \rho)
%\equiv \overline{\chi}_{r \leq \rho}
\, ,
\end{equation*}
and set $\chi_\rho := \chi(H_{\rm f} / \rho)$,
$\overline{\chi}_\rho := \overline{\chi}(H_{\rm f} / \rho)$.
Next we define 
\begin{align}
	& \boldsymbol{\chi}_\rho(s) := P_{\rm at}(s) \otimes
		\chi_\rho\,,  \label{eq:chibarsymchp4-1}  \\
	& \overline{\boldsymbol{\chi}}_\rho(s)
		:= \overline{P}_{\rm at}(s) \otimes \one
			+ P_{\rm at}(s) \otimes \overline{\chi}_\rho
			\,. \label{eq:chibarsymchp4}
\end{align}
Note that   \eqref{eq:chibarsymchp4-1}  and   \eqref{eq:chibarsymchp4}  are   commuting, non-zero, bounded operators    satisfying
$
\overline{\boldsymbol{\chi}}_\rho(s)^2 + \boldsymbol{\chi}_\rho(s)^2 = 1 , 
$
which   are  not necessarily self-adjoint.
Moreover we set
\begin{align} \label{defofboldchi} 
\boldsymbol{\chi}(s) := \boldsymbol{\chi}_{1}(s )\,, \quad 
\overline{\boldsymbol{\chi}}(s) := \overline{\boldsymbol{\chi}}_{1}(s).
\end{align}

The following theorem gives us the conditions for which we can
define the so called  first Feshbach operator.

\begin{proposition}
\label{thm:Feshbachpairchp41}
 Suppose Hypothesis~\ref{Hypo1}, \ref{Hypo2},  and \ref{Hypo3}  hold, and let
$\mathcal{U} \subset X_1 \times \C$ be given by
Hypothesis~\ref{Hypo3}.
Then there is a $g_{\rm b}  > 0 $ such that for all
$g \in [0,g_{\rm b})$
and all $(s,z) \in \mathcal{U}$, the pair
$(H_g(s)-z,H_0(s)-z)$
is a Feshbach pair for $\boldsymbol{\chi}(s)$.
Furthermore one has the absolutely convergent 
expansion on $\mathcal{U}$ 
\begin{align} \label{eq:absconvExpansion}
F_{\boldsymbol{\chi}(s)}&( H_g(s) -z , H_0(s) - z )
%\upharpoonright P_{\rm at}(s) \otimes \FF 
\\
	&= E_{\rm at}(s) - z  + H_{\rm f} \nonumber\\
		&\quad+  \sum_{L=1}^\infty(-1)^{L-1}
	\boldsymbol{\chi}(s)\, g\, W(s) 	\overline{\boldsymbol{\chi}}(s) \big(H_0(s) - z\big)^{-1} \nonumber  \\
	& \times  \left(
 \, g \,\overline{\boldsymbol{\chi}}(s)
	\,W(s)
	\overline{\boldsymbol{\chi}}(s) 	\big(H_0(s) - z\big)^{-1}
	\right)^{L-1} 	\overline{\boldsymbol{\chi}}(s)  W(s) 
	\boldsymbol{\chi}(s)\,. \nonumber
\end{align}
\end{proposition}

For  the proof of this  proposition  we make
use the following  lemma.

\begin{lemma} \label{lem:H0-boundschp4}
Suppose Hypothesis~\ref{Hypo2} and \ref{Hypo3} hold. 
Then 
	\begin{align} \label{eq:H0-zboundedinvert}
		\sup_{(s,z) \in \mathcal{U}}
		\big\|(H_0(s) - z)^{-1}
			\boldsymbol{\chi}(s)\big\|
		< \infty \,,
	\end{align}
	and
	\begin{equation} \label{eq:Hf+1H0-zchibound}
		\sup_{(s,z) \in \mathcal{U}}
		\big\|(H_{\rm f} +1)(H_0(s) -z)^{-1}
		\overline{\boldsymbol{\chi}}(s) \big\|
		< \infty \, .
	\end{equation}
\end{lemma}

\begin{proof} We recall that  by definition, cf.  Eq.~\eqref{eq:chibarsymchp4} and   \eqref{defofboldchi} ,  
$\overline{\boldsymbol{\chi}}(s)
	=  \overline{P}_{\rm at}(s) \otimes \one + P_{\rm at}(s) \otimes \overline{\chi}(H_{\rm f} )$.
First we estimate \eqref{eq:Hf+1H0-zchibound}. 
Applying the triangle inequality we obtain
\begin{align}
& \left\|(H_{\rm f} + 1 ) (H_{0}(s) - z)^{-1}  \overline{\boldsymbol{\chi}}(s) \right\| \nonumber  \\
& \leq    \left\| (H_{\rm f}+1) (H_{0}(s) - z)^{-1}  \overline{P}_{\rm at}(s) \otimes \one \right\|  \label{eq:Lemma4.2SecondEqPart1} \\ 
&\quad + 
\| (H_{\rm f}+1)  (H_{0}(s) - z)^{-1}    P_{\rm at}(s) \otimes \overline{\chi}(H_{\rm f}) \|  \,. \label{eq:Lemma4.2SecondEqPart2}
\end{align}
We estimate \eqref{eq:Lemma4.2SecondEqPart2}  by the spectral theorem and find 
\begin{align*} 
& \| (H_{\rm f} + 1 )  (H_{0}(s) - z)^{-1}    P_{\rm at}(s) \otimes \overline{\chi}(H_{\rm f}) \|  \\
& 
= \sup_{r \geq 0} \| (r+1)  (E_{\rm at}(s)  + r  - z)^{-1}    P_{\rm at}(s)  \otimes \overline{\chi}(r ) \|  \\
& \leq  \sup_{r \geq 3/4} \left| \frac{r+1}{ E_{\rm at}(s)  + r  - z } \right|  \|    P_{\rm at}(s)  \|  \\
& \leq  \sup_{r \geq 3/4} \left| 1 +  \frac{1 - E_{\rm at}(s) + z} { E_{\rm at}(s)  + r  - z } \right|  \|    P_{\rm at}(s)  \|    \\
& \leq \left(  1 + (1+|E_{\rm at}(s) - z |)  \sup_{r \geq 3/4} \frac{1}{  | r  - |E_{\rm at}(s) - z| | }  \right) \|    P_{\rm at}(s)  \|  \\
& \leq \left( 1 + \frac{3}{2}  \cdot \frac{1}{\frac{3}{4}-\frac{1}{2}} \right)  \|    P_{\rm at}(s)  \| = 7 \|   P_{\rm at}(s)  \| , %\label{reswithIIIboundPat} 
\end{align*} 
where the right hand side is finite by Hypothesis \ref{Hypo3}.
To estimate \eqref{eq:Lemma4.2SecondEqPart1} we use again the spectral theorem and find 
\begin{align*} 
& \left\| (H_{\rm f} + 1 )  (H_{0}(s) - z)^{-1}  \overline{P}_{\rm at}(s) \otimes \one \right\|   \\
& 
\leq \sup_{r \geq 0} \|(r+1)   (H_{\rm at}(s)  + r  - z)^{-1}    \overline{P}_{\rm at}(s)  \|  < \infty  %\label{reswithIIIboundPatbar} 
\end{align*} 
where the last bound follows from Hypothesis \ref{Hypo3}. 
This shows \eqref{eq:Hf+1H0-zchibound}. 

Next we similarly show \eqref{eq:H0-zboundedinvert}. Using   the triangle inequality, we find
\begin{align}
&\left\| (H_{0}(s) - z)^{-1}
	 \overline{\boldsymbol{\chi}}(s) \right\|  \nonumber \\
	 &\leq    \left\|  (H_{0}(s) - z)^{-1}  \overline{P}_{\rm at}(s) \otimes \one \right\|  +  \|  (H_{0}(s) - z)^{-1}    P_{\rm at}(s) \otimes \overline{\chi}(H_{\rm f}) \|  \,. \label{resboundIIIccc} 
\end{align}
We obtain for the second term in   \eqref{resboundIIIccc}  by the spectral theorem  
\begin{align*} 
& \|  (H_{0}(s) - z)^{-1}    P_{\rm at}(s) \otimes \overline{\chi}(H_{\rm f}) \| \\
& 
= \sup_{r \geq 0} \|  (E_{\rm at}(s)  + r  - z)^{-1}    P_{\rm at}(s)\otimes \overline{\chi}(r ) \|  \\
& \leq  \sup_{r \geq 3/4} |  (E_{\rm at}(s)  + r  - z)^{-1} |  \|    P_{\rm at}(s)  \|   \\
& \leq  \sup_{r \geq 3/4} |  ( r  - |E_{\rm at}(s) - z|)^{-1} |  \|    P_{\rm at}(s)  \|   \\
& \leq \frac{1}{3/4-1/2}  \|    P_{\rm at}(s)  \| = 4 \|   P_{\rm at}(s)  \| , %\label{reswithIIIbound2} 
\end{align*} 
where the right hand side is again  finite by Hypothesis \ref{Hypo3}.
To estimate the first term  in   \eqref{resboundIIIccc}  we use again the spectral theorem and find from Hypothesis~\ref{Hypo3}
\begin{align*} 
 \left\|  (H_{0}(s) - z)^{-1}  \overline{P}_{\rm at}(s) \otimes \one \right\|   
\leq \sup_{r \geq 0} \|  (H_{\rm at}(s)  + r  - z)^{-1}    \overline{P}_{\rm at}(s)  \| < \infty \,. %\label{reswithIIIbound1} 
\end{align*} 
This completes the proof. 
\end{proof} 

\begin{lemma} Let   Hypothesis~\ref{Hypo1} hold. Then  
	\begin{align}\label{eq:WHfboundmu}
		%\sup_{s \in X}
		\big\|W(s)\,(H_{\rm f} + 1) ^{-1/2}\big\|
			&\leq 2\, \max_{j=1,2} \sup_{(s,z)\in \mathcal{U}} \|G_{j,s}\|_\mu < \infty
			\,, \\
		%\sup_{s \in X}
		\big\|(H_{\rm f} + 1) ^{-1/2}W(s)\big\|
			&\leq 2\,\max_{j=1,2} \sup_{(s,z)\in \mathcal{U}} \|G_{j,s}\|_\mu  
			< \infty \,. \label{eq:WHfboundmu1}
	\end{align}
\end{lemma} 
\begin{proof}  This follows  from Eq.~\eqref{eq:estoncrea} in 
Appendix~\ref{sec:tecAux}  and Hypothesis~\ref{Hypo1}.
\end{proof} 

\begin{lemma}  \label{arelboundforneumann} 
Suppose Hypothesis~\ref{Hypo1},  \ref{Hypo2}, and \ref{Hypo3} hold. Then 
\begin{align*}%\label{eq:fesh3}
&\sup_{(s,z) \in \mathcal{U}} \left\| g \overline{\boldsymbol{\chi}}(s) W(s)(H_0(s)-z)^{-1}\overline{\boldsymbol{\chi}}(s)\right\|< \infty ,\\
& \sup_{(s,z) \in \mathcal{U}} \left\|( H_0(s) - z)^{-1}\overline{\boldsymbol{\chi}}(s) g
W(s)\overline{\boldsymbol{\chi}}(s)\right\|
 <  \infty . 
\end{align*}
\end{lemma}
\begin{proof} Follows from Lemma \ref{lem:H0-boundschp4}, and \eqref{eq:WHfboundmu} respective \eqref{eq:WHfboundmu1}.  
\end{proof}

Now we are ready to prove  Proposition~\ref{thm:Feshbachpairchp41}.
We will use the following notation. 

\begin{proof}[Proof of Proposition~\ref{thm:Feshbachpairchp41}] %\cite[Theorem 5.5]{Wei80}
 Let $\mathcal{U} \subset X \times \C$ be given by Hypothesis~\ref{Hypo3}.
First we show the Feshbach property. For this  we  need to  show that $H_g(s)$ and $H_0(s)$ are closed operators on the same domain
such that  the assumptions (a'), (b') and (c') of  Lemma \ref{lem:FeshbachCriteria} hold. 

Suppose $(s,z) \in \mathcal{U}$. 
To prove that $H_g(s) = H_0(s) + g W(s)$ is closed on $D(H_0(s))$ for all $g  > 0$ it suffices to  prove that $W(s)$ is infinitesimally 
bounded with respect to $H_0(s)$, cf.  \cite[Theorem 5.5]{Wei80}.

%%%%%%%%%%%%%%%%%%%%%%%%%%
%%%%%%%%%%%%%%%%%%%%%%%%%%

Note 
that $H_{\rm at}(s)$ leaves the ranges of $P_{\rm at}(s)$ and $\overline{P}_{\rm at}(s)$ invariant, cf.  Theorem \ref{thm:ReeSim4XII5}.
Thus  by the spectral theorem $H_0(s)$ leaves the range of  $P_{\rm at}(s) \otimes \one$ invariant. Moreover 
 for $w = z - 1$ we have
$w \in \rho( H_0(s)  |_{\ran P_{\rm at}(s) \otimes D(H_{\rm f})})$, since $\sup_{r \geq 0}|E_{\rm at}(s) - w + r |^{-1} 
\leq \sup_{r \geq 0}(1 -| E_{\rm at}(s) - z |   + r )^{-1} \leq 2$,  and  
\begin{align}\label{eq:fesh20-1}
 &  \|(H_{\rm f}+1) (H_0(s) - w)^{-1} P_{\rm at}(s) \otimes \one |  \nonumber \\
 & \leq \sup_{r \geq 0} \|\frac{r+1}{E_{\rm at}(s) - z + 1  + r } P_{\rm at}(s) \otimes \one\| \nonumber \\
  & \leq (1 + \sup_{r \geq 0} \frac{ |E_{\rm at}(s) - z| }{|E_{\rm at}(s) - z + 1  + r |}) \| P_{\rm at}(s)\| \nonumber \\
 & \leq 2  \| P_{\rm at}(s)\| < \infty  , 
\end{align}
where we used that  by Hypothesis \ref{Hypo3} we have $|E_{\rm at}(s) - z| <  1/2$ and the last inquality of  \eqref{eq:fesh20-1}. 
On the other hand by the spectral theorem  and Hypothesis \ref{Hypo3}  we find 
 \begin{align} 
& \left\| (H_{\rm f} + 1 )  (H_{0}(s) - w)^{-1}  \overline{P}_{\rm at}(s) \otimes \one \right\|   \nonumber \\
& 
\leq \sup_{r \geq 0} \|(r+1)   (H_{\rm at}(s)  + r  - w)^{-1}    \overline{P}_{\rm at}(s)  \|  \nonumber \\
& 
=  \sup_{r' \geq 1} \| r'  (H_{\rm at}(s)  + r' - z )^{-1}    \overline{P}_{\rm at}(s)  \|  \nonumber \\
& 
\leq   \sup_{r' \geq 1} \| ( r'  + 1 ) (H_{\rm at}(s)  + r' - z )^{-1}    \overline{P}_{\rm at}(s)  \| \nonumber  \\
& 
\leq   \sup_{r \geq 0} \| (r+1)  (H_{\rm at}(s)  + r - z )^{-1}    \overline{P}_{\rm at}(s)  \|  < \infty \label{reswithIIIboundPatbarnew} 
\end{align} 
In particular, for normalized $\varphi \in D(H_{\rm at}(s)) \otimes D(H_{\rm f})$ we obtain  using the triangle inequality 
together with  \eqref{eq:fesh20-1}   and \eqref{reswithIIIboundPatbarnew} 
\begin{align}\label{eq:fesh2}
 & \|( H_{\rm f}  + 1)(H_0(s) - w)^{-1} \varphi\|  \nonumber \\
%&= \|(H_{\rm f} + 1 )  (H_0(s) - w)^{-1} (P_{\rm at}(s) \otimes \one + \overline{P}_{\rm at}(s) \otimes \one) \varphi\| \nonumber \\
 & \leq \|(H_{\rm f} + 1) (H_0(s) - w)^{-1} P_{\rm at}(s) \otimes \one \varphi\|
 \|(H_{\rm f}+ 1)  (H_0(s) - w)^{-1} \overline{P}_{\rm at}(s) \otimes \one \varphi\| \nonumber \\ 
 & \leq 2  \| P_{\rm at}(s)\|  +  \sup_{r \geq 0} \| (r+1)  (H_{\rm at}(s)  + r - z )^{-1}    \overline{P}_{\rm at}(s)  \| .
\end{align}
Combining \eqref{eq:WHfboundmu} and \eqref{eq:fesh2} we see that, for
all $\phi\in D(H_{\rm at}(s))\otimes D(H_{\rm f})$ and $\epsilon > 0$ 
\begin{eqnarray*}
  \|W(s)\phi\|^2 &\leq & C_0 \langle \phi, (H_{\rm f}+1)\phi \rangle \\
  &=& C_0 \langle\phi, (H_{\rm f}+1)(H_{0}(s)-w)^{-1}(H_{0}(s)-w)\phi \rangle\\
  &\leq & C_1 \|\phi\| \|H_0(s)\phi\| + C_2 \|\phi\|^2\\
  & \leq & C_1 \epsilon \|H_0(s)\phi\|^2 + \left(\frac{C_1}{\epsilon}+C_2\right)\|\phi\|^2
\end{eqnarray*}
with constants $C_0,C_1,C_2$. This shows that $W(s)$ is infinitesimally bounded with respect to $H_0(s)$ and thus 
we have shown  that $H_g(s) = H_0(s) + g W(s)$ is closed on $D(H_0(s))$ for all $g  > 0$.

Next we verify the criteria for Feshbach pairs from
Lemma~\ref{lem:FeshbachCriteria}. 	On $D(H_{0}(s))$ we have by definition
	\begin{align*}%\label{eq:xH0=H0x}
		\boldsymbol{\chi}(s)H_0(s)
			= H_0(s) \boldsymbol{\chi}(s)
				\quad \textrm{ and } \quad
				\overline{\boldsymbol{\chi}}(s)H_0(s)
			= H_0(s)\overline{\boldsymbol{\chi}}(s) \, ,
	\end{align*}
Since this is valid on every core of $H_0(s)$, we get that Condition~$(a')$ of that Lemma~\ref{lem:FeshbachCriteria} is satisfied.
By Lemma \ref{lem:H0-boundschp4}, $H_0(s)-z$ is bounded invertible on
$\ran\overline{\boldsymbol{\chi}}(s)$. 
Moreover, by  Lemma  \ref{arelboundforneumann} we get that there exists a $g_b > 0$ such that 
\begin{align*}%\label{eq:fesh4}
&\sup_{(s,z) \in U} \left\| g \overline{\boldsymbol{\chi}}(s) W(s)(H_0(s)-z)^{-1}\overline{\boldsymbol{\chi}}(s)\right\|< 1,\\
& \sup_{(s,z) \in U} \left\|( H_0(s) - z)^{-1}\overline{\boldsymbol{\chi}}(s) g
W(s)\overline{\boldsymbol{\chi}}(s)\right\|
 < 1,\nonumber
\end{align*}
for  all $g \in [0, g_b)$.  This proves (b') and (c') of
Lemma~\ref{lem:FeshbachCriteria} and hence completes the proof that
$(H_g(s)-z, H_0(s) - z)$ is a Feshbach pair for $\boldsymbol{\chi}(s)$.
By choosing $g_b > 0$ sufficiently small it follows that the Neumann series 
\begin{align*}%\label{eq:fesh5}
& (H_g(s) - z)_{\overline{\boldsymbol{\chi}}(s)}^{-1} \vert_{\ran\overline{\boldsymbol{\chi}}(s)}  \\
& = (H_0(s) - z )^{-1}
\sum_{n=0}^\infty \left(-\overline{\boldsymbol{\chi}}(s) g W(s) ( H_0(s) - z)^{-1}\overline{\boldsymbol{\chi}}(s)\right)^n
\big\vert_{\ran\overline{\boldsymbol{\chi}}(s)} 
\end{align*}
converges uniformly for $(s,z)\in \UU$.
\end{proof}

\begin{remark}
We note that if Hypothesis~\ref{Hypo2} holds, then it is straight forward to see using \eqref{eq:spectralprojection1}
that $\boldsymbol{\chi}_\rho$
and $\overline{\boldsymbol{\chi}}_\rho$  commute with the group of symmetries $\mathcal{S}$ given by Hypothesis~\ref{Hypo2} (iii).
\end{remark}

Provided the right hand side exists, i.e. the Feshbach pair property holds,  
cf. 
Proposition~\ref{thm:Feshbachpairchp41}, we define the so called first Feshbach operator 
\begin{align} \label{eq:absconvExpansion0} \tilde{H}_g^{(0)}[s,z]
&  := F_{\boldsymbol{\chi}(s)}( H_g(s) -z , H_0(s) - z ) \\
	&=  H_{\rm at}(s) - z  + H_{\rm f}    +  \tilde{W}^{(0)}_{g}[s,z] )   \nonumber
\end{align}
where 
\begin{align} \label{eq:absconvExpansion0-1}
& \tilde{W}_g^{(0)}[s,z]\\
&  := 
  \sum_{L=1}^\infty(-1)^{L-1}
	\boldsymbol{\chi}(s)\, g\, W(s)  	\overline{\boldsymbol{\chi}}(s)  \big(H_0(s) - z\big)^{-1} \nonumber \\
	& \times \left(
	 \, g \,\overline{\boldsymbol{\chi}}(s)
	\,W(s)
	\overline{\boldsymbol{\chi}}(s) \big(H_0(s) - z\big)^{-1}
	\right)^{L-1}	\overline{\boldsymbol{\chi}}(s) W(s)
	\boldsymbol{\chi}(s)     \,.  \nonumber
\end{align}
Note that by the choice of the projection $\boldsymbol{\chi}(s)$ it follows that \eqref{eq:absconvExpansion0}
and  \eqref{eq:absconvExpansion0-1} leave the range of $P_{\rm at}(s) \otimes 1_{H_{\rm f} \leq 1} $ invariant. 
Furthermore, we  define the following restrictions, which are  for the isospectrality property 
 sufficient to study, cf.   Theorem  \ref{thm:isospectralFeshbach},
\begin{align}
& {H}_g^{(0)}[s,z] := \tilde{H}_g^{(0)}[s,z] \upharpoonright \ran  ( P_{\rm at}(s) \otimes 1_{H_{\rm f} \leq 1} )   \label{eq:defh0}.   \\
& {W}_g^{(0)}[s,z] := \tilde{W}_g^{(0)}[s,z] \upharpoonright \ran ( P_{\rm at}(s) \otimes 1_{H_{\rm f} \leq 1}  )  . \label{eq:absconvExpansion00}
\end{align} 
Note that  as operators acting on the range of $P_{\rm at}(s) \otimes 1_{H_{\rm f} \leq 1}$ we have 
\begin{align}
\label{dofofhzero}  {H}_g^{(0)}[s,z]  =  E_{\rm at}(s) - z + H_{\rm f} +   {W}_g^{(0)}[s,z] . 
\end{align} 
We shall refer to \eqref{dofofhzero} as the first Feshbach operator as well.  
Henceforth we  shall assume Hypothesis \ref{Hypo5} and so $H_g^{(0)}(s,z)$ acts on the Hilbert space 
\begin{align*}
\HH_{\rm red} & := \ran P_{\rm at}(s_0) \otimes \ran 1_{H_{\rm f} \leq 1}  = \ran P_{\rm at}(s) \otimes \ran 1_{H_{\rm f} \leq 1}
\end{align*}
 
\begin{remark}
Note that the  notation introduced in   \eqref{eq:absconvExpansion0}  -    \eqref{eq:absconvExpansion00}  is  similar to the one in~\cite{GriHas09} but not exactly the same.
\end{remark}

In the following theorem we show that the first Feshbach operator ${H}_g^{(0)}[s,z]$
is analytic on a suitable subset of $X \times \C$.
Moreover we show that this operator  
is isospectral to $H_g(s) - z$, in the sense of  Theorem  \ref{thm:isospectralFeshbach}.
Furthermore the first Feshbach operator  commutes with the set
of symmetries $\mathcal{S}$ from Hypothesis~\ref{Hypo2}.
Note that in the  theorem below we make use of the auxiliary operator $Q_{\chi}$
defined in Eq.~\eqref{eq:auxiliaryOp}.

%%%%%%%%%%%%%%%%%%%%%%%%%%%%%%%%%%%%%%%%%%%%%%%%%
%the isospectral operator (analyticity)
		%[Feshbach pair, Spectrum]
%%%%%%%%%%%%%%%%%%%%%%%%%%%%%%%%%%%%%%%%%%%%%%%%%

\begin{theorem}
\label{thm:feshbachpair}
 Suppose Hypothesis~\ref{Hypo1}, \ref{Hypo2},  and \ref{Hypo3}  hold, and let
$\mathcal{U} \subset X_1 \times \C$ be given by
Hypothesis~\ref{Hypo3}.
Then there is a $g_{\rm b}  > 0 $ such that for all
$g \in [0,g_{\rm b})$
and all $(s,z) \in \mathcal{U}$, the pair
$(H_g(s)-z,H_0(s)-z)$
is a Feshbach pair for $\boldsymbol{\chi}(s)$ and the following holds on $\mathcal{U}$. 
\begin{itemize}
\item[(a)] The map $(s,z) \mapsto H_g^{(0)}[s,z]$ is analytic. The map $(s,z) \mapsto Q_{\boldsymbol{\chi}}(s,z)$ is analytic. 
	\item[(b)] $H_g(s)-z : D(H_0(s))
		\subset \HH \to \HH$
		is bounded invertible if and only  if
		${H}_g^{(0)}[s,z]$ is bounded invertible. % on $\HH_{\rm red}$. 
		%$\ran  P_{\rm at}(s_0) \otimes
		%$\ran P_{\rm at}(s_0) \otimes \HH_{\rm red}$. \marginpar{???}
		%$\ran\, \boldsymbol{\chi}(s_0)$ {\tt or larger space ???} 
	\item[(c)] The following maps are linear isomorphisms
		and inverses of each other:
		\begin{align*}
			&\boldsymbol{\chi}(s) :
				\ker\,(H_g(s) - z )
				\to \ker\, {H}_g^{(0)}[s,z]\, , \\
			&Q_{\boldsymbol{\chi}}(s,z)  :
			\ker\,  {H}_g^{(0)}[s,z] \to \ker\,
			(H_g(s) - z )\,.
		\end{align*}
\end{itemize}
Furthermore, let $\mathcal{S}$ be the set of
symmetries given in Hypothesis~\ref{Hypo2}, then
\begin{itemize}
	\item[(d)] $S  {H}_g^{(0)}[s,z] \, S^* = {H}_g^{(0)}[s,z]$
	, \quad
	for all unitary $S \in \mathcal{S}$.
	\item[(e)] $S {H}_g^{(0)}[s,z] \, S^* = \left( H_g^{(0)}[s,z]\right)^*$
	\!\!, \quad
	for all antiunitary $S \in \mathcal{S}$.
\end{itemize}
In addition, if Hypothesis~\ref{Hypo4} is valid, we have for
 $(s,z) \in \mathcal{U} \cap \mathcal{U}^*$  that 
%For $s$ in a domain $X$ of the complex plane symmetric
%with respect to the real axis we have
\begin{itemize}
	\item[(f)]
	${H}_g^{(0)}[s,z]^*
		=   {H}_g^{(0)}[\overline{s}, \overline{z}]\,.$
\end{itemize}
\end{theorem} 

\begin{lemma}\label{analemW} Let Hypothesis \ref{Hypo1} hold. Then the  mapping $s \mapsto W(s) (H_{\rm f} + 1 )^{-1/2}$ is analytic on $X$. 
\end{lemma} 
\begin{proof} \cite[Lemma~12]{GriHas09}
%See Hasler Griesemer Lemma 12.
\end{proof}

\begin{proof}[Proof of Theorem~\ref{thm:feshbachpair}]
Let $g_{\rm b} > 0$ be such that the assertion of 
 Proposition~\ref{thm:Feshbachpairchp41} holds. Then the  Feshbach pair property holds by  Proposition~\ref{thm:Feshbachpairchp41}.  \\
(a) From \eqref{dofofhzero} and the analyticity of  $s \mapsto E_{\rm at}(s)$, 
the analyticity of $(s,z) \mapsto H_g^{(0)}(s,z)$ will follow provided 
$(s,z) \mapsto W_g^{(0)}(s,z)$ is analytic. Since that function can be obtained by  a restriction to a subspace
of the function  $(s,z) \mapsto \tilde{W}_g^{(0)}(s,z)$
the analyticity of the former  will   follow from the analycity of the latter.
To show
that the latter is analytic we use the  absolutely convergent expansion  given in  \eqref{eq:absconvExpansion},
which is  granted by   Proposition~\ref{thm:Feshbachpairchp41}. Since absolutely convergent sequences  of analytic
functions have an analytic limit, it remains to show that each summand in 
the following series 	is analytic in $s$ and $z$
\begin{align} 
(s,z) \mapsto &  \tilde{W}_g^{(0)}[s,z] \label{anaofWtilde}  \\
&  = 
  \sum_{L=1}^\infty(-1)^{L-1}
	\boldsymbol{\chi}(s)\, g\, W(s)  	\overline{\boldsymbol{\chi}}(s)  \big(H_0(s) - z\big)^{-1} \nonumber  \\
	& \times \left(
	 \, g \,\overline{\boldsymbol{\chi}}(s)
	\,W(s)
	\overline{\boldsymbol{\chi}}(s) \big(H_0(s) - z\big)^{-1}
	\right)^{L-1}	\overline{\boldsymbol{\chi}}(s) W(s)
	\boldsymbol{\chi}(s)       \nonumber \\
&  = 
  \sum_{L=1}^\infty(-1)^{L-1}
	\boldsymbol{\chi}(s)\, g\, W(s)  (H_{\rm f} + 1)^{-1} 	\overline{\boldsymbol{\chi}}(s)  (H_{\rm f} + 1 ) \big(H_0(s) - z\big)^{-1} \overline{\boldsymbol{\chi}}(s) 
	\nonumber \\
	& \times \left(
	 \, g \,
	\,W(s) (H_{\rm f} + 1)^{-1} 
	\overline{\boldsymbol{\chi}}(s) (H_{\rm f} + 1 ) \big(H_0(s) - z\big)^{-1} \overline{\boldsymbol{\chi}}(s) 
	\right)^{L-1}	 \nonumber  \\
	& \times W(s)(H_{\rm f} + 1 )^{-1} 
	(H_{\rm f} + 1 ) \boldsymbol{\chi}(s)     ,  \nonumber
\end{align}
where in the last equality   we used  associativity of composition and    that $H_{\rm f}$ commutes with $\boldsymbol{\chi}(s)$  and 
	$\overline{\boldsymbol{\chi}}(s)$.   
	First observe that  by Lemma \ref{analemW}, $W(s) (H_{\rm f} + 1)^{-1}$ is analytic.  Hence to establish 
	analyticity of \eqref{anaofWtilde} it remains to prove analyticity of 
	 \begin{align*}
	 & (H_{\rm f} + 1 ) (H_0(s) - z )^{-1} \overline{\boldsymbol{\chi}}(s). % |_{\ran  \overline{\boldsymbol{\chi}}(s) } . 
	 \end{align*} 
	To this end, we observe that from the definition of
	$\overline{\boldsymbol{\chi}}(s)$ we can write 
	\begin{align}
	 &  (H_{\rm f}+1)\big(H_0(s) - z\big)^{-1}
	 \overline{\boldsymbol{\chi}}(s)  \nonumber \\
		&
		= (H_{\rm f} + 1)\big(H_0(s)-z\big)^{-1}
			(\overline{P}_{\rm at}(s) \otimes \one) \nonumber  \\
			&+ (H_{\rm f} +1)\big(E_{\rm at}(s)
				+ H_{\rm f} - z\big)^{-1}
				(P_{\rm at}(s) \otimes
				\overline{\chi}_1) \, .  \label{eq:Hfbound} 
	\end{align}
The analyticity of the second  term in \eqref{eq:Hfbound} follows by means of the spectral theorem from the fact that  for  
	 every $r \geq 0$ the function  $(s,z) \mapsto (r + 1 )(E_{\rm at}(s) + r - z)^{-1} \overline{\chi}_1(r)$  
is analytic on $\mathcal{U}$ (by Hypothesis  \ref{Hypo3} we have on  $\mathcal{U}$  that $|E_{\rm at}(s) - z | <1/2$ and
so the denominator does not vanish for $r \geq 0$ for which $\overline{\chi}_1(r) \neq 0$) and  is uniformly bounded in $r \geq 0$. 
 The analyticity of the first term on the r.h.s of    \eqref{eq:Hfbound}   follows by means of the spectral theorem 
from the fact that   the function   
$(s,z) \mapsto (r + 1 )(H_{\rm at}(s) + r - z)^{-1} \overline{P}_{\rm at}(s)$  
is bounded uniformly in $r \geq 0$ by the estimate in Hypothesis \ref{Hypo3}
and for every $r \geq 0$ the function is
analytic on $\mathcal{U}$ by Proposition  \ref{GriesemerHaslerProp27}.
	This concludes the proof that
	${H}_g^{(0)}(s,z)$ is analytic on $\mathcal{U}$. From Eq.~\eqref{eq:auxiliaryOp} we see that 
the analyticity of $Q_{\boldsymbol{\chi}}(s,z)$ is established analogously as the analyticity of 
  \eqref{anaofWtilde}. 
	
	Part  $(b)$  follows in view of Hypothesis \ref{Hypo5} 
	from 	Theorem~\ref{thm:isospectralFeshbach} (a) by making the choice
	$Y = \HH_{\rm red}  = \ran \big(P_{\rm at}(s_0) \otimes 1_{H_{\rm f} \leq 1 } \big)$.
	Part $(c)$ follows from Theorem~\ref{thm:isospectralFeshbach}  (b). 
	%For more details we refer to the proof of
	%Proposition~10 in \cite{GriHas09}.
	Statements $(d)$ and $(e)$ follow from 
	Lemma~\ref{lem:FeshbachpreservesSymmetry}
	%Lemma~\ref{lem:trafofunc} 
	and the properties of the symmetry group given by Hypothesis \ref{Hypo2} (iii).
	
	Let us now show Part $(f)$.  First  observe that without loss the neighborhood $X_1 \subset X$ of $s_0$ 
	on which $P_{\rm at}$ is defined satisfies $X_1^* = X_1$ (otherwise take the intersection 
	of the two sets).  Now for $s \in \R \cap X_1$ close to $s_0$ we find from \eqref{eq:spectralprojection1} with $s_1 = s_0$ and 
	$E_{\rm at}(s_0) \in \R$ using Hypothesis \ref{Hypo4} (i),
	that  \begin{align} \label{Patsident} P_{\rm at}(\overline{s})^* = P_{\rm at}(s) %, \quad \overline{E_{\rm at}(\overline{s})} = E_{\rm at}(s)
	 \end{align}
	 Since both sides of   \eqref{Patsident}  are analytic functions of $s$ on $X_1$, we conclude
	 that     \eqref{Patsident} holds for all $s \in X_1$ (cf.  the unique continuation
 	property of analytic functions, e.g. \cite{Hoe90}). Furthermore it follows from Hypothesis \ref{Hypo4} (i) and   \eqref{eq:interaction} that 
	 \begin{align} \label{trafocompW} 
	 W(s)^* & = [ a(G_{1,\overline{s}}) +a^*(G_{2,s}) ]^* = a^*(G_{1,\overline{s}}) +a(G_{2,s})  \nonumber  \\
	 & =   a(G_{1,s}) +a^*(G_{2,\overline{s}})  
	   =   W(\overline{s}) 
	\end{align} 
	 for all $s \in X$. Now we recall  that for any densely defined, closed operator $A $ in $\HH$ and $z \in \rho(A)$ we find $\overline{z} \in \rho(A^*)$
	 and 
	 \begin{align} \label{adjointres} 
	 [ (A - z )^{-1}]^* = ( A^* - \overline{z} )^{-1}  .
	 \end{align} 
This follows directly from \cite[Theorem~4.17(b)]{Wei80} as is shown in the proof of Theorem~5.12 in \cite{Wei80}.
	Using the fact that $\tilde{H}_g(s,z)$ leaves the range of $P_{\rm at}(s_0) \otimes 1_{H_{\rm f} \leq 1}$ invariant we find for $(s,z) \in \mathcal{U} \cap \mathcal{U}^*$ that 
	%by Eq.~\eqref{eq:interaction},
	%Hypothesis~II,
	\begin{align*}
		{H}_g^{(0)}[s,z]^*
		&= \left(  \tilde{H}_g^{(0)}[s,z]\,
		\upharpoonright{\ran\,  P_{\rm at}(s_0) \otimes 1_{H_{\rm f} \leq 1} }   \right)^* \\
		&=  (\tilde{H}_g^{(0)}[s,{z}])^*
			\, \upharpoonright
			{\ran \, P_{\rm at}(s_0)
			\otimes    1_{H_{\rm f} \leq 1}   } \\
		&=   \tilde{H}_g^{(0)}[\overline{s},\overline{z}]
			\,  \upharpoonright
			{\ran  P_{\rm at}(s_0) \otimes   1_{H_{\rm f} \leq 1}   } \\
		&= {H}_g^{(0)}[\overline{s},\overline{z}] \, ,
	\end{align*}
	where the second to last  identity can be seen by taking the adjoint of  \eqref{eq:absconvExpansion}  
	and using  \eqref{Patsident}, \eqref{trafocompW},  and   \eqref{adjointres}. 
\end{proof}

%%%%%%%%%%%%%%%%%%%%%%%%%%%%%%%%%%%%%%%%%%%%%%%%%
			%Neighborhood of the free field energy Hf
%%%%%%%%%%%%%%%%%%%%%%%%%%%%%%%%%%%%%%%%%%%%%%%%%

%%%%%%%%%%%%%%%%%%%%%%%%%%%%%%%%%%%%%%%%%%%%%%%%%%%%%%%%%%%%%%%%%%%%%%

\section{Banach Space of Hamiltonians}
\label{sec:renorm}

To control the renormalization transformation, in
particular proving its convergence, it is convenient to introduce suitable Banach spaces of integral kernels, cf. \cite{BCFS,GriHas09}.
A generalization to matrix-valued integral kernels is a canonical choice to accommodate degenerate situations.
In this section we follow closely the definition and notation given in  \cite{GriHas09}.

%%%%%%%%%%%%%%%

The renormalization transformation is defined on a subset
of $\LL(\HH_{\rm red})$ that will be parameterized by
vectors of a Banach space $\WW_{\xi}=\oplus_{m,n\geq 0}\WW_{m,n}$.
We begin with the definition of this Banach space.

Let $\mathcal{L}(\C^d)$ denote the space of linear maps $A$ from $\C^d$ to $\C^d$ equipped 
with the operator norm  $\| A \|_{\rm op} := \sup \{ |A x | : |x| \leq 1 \}$.
The Banach space $\WW_{0,0}$ is the space of continuously differentiable functions
\begin{eqnarray*}
   \WW_{0,0} &:=& C^1([0,1] ; ( \mathcal{L}(\C^\degendim), \| \cdot \|_{\rm op}))\\
     \|w\|_{(\infty)}  & := &   \sup_{r \in [0,1] } \| w(r) \|_{\rm op} \\  
     \|w\| & := &  \| w \|_{(1,\infty)} := 
\|w\|_{(\infty)}+\|w'\|_{(\infty)}
\end{eqnarray*}
where $w'(r):=\partial_r w(r)$. For $m,n\in\N$
with $m+n\geq 1$ and $\mu>0$ we set
\begin{eqnarray}
   \WW_{m,n} &:=&
   L^2_s\left(B^{m+n},\frac{dK}{|K|^{2+2\mu}};\WW_{0,0}\right) \label{eq:WmnSec5} \\
    \|w_{m,n}\|_{\mu} &:=& \left(\int_{B^{m+n}}
    \|w_{m,n}(K)\|_{(1,\infty)}^2\frac{dK}{|K|^{2+2\mu}}\right)^{1/2} \label{defofmu2}
\end{eqnarray}
where $B:=\{k\in\R^3\times\{1,2\}:|k| <  1\}$ and
$$
     |K|:=\prod_{j=1}^{m+n}|k_j|,\qquad dK:=\prod_{j=1}^{m+n}dk_j.
$$
That is, $\WW_{m,n}$ is the space of measurable functions $w_{m,n}:B^{m+n}\to
\WW_{0,0}$ that are symmetric with respect to all permutations of
the $m$ arguments from $B^{m}$ and the $n$ arguments from $B^{n}$,
respectively, such that $\|w_{m,n}\|_{\mu}$ is finite.
We note that  the notation $\| \cdot \|_\mu$ introduced in  \eqref{defofmu2} 
also appears in \eqref{def:weaklymearsurBanach}.  Which of the definitions is meant should be clear from the context.

For given $\xi\in (0,1)$ and $\mu>0$ we define a Banach space
\begin{eqnarray*}
   \WW_{\xi} &:=& \bigoplus_{m,n\in \N} \WW_{m,n}\\
    \|w\|_{\mu,\xi} &:=& \sum_{m,n\geq 0}\xi^{-(m+n)}\|w_{m,n}\|_{\mu},
\end{eqnarray*}
$\|w_{0,0}\|_{\mu}:=\|w_{0,0}\|_{(1,\infty)}$, as the completion of the linear space of finite sequences
$w=(w_{m,n})_{m,n\in\N}\in \bigoplus_{m,n\in \N} \WW_{m,n}$ with respect to
the norm $\|w\|_{\mu,\xi}$.
The spaces $\WW_{m,n}$ will often be identified with the corresponding
subspaces of $\WW_{\xi}$.

Next we define a linear mapping $H:\WW_{\xi}\to\LL(\HH_{\rm
red})$. For \emph{finite} sequences $w=(w_{m,n})\in \WW_{\xi}$ the
operator $H(w)$ is the sum
$$
  H(w) := \sum_{m,n}H_{m,n}(w)
$$
of operators $H_{m,n}(w)$ on $\HH_{\rm red}$, defined by
$H_{0,0}(w) := w_{0,0}(H_{\rm f}),$ and, for $m+n\geq 1$,
\begin{align}  \label{Homegadef} 
  H_{m,n}(w) & := P_{\rm red}\left(\int_{B^{m+n}}
  a^{*}(k^{(m)})w_{m,n}(H_{\rm f},K) a(\tilde{k}^{(n)}){dK}\right) P_{\rm
  red},
\end{align} 
where $P_{\rm red}:=P_{[0,1]}(H_{\rm f})$, $K=(k^{(m)},\tilde{k}^{(n)})$, and
\begin{align*}
   k^{(m)} &= (k_1,\ldots,k_m)\in (\R^3\times\{1,2\})^{m},& a^{*}(k^{(m)}) &= \prod_{i=1}^{m}a^{*}(k_i),\\
   \tilde{k}^{(n)}&=(\tilde{k}_1,\ldots,\tilde{k}_n)\in (\R\times\{1,2\})^{n},& a(\tilde{k}^{(n)})&=\prod_{i=1}^{n}a(\tilde{k}_i).
\end{align*}
The formal definition of the operator valued distributions  $a^*(k)$ and $a(k)$   in \eqref{Homegadef}  can be found  in  Appendix \ref{sec:tecAux}.
 By the continuity established in the following proposition, the
mapping $w\mapsto H(w)$ has a unique extension to a bounded linear
transformation on $\WW_{\xi}$. 
\begin{proposition}[\cite{BCFS}[Theorem 3.1, Theorem 3.3]]   \quad  \label{H-is-bounded}
(i) For all $\mu>0$, $m,n\in \N$, with $m+n\geq 1$, and $w\in
\WW_{m,n}$, 
$$
   \|H_{m,n}(w) \| \leq \| P^\perp_\Omega H_{\rm f}^{-m/2}
   H(w_{m,n})  P^\perp_\Omega H_{\rm f}^{-n/2} \|
   \leq\frac{1}{\sqrt{m^m n^n}} \|w_{m,n}\|_\mu , 
$$
where we denoted the orthogonal projection in $\FF$  onto the subspace $\{ \Omega \}^\perp$ by $P_\Omega^\perp$. \\
(ii) For all $\mu>0$ and all $w\in \WW_{\xi}$
\begin{eqnarray*}
   \|H(w)\|  &\leq &  \|w\|_{\mu,\xi} \\ %\nonumber \\
   \|H(w)\|  &\leq &  \xi\|w\|_{\mu,\xi},\qquad \text{if}\
   w_{0,0}=0. %\label{eq:westimate}
\end{eqnarray*}
In particular, the mapping $w\mapsto H(w)$ is continuous.\\
(iii)  When restricted to  $$\{ w \in \mathcal{W}_\xi :   w_{m,n}({k}^{(m)} , \tilde{{k}}^{(n)})(r) 1_{r   + \max(\sum_{j=1}^m |{k}_j|,
 \sum_{l=1}^n |\tilde{{k}}_l | ) \geq 1  }  = 0   , \,  m +n \geq 1 \}$$ 
the map $H(\cdot)$ is injective. 
\end{proposition}

\begin{proof}
Statement (ii) follows immediately from the triangle inequality and (i) since $\xi \leq 1$. For
(i) we refer to the proof of \cite{BCFS}, Theorem~3.1. which generalizes trivially  to $\C^d$ with $d \geq 1$ from $d=1$. \\
(iii) For a proof see  the proof of \cite[Theorem 5.4]{HasHer11-1},   which generalizes straight forward to $\C^d$.
\end{proof}

Given $\alpha,\beta,\gamma\in \R_{+}$ we define polydiscs,
$\BB(\alpha,\beta,\gamma)\subset H(\WW_{\xi})$ of the operator
$P_{\rm red}H_{\rm f} P_{\rm red}\in \LL(\HH_{\rm red})$ by
$$
  \BB(\alpha,\beta,\gamma) := \big\{H(w) :  \|w_{0,0}(0) \|_{\rm op} \leq \alpha,\ \|w_{0,0}'-1\|_{(\infty)}\leq\beta,
  \ \|w-w_{0,0}\|_{\mu,\xi}\leq \gamma\big\}.
$$
Note that $w_{0,0}(0) \in \mathcal{L}(\C^d)$ is  uniquely determined by  the identity 
$$
\langle v_1  ,  w_{0,0}(0) v_2  \rangle  =  \langle v_1   \otimes \Omega ,  H(w) v_2  \otimes \Omega \rangle 
$$ 
which holds for all $v_1 , v_2 \in \C^d$. 
The definition of $\BB(\alpha,\beta,\gamma)$ is motivated by  Lemma~\ref{feshbachtest}  and by Theorem~\ref{bcfssigal}, below.

\section{First Transformation} 
\label{sec:firstrafo} 

In the following we denote by
\begin{equation} \label{eq:dimgroundstatespace}
d = \dim\big(\ran P_{\rm at}(s_0)\big)
\end{equation}
the dimension of the eigenspace corresponding to the
 eigenvalue $E_{\rm at}(s_0)$  of $H_{\rm at}(s_0)$.

\begin{theorem}\label{thm:iniFeshbachBall}
Suppose  Hypothesis~\ref{Hypo1} holds for some $\mu > 0$ , Hypothesis \ref{Hypo2} holds, Hypothesis \ref{Hypo3} holds
for some $\mathcal{U} \subset \C^{\nu} \times \C$,  and  Hypothesis \ref{Hypo5} holds. 
 Then, for all $\xi \in (0,1)$ and arbitrarily positive
 constants $\alpha_0$, $\beta_0$ and $\gamma_0$,
 there exits a positive constant $g_1$ such that for
 all $g \in [0,g_1)$ and all $(s,z) \in \mathcal{U}$,
 $(H_g(s)-z, H_0(s)-z)$ is a Feshbach pair for
 $\boldsymbol{\chi}(s)$, and
 \begin{align*} %\label{firstfeshbanachestimate} 
	{H}_g^{(0)}[s,z] - (E_{\rm at}(s) - z)
		\in \mathcal{B}
		(\alpha_0,\beta_0,\gamma_0)\,.
 \end{align*}

\end{theorem}

\begin{proof} 
Using Proposition~\ref{thm:Feshbachpairchp41}
we directly obtain that
the Feshbach property is satisfied
for sufficiently small~$g$.
Hence to prove the theorem it remains
to construct a sequence of integral kernels
$w \in \WW_\xi$ such that
${H}_g^{(0)}(s,z) = H(w)$.
By the definition of the space
$\mathcal{B}(\alpha_0,\beta_0,\gamma_0)$, the validity of Hypotheses~\ref{Hypo1}, \ref{Hypo2}, \ref{Hypo3},  and  $\Patsdimension = \dim\big(\ran P_{\rm at}(s_0)\big)$ (by  Hypothesis~\ref{Hypo5}) this construction is equal 
to the one in \cite[Theorem 23]{GriHas09} where a sequence of
integral kernels with values in $C^1([0,1])$ was
constructed. 
\end{proof} 

\begin{remark} 
We note that  a  result for  matrix-valued integral kernels  similar as in Theorem  \ref{thm:iniFeshbachBall}  can be found with a detailed proof  in \cite{HasLan18-1}.
\end{remark}

\section{RG Transformation} 
\label{sec:defrgtrafo}

By abuse of notation we shall denote the following operators on $\HH_{\rm red}$ 
$$
1_{\C^d} \otimes \chi_\rho  , \quad 1_{\C^d} \otimes \chib_\rho 
$$
again by $\chi_\rho$ and $\chib_\rho$, respectively, recalling the notation \eqref{eq:dimgroundstatespace}.
 It should be clear from the context which of the expressions is considered. 

\begin{lemma} \label{feshbachtest} Suppose $\rho,\xi\in(0,1)$ and $\mu > 0$.
If $H(w) \in \BB(\rho/2, \rho/8 , \rho/8)$, then $(H(w),
H_{0,0}(w))$ is a Feshbach pair for $\chi_\rho$. 
\end{lemma}

The proof of the  lemma follows from a straight forward generalization
of the proof  given in Lemma 15 in \cite{GriHas09}. Moreover a similar proof can be found in \cite{FroGriSig09}.

\begin{proof}
The assumption $H(w) \in \BB(\rho/2, \rho/8 , \rho/8)$ implies, by
Proposition \ref{H-is-bounded}, that
\begin{eqnarray*}
\| H(w) - H_{0,0}(w) \| \leq \xi \frac{\rho}{8} \; .
\end{eqnarray*}
For $r \in [\frac{3}{4} \rho , 1 ]$, and for $v\in \C^d$  a normalized vector we have by triangle inequality 
\begin{eqnarray*}
\| w_{0,0}(r) v \|_{\rm op}  &\geq& r - \| ( w_{0,0}(r) - w_{0,0}(0)) - r \|_{\rm op} - \|w_{0,0}(0)\|_{\rm op} \\
&\geq & r ( 1 - \sup_r \| {w'}_{0,0}(r) - 1 \|_{\rm op} ) - \frac{\rho}{2}  \\
& \geq & \frac{3 \rho}{4} ( 1 - \frac{\rho}{8} ) - \frac{\rho}{2}
%\geq \frac{5 \rho}{32}
 \geq \frac{\rho}{8} \; .
\end{eqnarray*}
Thus for $r \in [\frac{3}{4} \rho , 1]$ the linear map $w_{0,0}(r)$  is invertible and $\| w_{0,0}(r)^{-1} \|_{\rm op} \leq 8/\rho$.    
From this and the  spectral theorem,
\begin{eqnarray*}
\| H_{0,0}(w)^{-1} \upharpoonright \ran \chib_\rho \| = \| w_{0,0}(H_{\rm f}
)^{-1} \upharpoonright \ran \chib_\rho \| \leq  \sup_{r \in
[\frac{3}{4} \rho, 1 ]} \|  (w_{0,0}(r))^{-1} \|_{\rm op} 
 \leq  \frac{8}{\rho}  \; .
\end{eqnarray*}
Since $\| \chib_\rho \| \leq 1 $, it follows from the estimates
above that
$$
\| H_{0,0}(w)^{-1} \chib_\rho (H(w) - H_{0,0}(w)  ) \chib_\rho
\upharpoonright \ran \chib_\rho \| \leq \xi < 1 \; .
$$
This implies the  bounded invertibility of
\begin{eqnarray*}
\lefteqn{ \left( H_{0,0}(w) + \chib_\rho (H(w) - H_{0,0}(w) )
\chib_\rho
\right) \upharpoonright \ran \chib_\rho } \\
&& = H_{0,0}(w) \left( 1 + H_{0,0}(w)^{-1} \chib_\rho (H(w) -
H_{0,0}(w) ) \chib_\rho \right)\upharpoonright \ran \chib_\rho \; .
\end{eqnarray*}
The other conditions on a Feshbach pair are now also satisfied,
since $H(w) - H_{0,0}(w)$ is bounded on $\HH_{\rm red}$.
\end{proof}

%%%%%%%%%%%%%%%%%%%%%%%%%%%%%%%%%%%%%%%%%%%%%%%%%%%%%%%%%%%%%%%%%%%%

The \emph{renormalization transformation} we use is a composition
of a Feshbach transformation and a unitary scaling that puts the
operator back on the original Hilbert space $\HH_{\rm red}$.
Unlike the renormalization transformation of Bach et al
\cite{BCFS}, there is no analytic transformation of
the spectral parameter.

Given $\rho\in (0,1)$, let $\HH_{\rho}=1_{\C^d} \otimes  \ran\chi(H_{\rm f}\leq \rho)$. Let $w\in \WW_{\xi}$
and suppose $(H(w),H_{0,0}(w))$ is a Feshbach pair for $\chi_{\rho}$.
Then
$$
    F_{\chi_{\rho}}(H(w),H_{0,0}(w)): \HH_{\rho}\to \HH_{\rho}
$$
is iso-spectral with $H(w)$ in the sense of Theorem~\ref{thm:isospectralFeshbach}.   In
order to get a isospectral operator on $\HH_{\rm red}$, rather
than $\HH_{\rho}$, we use the linear isomorphism
$$
    \Gamma_{\rho}: \HH_{\rho}\to \HH_1=\HH_{\rm red},
    \qquad
    \Gamma_{\rho}:=\Gamma(U_{\rho})\upharpoonright\HH_{\rho},
$$
introduced in  \eqref{defscaling}. 
Note that $\Gamma_{\rho}H_{\rm f} \Gamma_{\rho}^{*}=\rho H_{\rm f}$, and hence
$\Gamma_{\rho}\chi_{\rho} \Gamma_{\rho}^{*}=\chi_1$.   The
renormalization transformation $\RR_{\rho}$ maps bounded operators
on $\HH_{\rm red}$ to bounded linear operators on $\HH_{\rm red}$
and is defined on those operators $H(w)$ for which $(H(w),
H_{0,0}(w))$ is a Feshbach pair with respect to $\chi_{\rho}$.
Explicitly,
$$
   \RR_{\rho}(H(w)) :=
   \rho^{-1}\Gamma_{\rho}\FF_{\chi_{\rho}}(H(w),H_{0,0}(w))\Gamma_{\rho}^{*},
$$
which is a bounded linear operator on $\HH_{\rm red}$.

The following theorem describes the action of the renormalization 
transformation on the polydiscs $\mathcal{B}(\alpha,\beta,\gamma)$.
For its statement we recall the notation  \eqref{defofAOmega}.

\begin{theorem}[\bf BCFS \cite{BCFS}] \label{bcfssigal}
There exists a constant $C_\chi\geq 1$  depending only on $\chi$,
such that the following holds. If $\mu > 0$, $\rho\in(0,1)$, $\xi
= \sqrt{\rho}/(4C_\chi)$, and $\beta, \gamma \leq \rho/(8C_\chi)$,
then
$$
\RR_\rho - \rho^{-1}  \langle \, \cdot \, \rangle_\Omega  :
\BB(\rho/2, \beta, \gamma) \to \BB(\alpha' , \beta' , \gamma' )\; ,
$$
where
\begin{align} \label{itpararen} 
\alpha' = C_\beta \frac{\gamma^2}{\rho}
 \; , \quad \beta' = \beta +   C_\beta \frac{\gamma^2}{\rho}  \; ,
 \quad
  \gamma' = C_\gamma  \rho^\mu \gamma  \; ,
\end{align} 
with $C_\beta := \frac{3}{2} C_\chi$, $C_\gamma := 128 C_\chi^2$.
\end{theorem}

 Theorem  \ref{bcfssigal} is a variant of Theorem~3.8 of
\cite{BCFS}, with additional information from the
proof of that theorem, in particular from Equations~(3.104),
(3.107) and (3.109). Another difference is due to our different
definition of the Renormalization transformation, i.e., without
analytic deformation of the spectral parameter.
We note that versions of  Theorem \ref{bcfssigal} have been used  in the literature in      \cite[Theorem~16] {GriHas09}
as well as in   \cite[Appendix 1]{FroGriSig09}, where a detailed proof was presented.

\section{Renormalization preserves analyticity and symmetry}
\label{sec:RenormPresAnalytSym}
In this section we show that the renormalization transformation preserves analyticity,  symmetry with respect to a group
of symmetries $\mathcal{S}$ and reflection symmetry. We study these properties   on the level of the operators. In principle one 
could also study the symmetry property on the level of the integral kernels. 

In \cite[Proposition~17]{GriHas09},
Griesemer and Hasler proved
that analyticity is preserved under renormalization.
The following proposition is a straight forward generalization of their result.

\begin{proposition}[Proposition~17, \cite{GriHas09}]
\label{prop:Feshbachanalytic}
Let $X$ be an open subset of $\C^{\nu + 1 }$ with
$\nu \geq 0$.
Suppose that the map $\sigma~\mapsto~H(w^\sigma)
\in \mathcal{L}(\HH_{\rm red})$
is analytic on $X$, and that $H(w^\sigma)$ belongs to
some polydics
$\mathcal{B}(\alpha, \beta, \gamma)$ for all $\sigma \in X$.
Then
	\begin{itemize}
		\item[(a)] $H_{0,0}(w^\sigma)$ is analytic on $X$.
		\item[(b)] If for all
		$\sigma \in X$, $(H(w^\sigma), H_{0,0}(w^\sigma))$
			is a Feshbach pair for $\chi_\rho$, then
			$F_{\chi_\rho}(H(w^\sigma),H_{0,0}(w^\sigma))$
			is analytic on $X$.
	\end{itemize}
\end{proposition}
\begin{proof}  Follows from   \cite[Proposition~17]{GriHas09} and an obvious 
change of notation to accommodate the matrix valued integral kernels. 
\end{proof} 
The  property in  Proposition  \ref{prop:Feshbachanalytic}   together with Proposition \ref{prop:Feshbachcommutes}, below,  will be one of the main ingredients in the proof of part (i) of Theorem~\ref{thm:symdegenSpinBoson}.

\begin{proposition}\label{prop:Feshbachcommutes}
Let $X$ be an open subset of $\C^{\nu + 1 }$ with $\nu \geq 0$. Assume that for  each $\sigma \in X$ we are given an operator $H(w^\sigma)$   
 in the polydisc
$\mathcal{B}(\alpha, \beta, \gamma)$. 

\begin{itemize}
	\item[(a)] Let $\mathcal{S}$ be a group of symmetries acting on $\HH_{\rm red}$   leaving the Fock vacuum
and the one particle subspace invariant. Assume that it commutes with $\Gamma_\rho$ and $H_{\rm f}$. 
Let $\sigma \in X$.  Suppose that 
$ H(w^\sigma) $
is  symmetric with respect to  $\mathcal{S}$.
\begin{itemize} 
\item[(i)] Then 
$H_{0,0}(w^\sigma)$ is symmetric  with respect to  $\mathcal{S}$.
	\item[(ii)]  If 
		$(H(w^\sigma), H_{0,0}(w^\sigma))$ is a Feshbach pair for $\chi_\rho$, then
		$F_{\chi_\rho}(H(w^\sigma),H_{0,0}(w^\sigma))$ and  $\mathcal{R}_\rho(H(w^\sigma))$  are  symmetric  with  respect to $\mathcal{S}$.
\end{itemize} 
\item[(b)] Suppose $X = {X}^*$ and $\sigma \mapsto H(w^\sigma)$ is reflection  symmetric.
\begin{itemize}
\item[(i)]  Then 
		$H_{0,0}(w^\sigma)$ is reflection symmetric. 
	\item[(ii)] If  $(H(w^\sigma), H_{0,0}(w^\sigma))$ is a Feshbach pair for $\chi_\rho$, then 
 $F_{\chi_\rho}(H(w^\sigma),H_{0,0}(w^\sigma))$ and $\mathcal{R}_\rho(H(w^\sigma))$ are reflection symmetric.
\end{itemize}
\end{itemize}
\end{proposition}
\begin{proof} We first show how one can recover $w_{0,0}(r)$ from $H(w)$. We follow the argument in [BCFS].
Let $w \in \WW_\xi$.
 Let $v_1,v_2 \in \C^d$. 
For $f, g \in \hh$ we have 
\begin{align}\label{eq:oneparticleexp} 
& \langle v_1 \otimes a^*(f) \Omega , H(w)  ( v_2 \otimes a^*(g) \Omega )  \rangle  \\
& =  \langle v_1 \otimes a^*(f ) \Omega , w_{0,0}(H_{\rm f}) ( v_2 \otimes a^*(f ) \Omega ) \rangle  +
 \langle v_1 \otimes a^*(f )  \Omega , H_{1,1}(w) ( v_2 \otimes a^*(f ) \Omega )  \rangle   \nonumber 
\end{align} 
A simple calculation shows that 
\begin{align} \label{eq:weakconv0} 
& \langle v_1 \otimes a^*(f) \Omega , w_{0,0}(H_{\rm f})  ( v_2 \otimes a^*(g) \Omega )  \rangle \nonumber \\
& =   \int_{B_1}  \overline{f(x)} g(x) \inn{ v_1, w_{0,0}(|x|) v_2}   dx 
\end{align} 
and 
\begin{align} \label{eq:weakconv} 
& \langle v_1 \otimes a^*(f) \Omega , H_{1,1}(w)  ( v_2 \otimes a^*(g) \Omega )  \rangle \nonumber \\
& =   \int_{B_1^2}  \overline{f(x)} g(x') \inn{ v_1, w_{1,1}(0,x,x') v_2}  dx dx' = 0  .
\end{align} 

We pick a function $f \in C_c^\infty(B_1;[0,\infty)$ with $\int |f(x)|^2 dx = 1$, and define $f_{\epsilon,k} :=  \epsilon^{-3/2} f(\epsilon^{-1}(x-k))$.
Then we find from \eqref{eq:weakconv0} 
\begin{align}
& \langle v_1 \otimes a^*(f_{\epsilon,k}) \Omega , w_{0,0}(H_{\rm f})   ( v_2 \otimes a^*(f_{\epsilon,k}) \Omega )  \rangle \nonumber   \\
%& =  \langle v_1 \otimes a^*(f_{\epsilon,k}) \Omega , w_{0,0}[H_{\rm f}] ( v_2 \otimes a^*(f_{\epsilon,k}) \Omega ) \rangle  \nonumber  \\
& = \int_{B_1}  | f_{\epsilon,k}(x)|^2 \inn{ v_1,  w_{0,0}(|x|) v_2}  dx .  \label{zeroforexp0} 
\end{align} 
This term tends to $\inn{ v_1, w_{0,0}(|k|) v_2}$ since
\begin{align} \label{deltaaproox} |f_{\epsilon,k}(x)|^2 \to \delta(x-k) \quad \epsilon \to 0 .
\end{align} 
On the other hand we find from \eqref{eq:weakconv} 
\begin{align}
& \langle v_1 \otimes a^*(f_{\epsilon,k}) \Omega , H_{1,1}(w)   ( v_2 \otimes a^*(f_{\epsilon,k}) \Omega )  \rangle \nonumber   \\
%& =  \langle v_1 \otimes a^*(f_{\epsilon,k}) \Omega , w_{0,0}[H_{\rm f}] ( v_2 \otimes a^*(f_{\epsilon,k}) \Omega ) \rangle  \nonumber  \\
%& = \int_{B_1} d^d x | f_{\epsilon,k}(x)|^2 \inn{ v_1,  w_{0,0}[|x|] v_2}  +  
&  =  \int_{B_1^2} \overline{f_{\epsilon,k}(x)} f_{\epsilon,k}(x') \inn{ v_1, w_{1,1}(0,x,x') v_2} dx dx' . \label{zeroforexp} 
\end{align} 
This term tends to 0, because $f_{\epsilon,k} \to 0$, weakly in $L^2(B_1)$.
Thus from \eqref{eq:oneparticleexp}  -- \eqref{zeroforexp}  we conclude using that $w_{0,0}$ is continuous  that 
\begin{align} \label{eq:mainident} 
& \lim_{\epsilon \downarrow 0} \langle v_1 \otimes a^*(f_{\epsilon,k}) \Omega , H(w)  ( v_2 \otimes a^*(f_{\epsilon,k}) \Omega )  \rangle  = \inn{ v_1, w_{0,0}(|k|) v_2 }  . 
\end{align} 

(a)  Since this part does not depend on $\sigma$ we drop it in the notation. 
Now since $S \in \mathcal{S}_2$ leaves the one photon space invariant, there is a map  $p_1(S)$ such that 
 \begin{align*} 
 	S a^*(f ) \Omega = a^*(p_1(S) f  ) \Omega . 
 \end{align*} 
If $S$ is unitary or antiunitary, it follows that  $p_1(S) $ is unitary or antiunitary, respectively.
%Now by unitarity of $S$ it follows that also $p_1(S)$ is unitary. 
%In particular also $p_1(S) f_{r,\epsilon}  \to 0$ weakly in $L^2(B_1)$. 
Now let $S = S_1 \otimes S_2 \in \mathcal{S}$ by a symmetry. If $S$ is unitary we write  $( \cdot )^\# = ( \cdot )$ and if it is antiunitary
we write  $( \cdot  )^\# = ( \cdot  )^*$.
Thus we find from  \eqref{eq:mainident}  that 
\begin{align}
&  \inn{ v_1, w_{0,0}(|k|) v_2} \nonumber  \\
& =    \lim_{\epsilon \downarrow 0} \langle v_1 \otimes a^*(f_{k,\epsilon})  \Omega ,   H(w)  ( v_2 \otimes a^*(f_{k,\epsilon}) \Omega  ) \rangle \nonumber  \\ 
& =  \lim_{\epsilon \downarrow 0} \langle v_1 \otimes a^*(f_{k,\epsilon})  \Omega ,  S  H(w)^\#  S^* (v_2 \otimes a^*(f_{k,\epsilon}) \Omega )  \rangle  \nonumber  \\
&  =  \lim_{\epsilon \downarrow 0} \langle v_1 \otimes a^*(f_{k,\epsilon})  \Omega ,  (S_1 \otimes S_2)  H(w)^\#  (S_1 \otimes S_2)^* (v_2 \otimes  a^*(f_{k,\epsilon}) \Omega )  \rangle   \nonumber  \\
&  =  \lim_{\epsilon \downarrow 0}  \langle S_1^* v_1 \otimes a^*(p_1(S_2^*) f_{k,\epsilon})  \Omega ,     H(w)^\#   (S_1 v_2 \otimes  a^*(p_2(S_2^*)f_{k,\epsilon}) \Omega )  \rangle^\# \nonumber   \\
&  =  \lim_{\epsilon \downarrow 0}  \langle S_1^* v_1 \otimes a^*(p_1(S_2^*) f_{k,\epsilon})  \Omega ,     w_{0,0}(H_{\rm f})^\#  (S_1 v_2 \otimes  a^*(p_2(S_2^*)f_{k,\epsilon}) \Omega )  \rangle^\# \label{firstlimit1}   \\
&  =  \lim_{\epsilon \downarrow 0} \langle v_1 \otimes a^*(f_{k,\epsilon})  \Omega ,  (S_1 \otimes S_2)   w_{0,0}(H_{\rm f})^\# (S_1^*  \otimes S_2^*)  ( v_2 \otimes  a^*( f_{k,\epsilon}) \Omega )  \rangle    \label{firstlimit2}  \\
&  =  \lim_{\epsilon \downarrow 0} \langle v_1 \otimes a^*(f_{k,\epsilon})  \Omega ,   S_1  w_{0,0}(H_{\rm f})^\#  S_1^*  ( v_2 \otimes  a^*( f_{k,\epsilon}) \Omega )  \rangle \nonumber    \\
&=  \inn{   v_1,   S_1 w_{0,0}(|k|)^\# S_1^* v_2 } ,\nonumber 
\end{align}
where in   \eqref{firstlimit1}    we made use of  \eqref{eq:oneparticleexp},  \eqref{eq:weakconv}  and the fact that  $p_2(S^*) f_{k,\epsilon}$ converges 
to zero.   In \eqref{firstlimit2}  we used that    $H_{\rm f}$ is symmetric with respect to  $S_2$. 
In the last line we used \eqref{zeroforexp} 
 and \eqref{deltaaproox}. We conclude that  $ S_1 w_{0,0}(r)S_1^* = w_{0,0}(r)$ for all $r \in [0,1]$.   This shows part (i) of (a).
This shows (i). \\
(ii) 
Then from (i) we know that  $H_{0,0}(w)$ is symmetric with respect to   $\mathcal{S}$. Thus 
it follows that also   $W := H(w) - H_{0,0}(w)$ is symmetric. Now the claim for the  Feshbach  operator follows 
from  Lemma~\ref{lem:FeshbachpreservesSymmetry}. Since the symmetry commutes with dilations the claim
follows also for the renormalized expression.

(b)  Suppose now  $X = {X}^*$ and $\sigma \mapsto H(w^\sigma)$ is reflection symmetric.
Then by \eqref{eq:mainident} it follows that 
\begin{align*} %\label{eq:mainident1} 
 & \inn{ v_1, w_{0,0}^{\overline{\sigma}}(|k|) v_2 } \\
& = \lim_{\epsilon \downarrow 0} \langle v_1 \otimes a^*(f_{\epsilon,k}) \Omega , H(w^{\overline{\sigma}})  ( v_2 \otimes a^*(f_{\epsilon,k}) \Omega )  \rangle    \\
&= \lim_{\epsilon \downarrow 0}  \langle v_1 \otimes a^*(f_{\epsilon,k}) \Omega , H(w^{\sigma})^* ( v_2 \otimes a^*(f_{\epsilon,k}) \Omega )  \rangle \\
&= \lim_{\epsilon \downarrow 0} \overline{  \langle v_2 \otimes a^*(f_{\epsilon,k}) \Omega , H(w^{\sigma}) ( v_1 \otimes a^*(f_{\epsilon,k}) \Omega )  \rangle }  \\
&  = \overline{ \inn{ v_2, w_{0,0}^{\sigma}(|k|) v_1 } }  =   \inn{ v_1, w_{0,0}^{\sigma}(|k|)^* v_2 }   .
\end{align*} 
Thus for $r \in [0,1]$ we find  $w_{0,0}^{\overline{\sigma}}(r) = w_{0,0}^{\overline{\sigma}}(r)^*$.
This shows part (i) of (b).  To show (ii) we write $T^\sigma = H_{0,0}(w^\sigma)$  and observe that $W^\sigma = H(w^\sigma) -  T(w^\sigma)$  is also reflection symmetric as well as $\chi = \chi_\rho$.  We  find 
\begin{align*}
	&   F_{\chi}(H(w^\sigma) ,T^\sigma)^*
  \\ & =  \left(   T^\sigma + \chi W^\sigma \chi -  \chi W^\sigma \overline{\chi}
		(( T^\sigma + \overline{\chi} W^\sigma \overline{\chi})|_{{\ran \overline{\chi}}})^{-1} \overline{\chi} W^\sigma \chi  \right)^*  \\
 & =     T^{\overline{\sigma}} + \chi W^{\overline{\sigma}} \chi -  \chi W^{\overline{\sigma}} \overline{\chi}
		(( T^{\overline{\sigma}} + \overline{\chi} W^{\overline{\sigma}} \overline{\chi})|_{{\ran \overline{\chi}}})^{-1} \overline{\chi} W^{\overline{\sigma}} \chi    \\
& = F_\chi(H(w^{\overline{\sigma}},T^{\overline{\sigma}})  \,. 
\end{align*}
This shows the claim for the Feshbach operator.  Since the symmetry commutes with dilation the claim
follows also for the renormalized expression. 
\end{proof}

\section{Iterating the Renormalization Transformation}
\label{sec:iterationSymDegen}

In this section we follow closely, Section 8 in  \cite{GriHas09},  and generalize the results 
given there to the non-degenerate situation.  In particular the two lemmas  stated below 
are  almost identical to the main  results  stated in Lemma 18, Lemma 19, Corollary 20, and Proposition 21 of  \cite{GriHas09}.

In   Part (c) of   Theorem \ref{thm:feshbachpair} we have reduced, for small $|g|$, the
problem of finding an eigenvalue  of $H_g(s)$ in the neighborhood
$$ U_0(s):=\{z\in\C : (s,z)\in\UU\}$$ of $\Eat(s)$ to finding an  $z\in \C$
such that $H^{(0)}[s,z]$ has a non-trivial kernel. We now use the
renormalization map to define a sequence $$H^{(n)}[s,z]:=\RR_\rho^n H^{(0)}[s,z]$$ of
operators on $\HH_{\rm red}$, which, by Theorem~\ref{thm:isospectralFeshbach}, are isospectral in
the sense that $\ker H^{(n+1)}[s,z]$ is isomorphic to $\ker
H^{(n)}[s,z]$. The main purpose of the present section is to show that for every $n \in \N$
the operator $H^{(n)}[s,z]$ is  well-defined
for all $z$  in a non-empty set $U_n(s)$ with the following properties. We have   $U_{n+1}(s) \subset U_n(s)$  and 
$$
\bigcap_{n=0}^\infty U_n(s) = \{ z_\infty(s) \}  .
$$
In Section \ref{sec:iterationSymDegen}
we will show that $H^{(n)}[s,z_{\infty}(s)]$ has a non-trivial
kernel and  hence $z_{\infty}(s)$ is  an eigenvalue of $H_g(s)$. The
construction of the sets $U_n(s)$ is based on Theorem~\ref{thm:feshbachpair} and  
Theorem 
\ref{thm:iniFeshbachBall}, but not on the
explicit form of $H^{(0)}[s,z]$ as given by \eqref{eq:defh0}.

Moreover, this construction is pointwise in $s$ and $g$,
all estimates being \emph{uniform in} $s\in X$ and $|g|<g_{\rm b}$ for
some $g_{\rm b}>0$. We therefore drop these parameters from our notations and we
now explain the construction of $H^{(n)}[z]$ making only the following assumption:

\begin{itemize}
\item[\textbf{(A)}] $U_0(s)$ is an open subset of $\C$ and for every $z\in U_0$,
$$
    H^{(0)}[z] \in \BB(\infty,\rho/8,\rho/8).
$$
If $ d \geq 1$  there is a group of  symmetries $\mathcal{S}$ of  $H_{\rm f}$ such that  $H^{(0)}[z]$ is symmetric with respect to each 
element of $\mathcal{S}$ and  $\mathcal{S}_1 := \{ S_1 : S_1 \otimes S_2 \in \mathcal{S} \}$   acts irreducibly on $\C^d$. 
 Each element of $\mathcal{S}_2 := \{ S_2 :  S_1 \otimes S_2 \in \mathcal{S} \}$ leaves the 
Fock vacuum as well as  the one particle subspace  invariant and 
 commutes with the operator 
of   dilations.

The polydisc $\BB(\infty,\rho/8,\rho/8)\subset
H(\WW_{\xi})$ is defined in terms of $\xi := \sqrt{\rho}/(4 C_\chi)$ and
$\mu>0$, where $\rho \in (0,1)$ and $C_\chi$ is given by Theorem~\ref{bcfssigal}.
\end{itemize}
By Lemma~\ref{feshbachtest}, we may define $H^{(1)}[z],\dots,H^{(N)}[z]$, recursively by
\begin{equation}\label{eq:def-Hn}
    H^{(n)}[z]:= \RR_{\rho}(H^{(n-1)}[z])
\end{equation}
provided that $H^{(0)}[z],\ldots, H^{(N-1)}[z]$ belong to
$\BB(\rho/2,\rho/8,\rho/8)$. Theorem~\ref{bcfssigal} gives us
sufficient conditions for this to occur: by iterating the map
$(\beta,\gamma)\mapsto (\beta',\gamma')$, cf. \eqref{itpararen},   starting with
$(\beta_0,\gamma_0)$, we find the conditions
\begin{eqnarray} \label{eq:cond1}
  \gamma_n := \left(C_\gamma \rho^{\mu}\right)^n \gamma_0 &\leq&  \rho/(8 C_\chi) \\
  \beta_n := \beta_0 + \left( \frac{C_\beta}{\rho} \sum_{k=0}^{n-1}
  (C_\gamma \rho^\mu)^{2k} \right) \gamma_0^2 &\leq& \rho/(8 C_\chi) \; ,
  \label{eq:cond2}
\end{eqnarray}
for $n=0,\ldots, N-1$. They are obviously satisfied for all $n\in\N$
if $C_\gamma \rho^\mu<1$ and if $\beta_0,\gamma_0$ are sufficiently
small. If this is the case we define $$T_0^{(n)}(z)  = \langle H^{(n)}[z]  \rangle_\Omega  .$$ 
Since the renormalization transformation $\mathcal{R}_\rho$ preserves the symmetry by Proposition \ref{prop:Feshbachcommutes}, it 
follows by induction  from Assumption (A) that each $H^{(n)}[z]$ is symmetric with respect to the elements of $\mathcal{S}$.  
Since the symmetries leave the vacuum invariant  it follows  from Lemma  \ref{lem:1-dimOp} that the linear map $T_0^{(n)}(z) $
is multiple of the identity.  
That is, there exists a function 
$E^{(n)} : U_n \to \C$ such that $$ T_0^{(n)}(z)  =  E^{(n)}(z) 1_{\C^d} . $$
Now  it remains to
make sure that $$\|T_0^{(n)}(z)\|_{\rm op} \leq \rho/2$$ for $n=0,\ldots, N-1$. Since $| E^{(n)}(z)   | = \|T_0^{(n)}(z)\|_{\rm op} $ this
is achieved by adjusting the admissible values of $z$ step by step.
We define recursively, for all $n\geq 1$,
\begin{equation*}%\label{def:Un}
   U_n := \{ z\in U_{n-1}: |E^{(n-1)}(z) | \leq \rho/2\}.
\end{equation*}
If $z\in U_N$, $H^{(0)}(z)\in \BB(\infty,\beta_{0},\gamma_{0})$,
and $\rho,\beta_0,\gamma_0$ are small enough, as explained above, then the
operators $H^{(n)}(z)$ for $n=1,\ldots,N$ are well defined by
\eqref{eq:def-Hn}. In addition we know from Theorem~\ref{bcfssigal}
that $H^{(n)}(z)\in \BB(\infty,\beta_n,\gamma_n)$, and that
\begin{equation} \label{eq:alphan3}
   \left|E^{(n)}(z)-\frac{E^{(n-1)}(z)}{\rho}\right| \leq
   \frac{C_\beta}{\rho}\gamma_{n-1}^2 =: \alpha_n.
\end{equation}
This latter information will be used in the proof of
Lemma~\ref{prop:balls} to show that the sets $U_n$ are not empty.

The subsequent lemma is a summary of the
above construction.

\begin{lemma} \label{cor:bcfs}
Suppose that (A) holds with $\rho\in(0,1)$ so small, that $C_\gamma \rho^\mu <1$.
Suppose $\beta_0, \gamma_0 \leq \rho/(8C_\chi)$ and, in addition,
\begin{equation} \label{eq:4446}
\beta_0 + \frac{C_\beta/\rho}{1 - ( C_\gamma \rho^\mu)^2}
\gamma_0^2 \leq \frac{\rho}{8 C_\chi} \; .
\end{equation}
If $H^{(0)}[z]\in \BB(\infty,\beta_0,\gamma_0)$ for all $z\in U_0$,
then $H^{(n)}[z]$ is well defined for $z\in U_n$, symmetric  with respect to the elments of $\mathcal{S}$, and satisfies
\begin{equation*} %\label{eq:basicstern}
   H^{(n)}[z] - \frac{1}{\rho} E^{(n-1)}(z) \in \BB(\alpha_n , \beta_n ,
\gamma_n), \quad {\it for} \ \ n \geq 1
\end{equation*}
with $\alpha_n$, $\beta_n$, and $\gamma_n$ as in \eqref{eq:alphan3},
\eqref{eq:cond2}, and  \eqref{eq:cond1}.
\end{lemma}

The next lemma establishes conditions under which the
set $U_0$ and $U_n$ are non-empty. We introduce  the
discs
$$
 D_r := \{z \in \mathbb{C} | |z| \leq r \}
$$
and note that $U_n={E^{(n-1)}}^{-1}(D_{\rho/2})$.

\noindent {\it Remark.} We call a function $f:A\to B$ {\bf conformal} if
it is the restriction of an analytic bijection $f:U\to V$ between open sets $U\supset A$ and $V\supset B$, and $f(A)=B$.

\begin{lemma} \label{prop:balls}
Suppose that (A) holds with $U_0\ni \Eat$ and
$\rho \in (0,4/5)$ so small that $C_\gamma\rho^\mu <1$ and
$\overline{B_\rho(\Eat)}\subset U_0$. Suppose that $\alpha_0 <
\rho/2$, $\beta_0, \gamma_0 \leq \rho/(8C_\chi)$ and that
\eqref{eq:4446} holds. If $z\mapsto H^{(0)}[z]\in \LL(\HH_{\rm at})$ is
analytic in $U_0$ and $$H^{(0)}[z]-(\Eat-z)\in\BB(\alpha_0,\beta_0,\gamma_0)$$ for all $z\in U_0$, then the following is true.
\begin{itemize}
\item[(a)]
For $n \geq 0$, $E^{(n)}:U_n \to \C$ is analytic in $U_n^{\circ}$ and a
conformal map from $U_{n+1}$ onto $D_{\rho/2}$. In particular, $E^{(n)}$ has a unique
zero, $z_n$, in $U_n$. Moreover,
$$
B_\rho(E_{\rm at}) \supset U_1 \supset U_2 \supset U_{3} \supset \cdots \; .
$$
\item[(b)] The limit $z_\infty := \lim_{n\to\infty} z_n$ exists and  for
$\epsilon := 1/2 - \rho/2 - \alpha_1 > 0$,
$$
| z_n - z_\infty | \leq {\rho^{n}}  \exp\left( \frac{1}{2
 \rho \epsilon^2} \sum_{k=0}^\infty \alpha_k \right) \; .
$$
	\item[(c)] Let $E_{\rm at} \in \R$ and
		$H^{(0)}[z]^* = H^{(0)}[\overline{z}]$ for all
		$z \in  B_\rho(E_{\rm at})$.
 Then for all $n\geq 0$, $ U_{n+1} \cap \R$ is an interval and $\partial_x E^{(n)}(x) < 0$  on $U_{n+1} \cap \R$. 
		Then there exists an
		$a < z_\infty$ such that $H^{(0)}[x]$ has a bounded
		inverse for all $x \in (a,z_\infty)$.
\end{itemize}
\end{lemma}

\begin{proof}[Proof of Lemma \ref{prop:balls}]
The Lemma follows  as a consequence of  Lemma \ref{cor:bcfs} and the
 property of the Feshbach map, cf. Theorem    \ref{thm:isospectralFeshbach}. The details of the proof 
are  the same as the proofs  of Lemma~19, Corollary~20, and Proposition~21 in \cite{GriHas09}.
 \end{proof}

Let us now   discuss the construction of  an eigenvector $\varphi^{(0)}$ such that {$H^{(0)}[z_\infty] \varphi^{(0)}= 0$}.
The same construction has been used in \cite{BacFroSig98-1,BacFroSig98-2,BCFS,GriHas09}. The result 
which we use is from \cite{GriHas09}.
In order to formulate the result we define the following auxiliary operator for $z \in U_n$
\begin{equation*}
	Q_n[z] := \chi_\rho - \overline{\chi}_\rho
	\Big(H_{0,0}^{(n)}[z] + \overline{\chi}_\rho\, W^{(n)}[z]
	\,\overline{\chi}_\rho \Big)^{-1}
	%\Big(H_{\overline{\chi}_\rho}^{(n)}[z]\Big)^{-1}
	\overline{\chi}_\rho \,W^{[n]}[z] \, \chi_\rho \, ,
\end{equation*}
where $W^{(n)}[z]$ and $H_{0,0}^{(n)}[z]$ are given  as follows. By construction of  $H^{(n)}[z]$ there exists 
by Proposition \ref{H-is-bounded} 
a  unique $w^{(n)}[z] \in \mathcal{W}_\xi$  such that $H^{(n)}[z] = H(w^{(n)}[z])$.  Then we set  $H_{0,0}^{(n)} := H_{0,0}(w^{(n)}[z])$
and
$W^{(n)}[z] := H^{(n)} [z]- H_{0,0}^{(n)}[z]$.  

\begin{theorem}[Theorem 22, \cite{GriHas09}]
\label{thm:eigenvector}
Suppose the assumptions of Lemma~\ref{prop:balls} hold.
Then for any  nonzero vector $v \in \C^d$ 
\begin{align} \label{eq:varphidefeigem} 
\varphi_v^{(0)} := \lim_{n \to \infty} Q_0[z_\infty] \,
\Gamma_\rho^*\,Q_1[z_\infty]\, \cdots \,
	\Gamma_\rho^* \,Q_n[z_\infty] \, (  v \otimes \Omega )
\end{align}
exists, $\varphi_v^{(0)} \neq 0$ and
$H^{(0)}[z_\infty] \,\varphi_v^{(0)} = 0$.
Moreover,
\begin{align} \label{convestforeigenstate} 
\Big\|\, \varphi_v^{(0)} - Q_0[z_\infty] \,\Gamma_\rho^*\,
	Q_1[z_\infty ] \,\cdots \,
	\Gamma_\rho^* \,Q_n[z_\infty ] \, ( v \otimes \Omega ) \, \Big\|
	\leq C \sum_{l=n+1}^\infty \gamma_l \,,
\end{align}
%\marginpar{\textcolor{red}{what are the $\gamma_l$'s?}}
where \begin{equation} \label{defofCinconves} 
   C= \frac{8}{\rho}\frac{\xi}{1-\xi}
   \exp\left(\frac{8}{\rho}\frac{\xi}{1-\xi}\sum_{n\geq 0}\gamma_n\right).\\
\end{equation}
\end{theorem}

\begin{proof} The proof follows    from   Lemma~\ref{prop:balls}    with the help of  Lemma~\ref{feshbachtest}  and  Theorem    \ref{thm:isospectralFeshbach}. 
The details of the  proof  carry over from the proof of Theorem 22 in  \cite{GriHas09} by merely replacing $\Omega$ by $v \otimes \Omega$. 
\end{proof} 

\begin{remark} \label{rem:thm:eigenvector} Let the assumptions and notations be as in  Theorem \ref{thm:eigenvector}.
It follows immediately from  \eqref{eq:varphidefeigem} 
 that the map $\C^d \to \HH_{\rm red}$,  $v \mapsto \varphi_v^{(0)}$ is linear. 
 Since by Theorem \ref{thm:eigenvector} that map  has  kernel $\{ 0 \}$, 
it is injective. 
\end{remark}

\section{Analyticity of Eigenvalues and Eigenvectors}
\label{sec:analyt-ev}

This section is devoted to the proof of Theorem~\ref{thm:symdegenSpinBoson}. It is
essential for this proof, that a neighborhoods $V_0\subset V$ of $s_0$ and a positive
bound, $g_1$, on $g$ can be determined in such a way that the renormalization analysis of
Sections~\ref{sec:iterationSymDegen}, and in particular the choices of
$\rho$ and $\xi$ are independent of $s\in V_0$ and $g\leq g_1$. Once $V_0$ and $g_1$
are found, the assertions of Theorem~\ref{thm:symdegenSpinBoson} are derived from
Proposition~\ref{prop:Feshbachanalytic} and \ref{prop:Feshbachcommutes} as well as  the uniform bounds of Sections~\ref{sec:iterationSymDegen}.

\begin{proof}[Proof of Theorem \ref{thm:symdegenSpinBoson}]  
First let us recall that by Lemma  \ref{lem:wlog} 
we can assume without loss that  
Hypothesis~\ref{Hypo5} holds and $P_{\rm at}(s) = P_{\rm at}(s_0) $ for all $s \in X$. Furthermore
by choosing a suitable basis we can assume that  $\ran P_{\rm at}(s_0) = \C^d$. \\

%\noindent

Let $\mu>0$ and $\UU\subset\C^{\nu + 1}$ be given by Hypothesis~\ref{Hypo1} and Hypothesis~\ref{Hypo3}, respectively.  For the
renormalization procedure to work, we first choose $\rho\in(0,4/5)$
and a open neighborhood $X_{\rm b} \subset X_1$ of $s_0$, both small enough,
so that $C_{\gamma}\rho^{\mu}<1$ and
\begin{equation} \label{eq:set}
%\{ s \} \times B(E_{\rm at}(s), \rho) \subset \UU \; .
\overline{B_\rho(E_{\rm at}(s))} \subset \{z : (s,z)\in
\UU\},\qquad \text{if}\ s\in X_{\rm b}, 
\end{equation}
which is possible since $s \mapsto E_{\rm at}(s)$ is continuous.
Here, and below we use the constants $C_{\gamma}, C_{\chi}$ and $C_{\beta}$ from Theorem~\ref{bcfssigal}. 
Let $\xi=\sqrt{\rho}/(4C_{\chi})$.
Next we pick small positive constants
$\alpha_0$, $\beta_0$, and $\gamma_0$ such that
\begin{equation}\label{eq:abc}
   \alpha_0 < \frac{\rho}{2}, \qquad \beta_0 \leq \frac{\rho}{8C_{\chi}},
   \qquad \gamma_0 \leq \frac{\rho}{8C_{\chi}},
\end{equation}
and in addition
\begin{equation} \label{le:zinfty2}
   \beta_0 + \frac{C_{\beta}/\rho}{1-(C_{\chi}\rho^\mu)^2 }\gamma_0^2 \leq
   \frac{\rho}{8C_{\chi}}.
\end{equation}
By Proposition~\ref{thm:Feshbachpairchp41} and Theorem~\ref{thm:iniFeshbachBall}, there exists a $g_1>0$ such that for $0 \leq g \leq g_1$
$$
H_g^{(0)}[s,z] - ( E_{\rm at}(s) - z ) \in \mathcal{B}(\alpha_0,
\beta_0, \gamma_0 ), \qquad \text{for}\ (s,z) \in \UU,
$$
where $H_g^{(0)}[s,z]$ is analytic on $\UU$, by Theorem~\ref{thm:feshbachpair}. We define
\begin{eqnarray*}
   \UU_0 &:=& \UU \\
   \UU_n &:=& \{ (s,z) \in \UU_{n-1}: | E^{(n-1)}(s,z) | \leq \rho/8 \}.
\end{eqnarray*}
and
$$
    U_n(s) := \{ z :  (s,z) \in \UU_n \}, \qquad n\in \N.
$$
Then, by \eqref{eq:abc}, \eqref{le:zinfty2}, and \eqref{eq:set} the
assumptions of Lemma~\ref{prop:balls} are satisfied for $s\in X_{\rm b}$
and $U_0=U_{0}(s)$. It follows that, for all $n\in \N$,
$H_g^{(n)}[s,z]=\RR^n H_g^{(0)}[s,z]$ is well-defined for
$(s,z)\in \UU_n$, and that $U_n(s)\neq \emptyset$. By
Proposition~\ref{prop:Feshbachanalytic}, $H_g^{(n)}[s,z]$ is analytic in $ \UU_n^{\circ}$.\\

\noindent \underline{Step 1}: $z_\infty(s)=\lim_{n\to\infty}z_n(s)$
exists and is analytic on $X_{\rm b}$.  \\

Since $H_g^{(n)}[s,z]$ is analytic on $\UU_n^{\circ}$, so is $E_g^{(n)}(s,z)$. Let $z_{n}(s)$ denote the
unique zero of the function $z\mapsto E_g^{(n)}(s,z)$ on $U_n(s)$ as
determined by Lemma~\ref{prop:balls}. That is,
$$
E_g^{(n)}(s,z_{n}(s)) = 0.
$$
By the implicit function theorem  $z_n(s)$ is analytic in $s$. The
application of the implicit function theorem is justified since $z
\mapsto E^{(n)}_g(s,z)$ is bijective in a neighborhood of $z_n(s)$,
and thus in this neighborhood $\partial_z E_g^{(n)}(s,z)\neq 0$. By
Lemma~\ref{prop:balls} {\it (b)}, $z_n(s)$ converges to
$z_\infty(s)$ uniformly in $s \in X_{\rm b}$.  This implies the
analyticity of $z_\infty(s)$ on $X_{\rm b}$, by the Weierstrass
approximation theorem of complex analysis.\\

\noindent \underline{Step 2}:  For $s \in X_{\rm b}$, there exist $d$ linearly independent 
eigenvectors  $\psi_{g,j}(s)$,  $j=1,...,d$,   of $H_g(s)$ with eigenvalue $z_\infty(s)$, such
that $\psi_{g,j}(s)$  depends analytically on $s$. \\

 Since $H_g^{(n)}[s,z]$ is analytic on $\UU_n^{\circ}$, it follows, by Proposition~\ref{prop:Feshbachanalytic} , that
$$
   Q_{g,n}[s,z] = \chi_\rho(s) - \overline{\chi}_\rho(s)
   {H^{(n)}_{g,\overline{\chi}_\rho}[s,z]}^{-1} \overline{\chi}_\rho(s)
   W_g^{(n)}[s,z]
   \chi_\rho(s)
$$
is analytic  on $\UU_n^{\circ}$, where $W_g^{(n)} := H_g^{(n)} -
H^{(n)}_{g,0,0}$. Hence, by Step~1, $s \mapsto Q_{g,n}[s,z_\infty(s)]$ is analytic on $X_{\rm b}$.
Let $e_1,...,e_d$ be a basis of $\C^d$. 
 It follows that
$$
\varphi^{(0,n)}_{g,j}(s):=Q_{g,0}[s,z_\infty(s)] \Gamma_\rho^* Q_{g,1}[s,z_\infty(s)]\dots
\Gamma_\rho^* Q_{g,n}[s,z_\infty(s)]  ( e_j \otimes \Omega ) 
$$
is analytic on $X_{\rm b}$. From Theorem~\ref{thm:eigenvector} we know that
these vectors converge uniformly on $X_{\rm b}$ to  a vector $
\ph_{g,j}^{(0)}(s)\neq 0$ and that $H_g^{(0)}[s,z_{\infty}(s)] \varphi_{g,j}^{(0)}(s)=
0$. Hence $\ph^{(0)}_{g,j}(s)$ is analytic on $X_{\rm b}$ and, by the Feshbach property (Theorem \ref{thm:feshbachpair} (c)), the vector
$$
\psi_{g,j}(s) = Q_{\boldsymbol{\chi}}(s,z_\infty(s)) \varphi_{g,j}^{(0)}(s)
$$
is an eigenvector of $H_g(s)$ with eigenvalue $z_\infty(s)$. Using  Theorem  \ref{thm:feshbachpair} (a)  and again by Step 1 
we see that 
$s \mapsto Q_{\boldsymbol{\chi}}(s,z_\infty(s))$ is analytic  on $X_{\rm b}$.  We conclude that $\psi_{g,j}$ is analytic on
$X_{\rm b}$ as well. The linear independence of $\psi_{g,j}(s)$, $j=1,...,d$,   follows from Remark \ref{rem:thm:eigenvector}  and 
Theorem \ref{thm:feshbachpair} (c). 
\\

\noindent
\underline{Step 3:} In the limit $g \to 0 $,  we have uniformly in $s \in X_{\rm b}$ that  $|z_\infty(s) - E_{\rm at}(s)| = o(1)$ and  that 
 $\| \psi_{g,j}(s)  - \varphi_{{\rm at},j}(s) \otimes \Omega \| = o(1)$ for some  $\varphi_{{\rm at},j}(s) \in \ran P_{\rm at}(s)$ .  \\

From Lemma \ref{prop:balls} we know that $z_\infty(s) \in B_\rho(E_{\rm at}(s))$.  Now by Theorem \ref{thm:iniFeshbachBall} we can
make $\alpha_0,\beta_0,\gamma_0$ arbitrarily small by choosing $g_{\rm b} > 0$ sufficiently small.  Thus from 
  \eqref{eq:abc}  we see that we can choose $\rho \in (0,1)$ arbitrarily small by choosing $g_{\rm b} > 0$ sufficiently small.
This shows $|z_\infty(s) - E_{\rm at}(s)| = o(1)$ uniformly in $X_{\rm b}$. 
From \eqref{convestforeigenstate} of Theorem  \ref{thm:eigenvector} we find $\| \psi_{g,j}(s)  - e_j \otimes \Omega \| \leq C \sum_{l=0}^\infty \gamma_l $
with $C$ given in \eqref{defofCinconves}.  
Now from Eq. \eqref{eq:cond1} of Lemma \ref{cor:bcfs} we see that the right hand side can be made arbitrarily small if $\gamma_0 > 0$ is 
sufficiently small. But  by Theorem \ref{thm:iniFeshbachBall} the latter can be made  small  by choosing $g_{\rm b} > 0$ sufficiently small.
This shows that  $\| \psi_{g,j}(s)  -  e_j \otimes \Omega \| = o(1)$ uniformly in $s$.  \\

\noindent\underline{Step 4}:
If in addition Hypothesis~\ref{Hypo4} holds, then   
\begin{itemize}
\item[($\alpha$)]  for all $s \in X_b \cap \R^\nu$ it holds that $z_\infty(s) = \inf \sigma(H_g(s))\,,$
\item[($\beta$)]   for all $s \in X_b \cap X_b^*$ it holds that $\overline{z}_\infty(s) =  z_\infty(\overline{s})$. 
\end{itemize}

Let $s \in X_{\rm b}\cap\R^{\nu}$. Then by Hypothesis~\ref{Hypo4} the operator   $H_g(s)$ is self-adjoint and its spectrum is a
half line $[\Sigma_g(s),\infty)$ (cf. \cite{Spo04}
), where $\Sigma_g(s) :=  \inf \sigma(H_g(s))$. By Step~2, $z_{\infty}(s)\geq \Sigma_g(s)$. We use
Proposition~\ref{prop:balls} (c) to show that $z_{\infty}(s)>\Sigma_g(s)$ is impossible.
Clearly $\Eat(s)\in\R$, and $H_g^{(0)}[s,z]^* = H_g^{(0)}[s,\overline{z}]$ for $z
\in B_\rho(E_{\rm at}(s))$ is a direct consequence of the definition of
$H_g^{(0)}$ and the self-adjointness of $H_g(s)$. Hence there exists a number $a(s) < z_\infty(s)$
such that $H^{(0)}_g[s,x]$ has a
bounded inverse for all $x \in (a(s) ,z_\infty(s))$. It follows, by
Theorem~\ref{thm:isospectralFeshbach}, that  $(a(s),z_\infty(s))\cap\sigma(H_g(s))=\emptyset$. Therefore $z_{\infty}(s)=\Sigma_g(s)$.
This shows ($\alpha$). Now ($\beta$) is a consequence of   Schwarz reflection principle.

\vspace{0.4cm}
\noindent
The Theorem now follows  for  $E_g(s) = z_{\infty}(s)$.
\end{proof}

If we  neglect the first Feshbach map  in the above  proof,    we obtain  the following theorem, 
which is independent of the explicit structure of the Hamiltonian.

\begin{theorem} \label{genrenthm} Suppose $\HH_{\rm red} = \C^d \otimes \FF$ with $ d \in \N$. Let $\mathcal{S}$ be a group of symmetries 
acting on $\HH_{\rm red}$ commuting with dilations  and $H_{\rm f}$ and $\mathcal{S}_1$ acts irreducibly on $\C^d$. 
For  $\mu > 0$ and $\rho  \in (0,1/2) $,  there exist positive  numbers $\alpha_0, \beta_0, \gamma_0$  with the following properties.  
Let $X$ be a nonempty subset of $\C^d$,  $e : X \to \C$ a function, and 
  $\mathcal{U} \subset X \times \C$ a    set such that 
$$
\overline{B_\rho(e(s))} \subset \{ z : (s,z) \in \mathcal{U} \}  \subset  B_{1/2}(e(s))  \text { for all }  s \in X . 
$$
Suppose for each $(s,z) \in \mathcal{U}$ an operator $H(w[s,z])$  on $\HH_{\rm red}$, with $w[s,z] \in \mathcal{W}_\xi$ is given which is   symmetric with respect to  $\mathcal{S}$ such that
 \begin{align*} %\label{firstfeshbanachestimate1} 
	{H}(w[s,z]) - (e(s) - z)
		\in \mathcal{B}
		(\alpha_0,\beta_0,\gamma_0)\,, \quad \forall (s,z) \in \mathcal{U} .
 \end{align*}

Then for each $s \in X$  there exists an element  $z_\infty(s) \in \overline{B_\rho(e(s))} $  and linearly independent 
functions $\varphi_j(s)$, $j=1,...,d$,  such that 
$$ H(w[s,z_\infty(s)]) \varphi_j(s)= 0 . $$ 
\begin{itemize} 
\item[(i)]  There exists an $e_j \in \C^d$. So that for  any $\epsilon > 0$ there exists  $(\alpha_1,\beta_1,\gamma_1)  \in (0,\alpha_0] \times 
 (0,\beta_0]  \times  (0,\gamma_0] $
 such that 
$|z_\infty(s) - e(s) | < \epsilon$ and $\| \varphi_j(s) - e_j \otimes \Omega \| < \epsilon $ 
whenever $ H(w[s,z]) -  ( e(s) - z )  \in \mathcal{B}(\alpha_1,\beta_1,\gamma_1) $. 
\item[(ii)]  If $X$ and $\mathcal{U}$ are open and $e$ and $H(w)$ analytic on $X$ and $\mathcal{U}$, respectively, then 
also   $z_\infty(s)$ and $\varphi_j(s)$ depend analytically on $s$. 
\end{itemize} 
\end{theorem} 

\begin{proof} This follows from the same  Proof as  Theorem \ref{thm:symdegenSpinBoson} by neglecting the first step. 
\end{proof}

%\section*{Acknowledgements}
\noindent{\bf Acknowledgements.}
Both authors acknowledge financial support by the Research Training Group (1523/2) “Quantum and Gravitational Fields” when this project was initiated. 
D. Hasler wants to thank M. Griesemer and I. Herbst for valuable discussions on the subject. 
M. Lange  also acknowledges financial support from the European Research Council (ERC) under the European Union’s Horizon 2020 research and innovation programme (ERC StG MaMBoQ, grant agreement n.802901).

\begin{appendix}

\section{Symmetries} \label{app:symm}
In this section we introduce anti-linear operators and symmetries in  a Hilbert space $\HH$.

\begin{definition}
Let $\HH$ be a complex Hilbert space.
\begin{itemize}
\item[(a)]
A mapping $T : \HH \to \HH$ is called {\bf anti-linear} operator  in  $\HH$  if
$$
T(\alpha x + \beta y) = \overline{\alpha} T x + \overline{\beta}T y ,
$$
for all $\alpha, \beta \in \C$ and $x,y \in \HH$. An anti-linear operator $T$ is called bounded if $$\sup_{x: \| x \| \leq 1 } \| Tx \| < \infty . $$
\item[(b)]
The {\bf adjoint} of a bounded anti-linear operator,  $T : \HH \to \HH$, is defined to be
the anti-linear operator $T^* : \HH \to \HH$ such that
$$
\langle x , T y \rangle = \overline{ \langle T^* x , y \rangle }
$$
for all $x, y \in \HH$.
\item[(c)] An  anti-linear operator  $V$  in $\HH$  is called {\bf antiunitary} if it is surjective and satisfies
$$
\langle V x , V y \rangle = \overline{ \langle  x , y \rangle }
$$
for all $x, y \in \HH$.
\end{itemize}
\end{definition}

In the following lemma we collect a few properties of anti-linear and antiunitary operators. 

\begin{lemma} \label{lem:prodanti}  Let $\HH$ be a complex Hilbert space.
Then the following holds.
\begin{itemize}
\item[(a)] Let $S$ and $T$ be a linear or an anti-linear  operator in $\HH$. Then $ST$ is linear if either both  $S$ and $T$ are  
linear or both $S$ and $T$ are anti-linear.  The operator $ST$ is anti-linear if one of the two operators $S$ and $T$ is linear and the other is anti-linear.
\item[(b)] Let $S$ and $T$  be anti-linear. Then $(\alpha S + \beta T)^* = \overline{\alpha} S^* + \overline{\beta} T^*$.
\item[(c)] Let $S$ and $T$ be linear or anti-linear.   Then we have $(ST)^* = T^* S^*$.
\item[(d)]  A bounded  anti-linear operator  $T$ is   antiunitary if and only if it   satisfies  $T^* T = 1 $ and $T T^* = 1$.
\item[(e)]  Let $S$ and $T$ be unitary or antiunitary.  Then $ST$ is unitary if either both $S$ and $T$ are unitary or both $S$ and $T$ are antiunitary. The operator $ST$ is antiunitary if one of the two operators $S$ and $T$  is unitary and the other is antiunitary.
\end{itemize}
\end{lemma}
\begin{proof} (a) and (b) are elementary to show. \\ (c) If $S$ and $T$ are linear, this is a well known identity. If  $S$ is linear and $T$ is antilinear,   then 
for all $x, y \in \HH$ 
\begin{align*}
\overline{\inn{ (ST)^* x , y }} = \inn{ x , S T y } = \inn{ S^* x , T y } = \overline{ \inn{ T^* S^* x ,  y }  }  
\end{align*} 
and so $(S T)^* = T^* S^*$ by the nondegeneracy of the inner product. 
If  $S$  and $T$ are  antilinear,   then $ST$ is linear by (a)   and 
for all $x, y \in \HH$ 
\begin{align*}
\inn{ (ST)^* x , y } = \inn{ x , S T y } = \overline{\inn{ S^* x , T y }} =  \inn{ T^* S^* x ,  y }    
\end{align*} 
and so $(S T)^* = T^* S^*$ by the nondegeneracy of the inner product. \\
(d) and (e) are elementary to show. 
\end{proof}

\begin{definition} Let $\HH$ be a complex Hilbert space.
\begin{itemize}
\item[(a)] A {\bf symmetry} in $\HH$ is a unitary or antiunitary operator in $\HH$.
\item[(b)] We say that $S$ is a {\bf symmetry of a linear operator}  $T$  in $\HH$ (possibly unbounded)  if
\begin{align*}
 S T S^* & = T , \quad \text{if } S \text{ is unitary } \\
 S T S^* & = T^* , \quad \text{if } S \text{ is antiunitary }.
\end{align*}
In that case, we also say that $T$ is {\bf symmetric or invariant with respect} to $\mathcal{S}$.
\end{itemize}
\end{definition}

We note that it is elementary to show %using Lemma \ref{lem:prodanti}
 that    the set of symmetries of an operator form a group.

\begin{lemma}
	Let $\HH$ be a complex Hilbert space. Then the set of symmetries of an operator in $\HH$ form a group.
\end{lemma}

\begin{proof} 
If  $S_1$ and $S_2$  are  symmetries, then we see from Lemma  \ref{lem:prodanti}  (c), (d), and (e) that  also    $S_1 S_2$ and $S_1^{-1}$
are symmetries.  
\end{proof}

\section{Eigenprojections and their properties} 
\label{sec:eigenproj}

In this appendix  we  recall well-known
properties about isolated points of the spectrum.   For a detailed treatment
 we refer the
reader to  the discussion in  \cite{ReeSim4} surrounding  Theorems XII.4 and XII.5.

\begin{theorem}\label{thm:ReeSim4XII5}
Suppose that $A$ is a closed operator with $\{ z \in \C : |z- \lambda | = r \} \subset \rho(A)$ for some $r> 0$. Then
$$ %\begin{equation}
P := - \frac{1}{2 \pi i} \ointctrclockwise_{|\mu-\lambda|=r} (A-\mu)^{-1} d \mu
$$ %\end{equation}
 and $\overline{P} :=1 - P$ are  bounded projections with  the following properties.
\begin{itemize}
\item[(a)] The ranges of $P$ and $\overline{P}$ are complementary closed subspaces, that is $  \ran P +   \ran  \overline{P} = \HH$ and
$\ran P \cap   \ran  \overline{P} = \{ 0 \}$. Moreover, $A$ leaves this subspaces invariant.
More precisely, $\ran P \subset D(A)$, $A \ran P \subset \ran P$,
$\ran \overline{P} \cap D(A)$ is dense in $\ran \overline{P}$, and $A \left[ \ran \overline{P} \cap D(A) \right] \subset \ran \overline{P}$. %_\lambda.
\item[(b)] For    $|z - \lambda| \neq r$
\begin{equation*} %\label{eq:redres}
\hat{R}_z :=   - \frac{1}{2 \pi i} \ointctrclockwise_{|\mu-\lambda|=r} (z - \mu)^{-1} (A-\mu)^{-1} d \mu
\end{equation*}
exists and we have the following two cases.
\begin{itemize}
\item[(i)]
If $|z-\lambda|< r$, then  $(A-z)|_{\ran \overline{P} \cap D(A)}$ is invertible and
$$\hat{R}_z = ((A-z)|_{\ran \overline{P} \cap D(A)})^{-1}  \overline{P}, $$ i.e.,
$\hat{R}_z P = P \hat{R}_z = 0$,
$(A-z) \hat{R}_z = \overline{P}$,  and $\hat{R}_z (A-z) = \overline{P}$.
\item[(ii)]
If $|z-\lambda| > r$, then $(A-z)|_{\ran {P}}$ is invertible and
$$\hat{R}_z = ((A-z)|_{\ran{P}})^{-1}  {P},$$ i.e.,  $\hat{R}_z \overline{P} = \overline{P} \hat{R}_z = 0$,
$(A-z) \hat{R}_z = -P$,  and $\hat{R}_z (A-z) = -P$.
\end{itemize}
\item[(c)] We have  $\sigma(A) \cap B_r(\lambda) = \sigma(A|_{\ran P})$ and
$\sigma(A) \setminus B_r(\lambda)  =  \sigma(A|_{\ran \overline{P} \cap D(A)})$.
\item[(d)] If $\lambda$ is an isolated element of the spectrum $\sigma(A)$ its algebraic multiplicity is greater or equal to
its  geometric multiplicity. 
\end{itemize}
\end{theorem}

\begin{proof}[Sketch of Proof.] (a) $\ran P \subset D(A)$ follows by expressing the integral as a limit
of Riemann sums, using that $A$ is closed and the identity $A (A-\mu)^{-1} = 1 + \mu (A-\mu)^{-1}$.
The remaining properties are elementary to verify, for details see  \cite{ReeSim4} Theorem XII.6 (or more precisely 
 \cite[Theorems XII.5 (b)]{ReeSim4}   whose proof 
carries through without change).  \\
(b) The algebraic identities are straight forward to verify. They then imply the property
about the invertibility. \\
(c) For all $z \in \rho(A)$  it follows from (a) that
 $(A-z)^{-1} =  (A-z)^{-1} P +   (A-z)^{-1} \overline{P} =   ((A-z)|_{\ran  P})^{-1} P +  ( (A-z)|_{\ran \overline{P}  \cap D(A)})^{-1} \overline{P} $. In view of  this identity  the claim now follows from (b). \\
(d) Let $\lambda$ be an isolated element of the spectrum. As in Theorem \ref{thm:ReeSim4XII5FirstPart}
choose $\epsilon > 0$  such that $\{ \lambda \} = \sigma(A) \cap B_\epsilon(\lambda)$.
Let  $(A - \lambda ) v = 0$. Then for every $r \in (0,\epsilon)$
\begin{align*}
P_\lambda v  & = - \frac{1}{2 \pi i} \ointctrclockwise_{|\mu-\lambda|=r} (A-\mu)^{-1} v  d \mu \\
& =  - \frac{1}{2 \pi i} \ointctrclockwise_{|\mu-\lambda|=r} (A-\mu)^{-1} \frac{( A - \mu)}{\lambda - \mu}  v  d \mu \\
& =  - \frac{1}{2 \pi i} \ointctrclockwise_{|\mu-\lambda|=r} \frac{1}{\lambda - \mu}  v  d \mu = v ,
\end{align*}
and so $v \in \ran P_\lambda$. 
\end{proof}

\begin{proposition} \label{GriesemerHaslerProp27} Let $R \ni s \mapsto T(s)$ be an analytic family. Suppose 
there is a non-defective  eigenvalue $E(s)$  isolated from the rest of the spectrum with analytic projection
operator $P(s)$. Let $\overline{P}(s) = 1 - P(s)$ and let 
\begin{align*}
\Gamma  & := \{(s,z) \in R \times \C : T(s) - z :   D(T(s)) \cap \ran \overline{P}(s) \to  \ran \overline{P}(s) \text{ is bijective }  \}
\end{align*} 
Then $\Gamma$ is open and $(s,z) \mapsto (T(s) - z)^{-1} \overline{P}(s)$ is analytic on $\Gamma$. 
\end{proposition}

\begin{proof} Let $(s_0,z_0) \in \Gamma$. There exists in a
neighborhood of $s_0$ and a bijective operator $U(s): \HH \to \HH$,
analytic in $s$, such that $U(s) P(s) U(s)^{-1} = P(s_0)$ and hence
$U(s) \overline{P}(s) U(s)^{-1} = \overline{P}(s_0)$, (cf. 
\cite[Thm. XII.12]{ReeSim4}). The operator $\widetilde{T}(s) =
U(s) T(s) U(s)^{-1}$ is an analytic family. It leaves the closed space $\ran
\overline{P}(s_0)$ invariant and thus $\widetilde{T}(s) |_{ 
\ran \overline{P}(s_0)} : \ran \overline{P}(s_0) \cap
D(\widetilde{T}(s)) \to \ran \overline{P}(s_0)$ is an
analytic family as well. By this and the fact that
$(\widetilde{T}(s_0) - z_0 ) |_{ \ran \overline{P}(s_0)}$  is
bijective  since $(s_0,z_0) \in \Gamma$, it
follows by \cite[Thm. XII.7]{ReeSim4}  that in a neighborhood of
 $(s_0,z_0)$, the operator  $(\widetilde{T}(s) - z ) |_{ 
\ran \overline{P}(s_0)}$ is bijective  and
$(\widetilde{T}(s) - z )^{-1} \overline{P}(s_0)$ is analytic in both
variables. Thus in this neighborhood also the linear operator   
$(T(s) - z ) |_{ \ran \overline{P}(s)} = U(s)^{-1}  (\tilde{T}(s) - z )  U(s) |_{ \ran \overline{P}(s)}  $ is bijective and 
$(T(s)-z)^{-1} \overline{P}(s) =  U(s)^{-1}   (\widetilde{T}(s) - z )^{-1} \overline{P}(s_0) U(s)$
 is an analytic function of two variables.
\end{proof}

\section{Field~operators, Elementary estimates and identities }
\label{sec:tecAux}
We consider the Hilbert space
$\HH = \HH' \otimes \FF$ consisting of a
separable Hilbert space $\HH'$ and the
bosonic Fock Space $\FF$.

Let ${X} := \R^3 \times \Z_2$.
For a separable Hilbert space $\HH'$ we define for $n \geq 1$ 
\begin{align*} 
 L^2_s( X^n ; \HH'  ) 
:= \{ \varphi \in L^2(X^n ; \HH') : \varphi(k_1,...,k_n) = \varphi(k_{\sigma(1)},...,k_{\sigma(n)}) , \sigma \in \mathfrak{S}_n \} 
\end{align*} where $\mathfrak{S}_n $ denotes the set of permutations of $\{1,...,n\}$. 
We set  $L^2_s( X^0 ; \HH'  ) := \HH'$.  
 We shall use 
the canoncial  identification \cite{ReeSim1}
\begin{align*}
\HH' \otimes \FF = \bigoplus_{n=0}^\infty  L^2_s( X^n ; \HH'  ) \,.
\end{align*}
 For $G \in L^2(\R^3\times \Z_2 ; \mathcal{L}(\HH'))$
the creation operator $a^*(G)$ is by defnition  the adjoint of $a(G)$,
cf. \eqref{eqdefofaG}. The domain of the creation operator contains the so called
finite particle vectors 
  $\psi = (\psi_n)_{n =0}^\infty \in \HH' \otimes \FF$ with the property that $\psi_n = 0$ for all
but finitely many $n$,  and $a^*(G) \psi$ is a  sequence of  
 $\HH'$-valued measurable functions such for  $n$-th term 
\begin{align} \label{eqdefofastarG}
[a^*(G) \psi]_n(k_1,....,k_n)  =  n^{-1/2} \sum_{j=1}^n \int  G(k_j) \psi_{n-1}(k_1,..., \widetilde{k}_j,..,k_n)  dk  , 
\end{align} 
where  $\widetilde{ \ }$ means that this variable is to be omitted and 
 the integral on the right hand side is defined as a Bochner integral. 
A straight forward calculation using  \eqref{eqdefofaG}  and \eqref{eqdefofastarG}   shows that on finite particle vectors
we have the commutation relations  
\begin{align*}
[a(F),a^*(G)] = \int F^*(k) G(k) dk , \quad    [a(F),a(G)] = 0 , \quad  [a^*(F),a^*(G)] = 0 , 
\end{align*} 
which extend to their natural domains. 

 Next we  express the creation and annihilation operator in terms
of so called operator valued distributions,  $a^*(k)$ and $a(k)$. 
For an element  $\psi \in \HH' \otimes \mathcal{F}$ 
we define $a(k) \psi$  for a.e. $k \in \R^3 \times \Z_2$ as the 
sequence of $\HH'$-valued measurable functions  such that  the $n$-th  term  satisfies  a.e. 
 \begin{align} \label{eq:formalani} 
[a(k) \psi]_n(k_1,....,k_n) :=  (n+1)^{1/2} \psi_{n+1}(k,k_1,....,k_n) .
\end{align} 
Moreover, using Fubini's theorem
\cite[Theorem~I.21]{ReeSim1},
it is elementary to see that the vector-valued map 
 $k \mapsto  a(k) \psi$ is
an element of $L^2(X; \HH' \otimes \FF)$. 
For $G \in L^2(\R^3\times \Z_2 ; \mathcal{L}(\HH'))$ we obtain 
the following identity 
$$
a(G) = \int G^*(k) a(k) dk   , 
$$
which holds on finite particle vectors.    
The creation  operator valued distribution $a^*(k)$ 
is defined as the adjoint of $a(k)$  in the sense of  forms, i.e., we define the form 
$\inn{ \varphi , a^*(k) \psi } := \inn{ a(k) \varphi , \psi }$ for  smooth finite particle vectors 
$\varphi,\psi$. 
 On such vectors one obtains the following identity in the sense of forms and 
weak integrals 
$$
 a^*(G) = \int G(k) a^*(k) dk . 
$$
Using  \eqref{eq:formalani} we can express the free field energy in 
terms of the following  identity on  vectors 
 $\varphi, \psi \in D(H_{\rm f})$  
\begin{equation}  \label{eq:fieldenergy2}
\inn{ \varphi , H_{\rm f} \psi}  = \int  \omega(k) \inn{ a(k)  \varphi , a(k) \psi }   dk . 
\end{equation}

We use the following estimates on
multiple occasions in this paper.
They establish well known  elementary estimates for the
annihilation and creation operators introduced  following 
Eq. \eqref{eqdefofaG}. 

\begin{lemma}\label{lem:bddforHf}
For
$G \in L^{2}(\R^3 \times \Z_2 ;
	\mathcal{L}(\HH'))$
we have
\begin{align} \label{eq:estoncrea}
\|\, a(G)\,H_{\rm f}^{-1/2}\, \|
	&\leq \| \,\omega^{-1/2} G\, \| \, , \\
\|\, a^*(G)\,(H_{\rm f}+1)^{-1/2}\, \|
	&\leq \| \,(\omega^{-1} +1)^{1/2} G \,\| \, .
	\nonumber
\end{align}
\end{lemma}
\begin{proof}
By density it suffices to show the identities for 
smooth finite particle vectors  $\psi \in  \HH' \otimes \FF$.
In order to prove the first inequality we estimate
\begin{align*}
\big\| a(G) \psi \big\|
&\leq \int \big\| G(k) \,a(k) \psi \big\|\, dk \\
&=  \int \big\| G(k)\,|k|^{-1/2} |k|^{1/2} \,
		a(k)\psi\big\|\, dk  \\
&\leq   \left( \int  |k|\, \big\| a(k) \psi \big\|^2
	dk \right)^{1/2}
\left(\int  |k|^{-1} \big\| G(k) \big\|^2
	dk \right)^{1/2} \\
&= \left(\int  |k|^{-1} \big\| G(k ) \big\|^2 dk
	\right)^{1/2}  \big\| H_{\rm f}^{1/2} \psi \big\|\,.
\end{align*}
To prove the second inequality we use the commutation relations
\begin{align*}
\big\| a^*(G) \psi \big\|^2
&= \langle a^*(G) \psi , \, a^*(G) \psi \rangle
=  \langle  \psi , \, a(G) \,a^*(G) \psi \rangle \\
&= \langle  \psi , \,\big( a^*(G) \,a(G)
	+ \int \big\| G(k) \big\|^2 dk \big) \psi \rangle \\
&\leq \left(\int  |k|^{-1} \big\| G(k) \big\|^2 dk \right)
	\big\| H_{\rm f}^{1/2} \psi \big\|^2
		+ \int \big\| G(k) \big\|^2 dk   \| \psi \|^2 \,.
\tag* \qedhere
\end{align*}
\end{proof}

The subsequent lemma states the well-known
{Pull-Through Formula}.
It can be proved using Eq. \eqref{eq:formalani}.
For a detailed proof we refer the reader to
\cite{BacFroSig98-2, HasHer11-1}.
%Detailed proofs are given in
%\cite{BacFroSig98-2,HasHer11-1}.
\begin{lemma} \label{lem:pullthrough}
Let $f : \R_+ \to \C$ be a bounded measurable
function. Then for all $k \in \R^3 \times \Z_2$
$$
%f(H_{\rm f}) \,a^*(k) = a^*(k) \,f(H_{\rm f} + \omega(k) ) \,, \qquad
a(k)\, f(H_{\rm f}) = f(H_{\rm f} + \omega(k)) \,a(k) \,.
$$
\end{lemma}

In order to define field operators that depend on the
 free field
 energy we consider measurable functions $w_{m,n}$ on
 $\R_+ \times {X}^{n+m}$ with values in
 the bounded linear operators of $\HH'$.
To such a function we associate the sesquilinear form
\begin{equation} \label{defofoform} 
q_{w_{m,n}}(\varphi,\psi)
:= \int \displaylimits_{X^{m+n}} \!\!\!\!
\big\langle a(k^{(m)}) \varphi ,
\,w_{m,n}(H_{\rm f}, K^{(m,n)})\, a(\tilde{k}^{(n)})
\psi \big\rangle \,
dK^{(m,n)} \,,
\end{equation}
defined  for all $\varphi$ and $\psi$ in $\HH' \otimes \FF$, for
which the integrand on the right hand side is integrable.
Here the r.h.s. of  \eqref{defofoform} is defined by means of an interated application of  \eqref{eq:formalani}. 
If the integral kernel  $w_{m,n}$ has sufficient  regularity  and decay, one can show that 
the sesquilinear  form   \eqref{defofoform}  defines a closed linear  operator which we denote by 
\begin{equation} \label{defofoform22} 
\int \displaylimits_{X^{m+n}} \!\!\!\!
a^*(k^{(m)}) 
\,w_{m,n}(H_{\rm f}, K^{(m,n)})\, a(\tilde{k}^{(n)})
dK^{(m,n)} \, . 
\end{equation}
In particular, in the case where $w_{m,n} \in \mathcal{W}_{m,n}$, cf. \eqref{eq:WmnSec5}, it follows from a simple application of  Lemma  \ref{kernelopestimate}, below,  that   \eqref{defofoform22}    is  bounded operator. 
To formulate the next lemma we denote by  $B([0,\infty);\mathcal{L}(\HH'))$  the Banach space    of all bounded measurable functions on $[0,\infty)$
with values in the bounded linear operators  of $\HH'$.

\begin{lemma} \label{kernelopestimate}
For measurable $w : X^{m+n} \to B([0,\infty); \mathcal{L}(\HH'))$, we
define
\begin{align*}
&  \big\| w_{m,n} \big\|_\sharp^2 \\
& :=
\int_{{X}^{m+n}}
\sup_{r \geq 0} \Big[ \big\|w_{m,n}(r,K^{(m,n)}) \big\|^2
\prod_{l=1}^m \big\{ r +  \sum_{j=1}^l | k_j | \big\}
 \prod_{\tilde{l}=1}^n
 \big\{ r +  \sum_{\tilde{j}=1}^{\tilde{l}} | k_{\tilde{j}}|  \big\}
 \Big]
 \frac{d K^{(m,n)}}{|K^{(m,n)}|} \,.
\end{align*}
Then for all  finitely many particle vecotors  $\varphi, \psi \in \HH' \otimes \FF$ 
\begin{equation} \label{eq:boundonq}
|\, q_{w_{m,n}}(\varphi,\psi) \,|
	\leq \| w_{m,n} \|_\sharp \,
		\| \varphi \|\, \| \psi \| \, .
\end{equation}
If $\| w_{m,n} \|_\sharp  < \infty$,
 the form $q_{w_{m,n}}$ determines uniquely a bounded
linear operator  ${h}_{w_{m,n}}$ such that
\begin{equation*}
q_{w_{m,n}}(\varphi,\psi ) = \langle \varphi,
\,  {h}_{w_{m,n}}  \psi \rangle \, ,
\end{equation*}
for all $\varphi, \psi$ in $\HH' \otimes \FF$ and 
$\| h_{w_{m,n}} \|_{}
	\leq \| w_{m,n} \|_\sharp$.
\end{lemma}
\begin{proof} Let us first introduce the number operator $N$, which is the linear operator on $\HH' \otimes \FF$ such that $N |_{\HH' \otimes \FF_n} = n$. 
It is straight forward to  verify that $N$ is  self-adjoint. 
First observe that $q_{w_{m,n}}(\varphi,\psi )  = q_{w_{m,n}}(1_{N \geq m} \varphi, 1_{N \geq n} \psi ) $ 
 We set
$P[k^{(n)}] := \prod_{l=1}^n
( H_{\rm f} +  \sum_{j=1}^l | k_j | )^{1/2}$
and insert $\one$'s into the left hand side
of Eq.~\eqref{eq:boundonq}
to obtain the trivial identity
\begin{align*}
 \big|\, q_{w_{m,n}}(\varphi,\psi) \,\big| =
\Bigg|& \int_{{X}^{m+n}}\Big\langle
 P[k^{(m)}] P[k^{(m)}]^{-1} |k^{(m)}|^{1/2}
 a(k^{(m)}) \varphi , w_{m,n}(H_{\rm f} ,K^{(m,n)})
 \\
&\qquad P[\tilde{k}^{(n)}] P[\tilde{k}^{(n)}]^{-1}
| \tilde{k}^{(n)}|^{1/2} a(\tilde{k}^{(n)}) \psi
\Big\rangle \frac{d K^{(m,n)}}{|K^{(m,n)}|^{1/2}} \Bigg| \, .
\end{align*}
The lemma now follows using the Cauchy-Schwarz
inequality and
the following   identity for $n \geq 1$ and
$\phi \in \ran 1_{N \geq n}$.  Relabeling the coordinates $(k_1,...,k_n) \mapsto (k_n,k_1,....,k_{n-1})$ and  using  \eqref{lem:pullthrough}
as well as 
  \eqref{eq:fieldenergy2} we find 
\begin{align} \label{eq:trivialA}
& \int_{{X}^n} \big| k^{(n)} \big| \,
\bigg\|\, \prod_{l=1}^n \left[ H_{\rm f} + \sum_{s=1}^l \omega( k_j ) 
\right]^{-1/2} a(k^{(n)}) \phi\, \bigg\|^2 \,d k^{(n)}\nonumber  \\
& = \int_{{X}^{n-1}} \int_X |k_n| \big| k^{(n-1)} \big| \,
\bigg\|\,      a(k_n)  H_{\rm f}^{-1/2}  \prod_{l=1}^{n-1} \left[ H_{\rm f} + \sum_{s=1}^l \omega( k_j ) 
\right]^{-1/2} a(k^{(n-1)}) \phi\, \bigg\|^2 \, dk_n d k^{(n-1)}  \nonumber \\
& =  \int_{{X}^{n-1}} \big| k^{(n-1)} \big| \,
\bigg\|\, \prod_{l=1}^{n-1} \left[ H_{\rm f} + \sum_{s=1}^l \omega( k_j ) 
\right]^{-1/2} a(k^{(n-1)}) \phi\, \bigg\|^2 \,d k^{(n-1)} \nonumber  \\
& \ \  \vdots \nonumber  \\ 
& = \big\| \phi \big\|^2  \, .
\end{align}
The  proof of Eq.~\eqref{eq:trivialA} is from 
\cite[Appendix A]{HasHer11-1}.
The last statement of the lemma follows from the first and
the Riesz lemma \cite[Theorem II.4]{ReeSim1}.
\end{proof}

\section{The smooth Feshbach-Schur map}\label{app:Feshbach}
In this section we review properties of the Feshbach-Schur map,
introduced in \cite{BCFS}. The presentation follows \cite{GriHas08}.
Let $\chi$ and $\overline{\chi}$ be commuting non-zero
bounded operators,
acting on a separable Hilbert space $\HH$ satisfying $\chi^2 + \overline{\chi}^2 = 1$.
\begin{definition}%\label{def:FeshbachPair}
A {\bf Feshbach pair} $(H,T)$ for $\chi$ is a pair
of closed operators with the same domain
\begin{equation*}
H, T : D(H) = D(T) \subset \HH \to \HH
\end{equation*}
such that $H, T, W := H - T$, and the operators
\begin{align*}
 &W_\chi := \chi W \chi \,,
	\qquad W_{\overline{\chi}} := \overline{\chi}
		W \overline{\chi} \,, \\
 &H_\chi := T  + W_\chi \,,
	\quad H_{\overline{\chi}} := T + W_{\overline{\chi}} \,,
\end{align*}
defined on $D(T)$ satisfy the following assumptions
\begin{enumerate}
	\item[(a)] $\chi T \subset T \chi $ and
	$\overline{\chi} T \subset T \overline{\chi}$,
	\item[(b)] $T, H_{\overline{\chi}}: D(T) \cap \ran \overline{\chi} \to \ran \overline{\chi}$
		are bijections with bounded inverse.
	\item[(c)] $\overline{\chi}
		H_{\overline{\chi}}^{-1} \overline{\chi}
			W \chi : D(T) \subset \HH \to \HH$
			is a bounded operator.
\end{enumerate}
\end{definition}
Given a Feshbach pair $(H,T)$ for $\chi$, the operator
\begin{align}\label{eq:feshbachmap}
	F_{\chi}(H,T) := H_{\chi}
		- \chi W \overline{\chi}
		H_{\overline{\chi}}^{-1} \overline{\chi} W \chi
\end{align}
on $D(T)$ is called {Feshbach operator}.
The mapping $(H,T) \mapsto F_\chi(H,T)$ is
called {Feshbach map}.
 We say that  an operator $A : D(A) \subset \HH \to \HH$
 is {bounded invertible}
 in a subspace $Y \subset \HH$, if
 $A : D(A) \cap Y \to Y$ is a bijection with
 bounded inverse.
Note that $Y$ does  not necessarily need to be closed.
If  $(H,T)$ is a Feshbach pair for $\chi$,  we define the following auxiliary operators
\begin{align}
	& Q_\chi := \chi  \label{eq:auxiliaryOp}
		- \overline{\chi}
		H_{\overline{\chi}}^{-1} \overline{\chi}
			W \chi\,, \\
	& Q_\chi^\# := \chi
		- \chi W \overline{\chi}
		H_{\overline{\chi}}^{-1} \overline{\chi}\,.
		\nonumber
\end{align}
 By conditions (a) and (c)
 $Q_\chi$ and $Q_\chi^\#$ are bounded operators on
 $D(T)$ and $Q_\chi$ leaves $D(T)$ invariant. 
\begin{theorem}[Theorem 1, \cite{GriHas08}]
\label{thm:isospectralFeshbach}
Let $(H,T)$ be a Feshbach pair for $\chi$ on a separable
Hilbert space $\HH$.
Then the following holds
\begin{itemize}
 \item [(a)] Let $Y$ be a subspace with
 $\ran \chi \subset Y \subset \HH$,
 \begin{align*} %\label{eq:propertiesofV}
	T : D(T) \cap Y \to Y, \quad \textrm{ and }
		\quad \overline{\chi}T^{-1}\overline{\chi}Y
			\subset Y.
 \end{align*}
	Then $H : D(H) \subset \HH \to \HH$ is
	bounded invertible if and only if
	$F_\chi(H,T) : D(T) \cap Y \to Y$ is bounded
	invertible in $Y$. Moreover,
	\begin{align*}
		H^{-1} &= Q_\chi F_\chi(H,T)^{-1}Q_\chi^\#
			+ \overline{\chi}H_{\overline{\chi}}^{-1}
			\overline{\chi}\,, \\
		F_\chi(H,T)^{-1} &= \chi H^{-1} \chi
			+ \overline{\chi} T^{-1} \overline{\chi}\,.
	\end{align*}
 \item [(b)] $\chi \ker H \subset \ker F_\chi(H,T)$
	and $Q_\chi \ker F_\chi(H,T) \subset \ker H$.
	The mappings
	\begin{align*}
		& \chi : \ker H \to \ker F_\chi(H,T)\,,
		%\label{eq:kernelofchi} \\
		& Q_\chi : \ker F_\chi(H,T) \to \ker H \,,
		%\label{eq:kernelofQchi}	
	\end{align*}
	are linear isomorphisms and inverse to each other.
\end{itemize}
\end{theorem}

\begin{lemma}[Lemma 3, \cite{GriHas08}]
\label{lem:FeshbachCriteria}
 Conditions (a), (b) and (c) on Feshbach pairs
 are satisfied if
 \begin{itemize}
  \item [(a')] $\chi T \subset T \chi \quad
	\textrm{ and } \quad \overline{\chi}T
		\subset T\overline{\chi}$\,,
  \item [(b')] $T$ is bounded invertible on
	$\ran \overline{\chi}$,
  \item [(c')] $\|T^{-1}\overline{\chi}W\overline{\chi}\|
	< 1 \quad \textrm{ and }
  	\quad \|\overline{\chi}WT^{-1}\overline{\chi}\| < 1$.
 \end{itemize}
\end{lemma}

\end{appendix}

\providecommand{\bysame}{\leavevmode\hbox to3em{\hrulefill}\thinspace}
\providecommand{\MR}{\relax\ifhmode\unskip\space\fi MR }
\providecommand{\MRhref}[2]{%
  \href{http://www.ams.org/mathscinet-getitem?mr=#1}{#2}
}
\providecommand{\href}[2]{#2}

%\bibliography{references}
%\bibliographystyle{amsplain}
\end{document}